\documentclass[aps,superscriptaddress,twocolumn,pra,10pt]{revtex4-2}

\usepackage{amsfonts}
\usepackage{amsmath}

\usepackage{xfrac}

\usepackage{amsthm}
\usepackage{amssymb}
\usepackage{graphicx}
\usepackage{mathtools}
\usepackage{fancyhdr}
\usepackage{algorithm}
\usepackage{algorithmic}
\usepackage{quantikz}

\usepackage{xcolor}
\usepackage[colorlinks=true,linkcolor=blue,citecolor=red,plainpages=false,pdfpagelabels]
{hyperref}
\providecommand{\U}[1]{\protect\rule{.1in}{.1in}}

\newtheorem{theorem}{Theorem}

\newtheorem*{algorithmmain*}{Algorithm}

\newtheorem{assumption}{Assumption}

\newtheorem{definition}[theorem]{Definition}

\newtheorem{lemma}[theorem]{Lemma}

\newtheorem{proposition}[theorem]{Proposition}
\newtheorem{remark}[theorem]{Remark}

\usepackage[letterpaper,margin=.6in,tmargin=.8in,bmargin=.9in]{geometry}
\fancyhfoffset[E,O]{0pt}

\newcommand{\opn}[1]{\operatorname{#1}}
\newcommand{\bs}[1]{\boldsymbol{#1}}

\allowdisplaybreaks

\begin{document}
\title[Quantum Boltzmann machine learning of ground states]{Quantum Boltzmann machine learning of ground-state energies}
\author{Dhrumil Patel}
\affiliation{Department of Computer Science, Cornell University, Ithaca, New York 14850, USA}
\author{Daniel Koch}
\affiliation{Air Force Research Lab, Information Directorate, Rome, New York 13441, USA}
\author{Saahil Patel}
\affiliation{Air Force Research Lab, Information Directorate, Rome, New York 13441, USA}
\author{Mark M. Wilde}
\affiliation{School of Electrical and Computer Engineering, Cornell University, Ithaca, New
York 14850, USA}
\keywords{quantum Boltzmann machines, ground-state energy, variational quantum eigensolver}

\begin{abstract}
Estimating the ground-state energy of Hamiltonians  is a fundamental task for which it is believed that quantum computers can be helpful. Several approaches have been proposed toward this goal, including algorithms based on quantum phase estimation and hybrid quantum-classical optimizers involving parameterized quantum circuits, the latter falling under the umbrella of the variational quantum eigensolver. Here, we analyze the performance of quantum Boltzmann machines for this task, which is a less explored ansatz based on parameterized thermal states and which is not known to suffer from the barren-plateau problem. We delineate a hybrid quantum-classical algorithm for this task and rigorously prove that it converges to an $\varepsilon$-approximate stationary point of the energy function optimized over parameter space, while using a number of parameterized-thermal-state samples that is polynomial in $\varepsilon^{-1}$, the number of parameters, and the norm 
of the Hamiltonian being optimized. Our algorithm estimates the gradient of the energy function efficiently by means of a quantum circuit construction that combines classical random sampling, Hamiltonian simulation, and the Hadamard test.
Additionally, supporting our main claims are calculations of the gradient and Hessian of the energy function, as well as an upper bound on the matrix elements of the latter that is used in the convergence analysis.
\end{abstract}
\date{\today}
\startpage{1}
\endpage{100}
\maketitle

\tableofcontents

\section{Introduction}

Calculating the ground-state energies of  Hamiltonians is one of the chief goals of quantum
physics~\cite{lieb2005stability}. This is typically the first step employed in
computing energetic properties of molecules and materials, and  it thus has
wide-ranging applications in materials science~\cite{Steinhauser2009},
condensed-matter physics~\cite{Continentino2021}, and quantum chemistry
\cite{DSL15}.

Stemming from the exponential growth of the state space as the number of
particles increases, calculating ground-state energies is generally a
difficult problem, and in fact it has been rigorously proven that the
worst-case complexity of doing so for physically relevant Hamiltonians is
computationally difficult in principle, even for a quantum
computer~\cite{Schuch2009,Childs2014,Huang2021}. In spite of this
complexity-theoretic barrier and due to the aforementioned applications, many
approaches have emerged for calculating ground-state energies on classical
computers. One of the oldest and most widely used approaches is based on the variational
principle~\cite{Gerjuoy1983}, in which one reduces the search space by
parameterizing a family of trial ground states and then searches over this
reduced space by means of gradient-descent like algorithms. This has
culminated in powerful methods like matrix product states
\cite{Fannes1992,Verstraete2006,PerezGarcia2007}, which perform well in practice.

In another direction, Ref.~\cite{Abrams1999} has argued that quantum computers could be
effective at calculating ground-state energies, due to their
ability to simulate quantum mechanical processes faithfully and with reduced
overhead, in principle, when compared to classical algorithms. Building upon
\cite{Abrams1999}, one of the first approaches proposed for doing so involves
employing the quantum phase estimation algorithm for small molecules
\cite{ADLH05}. More recently, other phase-estimation-based algorithms for
ground-state energy estimation have been proposed and
analyzed~\cite{Lin2022,Dong2022,Wan2022,Ding2023,Wang2023quantumalgorithm,wang2023fastergroundstateenergy}, with the goal being to reduce the resources required, in a way that is more
amenable to \textquotedblleft early fault-tolerant\textquotedblright\ quantum
processors. All of these approaches assume the availability of an initial
trial state that has non-trivial overlap with the true ground state.

Due to the approach of~\cite{ADLH05} requiring quantum circuits of large depth
(i.e., a long sequence of consecutive quantum logic gates), researchers
subsequently proposed the variational quantum eigensolver (VQE) as another
approach for the ground-state energy estimation problem~\cite{Peruzzo2014}.
The VQE approach employs parameterized quantum circuits (PQCs) of shorter
depth and involves a hybrid interaction between such shorter-depth quantum
circuits and a classical optimizer. Interestingly, the VQE\ approach provides
a quantum computational implementation of the aforementioned variational
method. While the VQE approach at first seemed promising, later research 
pointed out a number of bottlenecks associated with it~\cite{Tilly2022}, which
will likely preclude VQE from achieving practical quantum
advantage in the near term. One of the primary bottlenecks is the barren-plateau problem~\cite{McClean_2018,marrero2021entanglement,arrasmith2022equivalence,holmes2022connecting,Fontana2024,Ragone2024}, in which the
landscape of the objective function becomes extremely flat, so that a costly
(exponential)\ number of measurements is required to determine which
direction the optimizer should proceed to next, at any given iteration of the algorithm.

While the VQE\ approach is based on employing parameterized quantum circuits
(a particular ansatz for generating trial states), an alternate ansatz
involves using quantum Boltzmann machines
(QBMs)~\cite{Amin2018,Benedetti2017,Kieferova2017}, and this is
the approach that we pursue and analyze here for ground-state energy
estimation. Indeed, in the QBM\ approach to ground-state
energy estimation, one substitutes parameterized quantum circuits with
parameterized thermal states of a given Hamiltonian and performs the search
over parameterized thermal states. Furthermore, the QBM\ approach appears to be viable,
due to significant recent progress on the problem of preparing thermal
states on quantum
computers~\cite{chen2023quantumthermalstatepreparation,chen2023efficient,bergamaschi2024,chen2024randomizedmethodsimulatinglindblad,rajakumar2024gibbssamplinggivesquantum,rouze2024efficient,bakshi2024hightemperature,ding2024polynomial}, in spite of known worst-case complexity-theoretic barriers~\cite{Bravyi2022}. Hitherto, QBMs have been analyzed in the context of Hamiltonian learning~\cite{Anshu2021,GarciaPintos2024}  and generative
modeling~\cite{Coopmans2024}, but, to the best of our knowledge,
they have not been considered yet for ground-state energy estimation. Another
significant and promising aspect of QBMs is that there is evidence that they do not suffer from the barren-plateau
problem in certain contexts~\cite{Coopmans2024}. In this context, we should also note that~\cite{OrtizMarrero2021} proved that QBMs with hidden units can suffer from the barren plateau problem, assuming a particular approach to generating parameterized thermal states randomly; however, this statement is not applicable to QBMs with visible units only, i.e., the model that we employ here (see~\cite{OrtizMarrero2021} for definitions of hidden and visible units in QBMs). Indeed it was proven in~\cite{Coopmans2024} that QBMs with visible units do not suffer from the barren plateau problem when used in the context of generative modeling, and much more so, they provably converge in this setting, using a number of thermal state samples polynomial in the number of qubits.  

\section{Summary of main results}

The main finding of our paper is a rigorous
mathematical proof that the QBM learning approach to approximating
ground-state energies is \textit{sample efficient}, in the sense that the number of samples of parameterized thermal states used by our algorithm is polynomial in several quantities of interest, the latter to be clarified later. In doing so, we also overcome a key obstacle to efficient training of QBMs, discussed in further detail in what follows.

In more detail, we adopt a hybrid quantum-classical approach,
similar to what is used in VQE, but we instead replace PQCs with QBMs, as
mentioned above. Let $H$ denote the Hamiltonian of interest, which we assume
can be efficiently measured on a quantum computer. We suppose that this Hamiltonian acts on $n$ qubits, but let us note that all of the analysis and algorithms that follow apply also to qudit systems ($d$-dimensional systems).  The Hamiltonian $H$ can be efficiently measured  when
\begin{equation}
H=\sum_{k=1}^{K}\alpha_{k}H_{k},
\label{eq:local-Ham-for-GS}
\end{equation}
where, for all $k\in\left[  K\right]  $, the coefficient $\alpha_{k}
\in\mathbb{R}$ and $H_{k}$ is a local Hamiltonian acting on a constant number
of particles. Without loss of generality, we assume that $\left\Vert
H_{k}\right\Vert \leq1$ by absorbing the norm of $H_{k}$ into $\alpha_{k}$, so
that
\begin{equation}
\left\Vert H\right\Vert \leq\sum_{k=1}^K\left\vert \alpha_{k}\right\vert
\left\Vert H_{k}\right\Vert \leq\sum_{k=1}^K \left\vert \alpha_{k}\right\vert
\eqqcolon \left\Vert \alpha\right\Vert _{1},
\label{eq:ham-bound-and-norm}
\end{equation}
and we also assume that $\alpha_{k}>0$ for all $k\in\left[  K\right]  $,
because any negative sign for $\alpha_{k}$ can be absorbed into $H_{k}$. Let
\begin{equation}
G(\theta)\coloneqq\sum_{j=1}^{J}\theta_{j}G_{j}
\end{equation}
be a trial Hamiltonian, with $J \in \mathbb{N}$, each $\theta_{j}\in\mathbb{R}$ a parameter,
\begin{equation}
\theta\coloneqq(\theta_{1},\ldots,\theta_{J}),    
\end{equation}
 and each $G_{j}$ a
Hamiltonian that acts locally on a constant-sized set of qubits. We assume that we have some knowledge of the problem at hand and that therefore we can appropriately  choose a subspace of parameters such that $J = \operatorname{poly}(n)$. Furthermore, similar to~\cite{Anshu2021,Coopmans2024}, we assume that samples
of the thermal state 
\begin{equation}
\rho(\theta)\coloneqq\frac{1}{Z(\theta)}e^{-G(\theta)},
\end{equation}
where
\begin{equation}
Z(\theta)\coloneqq\operatorname{Tr}[e^{-G(\theta)}],    
\end{equation}
are available,
for every possible choice of $\theta \in \mathbb{R}^J$. As such, the QBM model that we employ here has only visible units and no hidden units, using the terminology of \cite{Amin2018,Kieferova2017}. 

The following inequality is a basic
consequence of the variational principle:
\begin{align}
\inf_{\rho\in\mathcal{D}}\operatorname{Tr}[H\rho] & \leq\inf_{\theta\in
\mathbb{R}^{J}} f(\theta) ,\label{eq:VP} \\
\text{ where } f(\theta)  & \coloneqq \operatorname{Tr}[H\rho(\theta)],
\label{eq:obj-f-def}
\end{align}
and
$\mathcal{D}$ represents the set of all possible quantum states acting
on the same Hilbert space on which $H$ acts (i.e., $\mathcal{D}$ is the set of
all such density operators, which are unit trace, positive semi-definite
operators). The inequality in~\eqref{eq:VP} indicates that the true
ground-state energy is bounded from above by the minimal energy of the
Hamiltonian $H$ over every possible trial state $\rho(\theta)$.

With these
notions in place, we can state our main claim:\ finding an $\varepsilon$-approximate stationary point of $f(\theta)$  is
sample efficient, in the sense that our algorithm uses a number of parameterized-thermal-state samples that is polynomial in $\varepsilon^{-1}$,  $J$, and $\left\Vert \alpha\right\Vert _{1}$. Since the function $\theta\mapsto f(\theta)$ is generally non-convex, finding an $\varepsilon
$-stationary point (local minimum), rather than a global minimum, is
essentially the best that one can hope for when using this approach. Indeed, we argue in Appendix~\ref{sec:prob_setup} that even the following basic instance of $f(\theta)$ is non-convex:
\begin{equation}
    (\theta_1, \theta_2) \mapsto \operatorname{Tr}\!\left[\sigma_Y \frac{e^{-G(\theta_1,\theta_2)}}{\operatorname{Tr}[e^{-G(\theta_1,\theta_2)}]}\right]
\end{equation}
 with $G(\theta_1, \theta_2) =\theta_1 \sigma_X + \theta_2 \sigma_Y$.

 One of the essential steps in optimizing $f(\theta)$ in~\eqref{eq:obj-f-def} is to determine its gradient. This is needed in any gradient-descent
like algorithm, in order to determine which step to take next in an iterative
search. An analytical form for the gradient $\nabla_{\theta
}f(\theta)$ is  based on an analytical form for $\nabla_{\theta}\rho(\theta)$, the latter of which follows from the developments in~\cite{Hastings2007}, \cite[Appendix~B]{Anshu2021}, and \cite[Lemma~5]{Coopmans2024} (see also \cite[Section~III-C]{Kim2012} and \cite[Section~IV-A]{Kato2019}). In more detail, it follows from these works that
\begin{align}
 \partial_{j}f(\theta)
&  = -\frac{1}{2}\left\langle\left\{  H,\Phi_{\theta}(G_{j})\right\}
\right\rangle+\left\langle H\right\rangle \left\langle G_{j}\right\rangle
\label{eq-mt:gradient-alt}\\
&  =-\frac{1}{2}\!\left\langle\left\{  H-\left\langle H\right\rangle
,\widetilde{G_{j}} -\left\langle \widetilde{G_{j}} \right\rangle \right\}\right\rangle\label{eq-mt:gradient},
\end{align}
where $\partial_{j}\equiv\frac{\partial}{\partial\theta_{j}}$,
\begin{equation}
\{A,B\}\coloneqq AB+BA    
\end{equation}
denotes the anticommutator of operators $A$ and $B$, 
\begin{align}
\widetilde{G_{j}} & \equiv \Phi_{\theta}(G_{j}) ,\\
\left\langle
C\right\rangle & \coloneqq\operatorname{Tr}[C\rho(\theta)] ,   
\end{align}
for a Hermitian
operator$~C$, and $\Phi_{\theta}$ is the following quantum channel:
\begin{align}
\Phi_{\theta}(X) & \coloneqq\int_{\mathbb{R}}dt\ p(t)\ e^{-iG(\theta
)t}Xe^{iG(\theta)t} \ ,
\label{eq:q-channel-phi-theta} \\
\text{ with }
p(t) & \coloneqq\frac{2}{\pi}\ln\left\vert \coth(\pi t/2)\right\vert
\label{eq-mt:prob-dens}
\end{align}
a probability density function on $t\in\mathbb{R}$ (we refer to $p(t)$ as the ``high-peak-tent'' probability density function, due to the form of its graph when plotted). We also used that
$\left\langle \Phi_{\theta}(G_{j})\right\rangle =\left\langle G_{j}
\right\rangle $, which follows because $\Phi_{\theta}(\rho(\theta
))=\rho(\theta)$. Prior work~\cite{Hastings2007,Anshu2021,Coopmans2024} refers to the map $\Phi_{\theta}$ as the quantum belief propagation superoperator. Here we observe that it is in fact a quantum channel (completely positive, trace-preserving map), due to the fact that $p(t)$ is a probability density function. Note that this is remarked upon (without proof) in \cite[Footnote~32]{Kim2012}. 

Supporting our main finding are various contributions of our work, which we
list now.
Here we prove that the gradient
$\nabla_{\theta}f(\theta)$ is Lipschitz continuous,
which is needed to make rigorous claims about the convergence of the
stochastic gradient descent (SGD) algorithm. Moreover (and essential to our overall algorithm), we demonstrate how the
gradient $\nabla_{\theta}f(\theta)$ can be efficiently
estimated on a quantum computer, which is a consequence of the formula in~\eqref{eq-mt:gradient-alt} and the observation that $p(t)$ is a probability density function. That is, we provide an efficient quantum
algorithm, called the \textit{quantum Boltzmann gradient estimator}, that computes an unbiased estimate of~$\nabla_{\theta
}f(\theta)$. By doing so, we have thus overcome a key obstacle in QBM learning going back to \cite[Section~II]{Amin2018}, in which it was previously thought that estimating the gradient could not be done efficiently (see also~\cite{Kieferova2017,wiebe2019generative,anschuetz2019realizing,Kappen_2020,Zoufal2021} for similar previous discussions on the perceived difficulty of training QBMs by directly estimating the gradient). We discuss this in more detail in Appendix~\ref{sec:resolving-open-problem}. Having an unbiased estimator is also
helpful in analyzing the convergence of SGD. With these analytical results in
place, we then invoke known results~\cite[Corollary~1]{Khaled2020} on the
convergence of SGD to conclude that the sample complexity of our algorithm for
finding an $\varepsilon$-stationary point of $\theta\mapsto f(\theta)$ is polynomial in $\varepsilon^{-1}$, $J$, and $\left \| \alpha \right \|_1$  (recall that
sample complexity here is the number of parameterized thermal states
 needed). This summarizes the main contributions of our paper.

Our results reported here can be contrasted with an analytical study of
VQE and PQCs~\cite{Harrow2021}. Indeed, therein, the authors studied the
convergence of VQE\ and PQCs when performing analytic measurements of the
gradient of the cost function (a first-order method), as compared to
a gradient-free method of measuring the cost function directly (a
zeroth-order method). They found that, for certain Hamiltonians,
VQE\ algorithms employing an analytic gradient measurement (first-order
methods)\ are faster than zeroth-order methods. However, their analysis was
restricted to non-interacting Hamiltonians, for which one can actually
calculate the ground-state energy by hand; regardless, the authors suggested
that their analytical finding should be indicative of what one might find for
more complex Hamiltonians. In contrast, our analysis applies to all
Hamiltonians that are efficiently measurable on quantum computers, thus
encompassing a significantly wider class of Hamiltonians.

In what follows, we provide further details of our results, while the
appendices give complete proofs of all of our claims.

\section{Calculation of gradient and Hessian}

Let us first briefly review the proof of the equality
in~\eqref{eq-mt:gradient}, which begins with the following equality:
\begin{equation}
\partial_{j}\rho(\theta)=-\frac{1}{2}\left\{  \Phi_{\theta}(G_{j}),\rho
(\theta)\right\}  +\rho(\theta)\left\langle G_{j}\right\rangle
,\label{eq:derivative-thermal-state}
\end{equation}
along with some further algebraic manipulations.\ To
see~\eqref{eq:derivative-thermal-state}, consider that
\begin{align}
\partial_{j}\rho(\theta) &  =\partial_{j}\left[  \frac{1}{Z(\theta
)}e^{-G(\theta)}\right]  \\
&  =-\frac{1}{Z(\theta)^{2}}\left[  \partial_{j}Z(\theta)\right]
e^{-G(\theta)}+\frac{1}{Z(\theta)}\partial_{j}e^{-G(\theta)}\\
&  =-\frac{\rho(\theta)}{Z(\theta)}\left[  \partial_{j}\operatorname{Tr}
[e^{-G(\theta)}]\right]  +\frac{1}{Z(\theta)}\partial_{j}e^{-G(\theta
)}.\label{eq-mt:gradient-calc-mid}
\end{align}
Now recall~\cite[Eq.~(9)]{Hastings2007} (however, for the precise statement that we use, see \cite[Proposition~20]{Anshu2021} and \cite[Lemma~5]{Coopmans2024}):
\begin{equation}
\partial_{j}e^{-G(\theta)}=-\frac{1}{2}\left\{  \Phi_{\theta}(G_{j}
),e^{-G(\theta)}\right\}  .\label{eq-mt-gradient-unnorm-state}
\end{equation}
It was proven in~\cite[Appendix~B]{Anshu2021} that $\frac
{\tanh(\omega/2)}{\omega/2}$ is the Fourier transform of $p(t)$ in~\eqref{eq-mt:prob-dens}, i.e.,
\begin{equation}
    \frac
{\tanh(\omega/2)}{\omega/2}=\int_{-\infty}^{\infty}
dt\ p(t) e^{-i\omega t},
\end{equation}
 and that it has the
explicit form given in~\eqref{eq-mt:prob-dens}. The latter implies that $p(t)$ is a probability density function and thus that $\Phi_{\theta
}$ is a quantum channel. As mentioned above, this observation is paramount later on for our
quantum circuit construction that provides an unbiased estimate of
$\partial_{j}f(\theta)$. Now plugging~\eqref{eq-mt-gradient-unnorm-state} into~\eqref{eq-mt:gradient-calc-mid} and
simplifying, we conclude~\eqref{eq:derivative-thermal-state}. Plugging~\eqref{eq:derivative-thermal-state} into $\partial_{j} f(\theta)$ and simplifying, we arrive at~\eqref{eq-mt:gradient-alt}.
Further algebraic manipulations lead to~\eqref{eq-mt:gradient}. See Appendix~\ref{sec:grad}.

We also compute the Hessian of $f(\theta)$. Due to the
length of the expression, we only include it in Appendix~\ref{sec:hessian}, along
with its derivation. While this quantity can also be efficiently estimated on a quantum
computer (as argued in the supplementary material) and incorporated into a
Newton method search (i.e., an extension of gradient descent that incorporates
second-derivative information), we mainly use it to determine a Lipschitz
constant for the gradient $\nabla_{\theta}f(\theta)$,
which in turn implies rigorous statements about the convergence of SGD, as previously mentioned.

\section{Lipschitz constant for gradient}

By bounding the matrix elements of the Hessian of the
objective function~$f(\theta)$, we can use it to
establish a Lipschitz constant for its gradient $\nabla_{\theta}
f(\theta)$. Indeed, recalling that
\begin{equation}
    \left \| A\right\| \coloneqq \sup_{ |\psi \rangle  : \left \| \psi\rangle \right \| = 1} \left \| A |\psi\rangle \right \|,
\end{equation}
 we find that
\begin{equation}
\left\vert \partial_{j}\partial_{k}f(\theta)\right\vert
\leq8\left\Vert H\right\Vert \left\Vert G_{j}\right\Vert \left\Vert
G_{k}\right\Vert ,
\end{equation}
which we can substitute into \cite[Lemma~8]{Patel2024}, in order to conclude the following
Lipschitz constant for the gradient~$\nabla_{\theta}
f(\theta)$:
\begin{equation}
8 J\left\Vert H\right\Vert\max\left\{
\left\Vert G_{j}\right\Vert^2 \right\}  _{j=1}^{J}.    
\end{equation}
If the Hamiltonian $H$ is of the form in~\eqref{eq:local-Ham-for-GS}, then a
Lipschitz constant $\ell$ for the gradient $\nabla_{\theta}\operatorname{Tr}[H\rho
(\theta)]$ is as follows:
\begin{equation}
\ell\coloneqq 8 J \left\Vert \alpha\right\Vert _{1}
\max\left\{  \left\Vert G_{j}\right\Vert^2  \right\}  _{j=1}^{J}.
\label{eq:lips-smooth}
\end{equation}
As we will see, this Lipschitz constant for the gradient~$\nabla_{\theta
}f(\theta)$ implies that the sample complexity of
SGD\ is polynomial in $J$ and $\left\Vert \alpha\right\Vert _{1}$. Let us finally note that $\ell$ is also called a smoothness parameter for $f(\theta)$.

\section{Quantum algorithm for gradient estimation}

In the $m$th step of
the SGD\ algorithm (reviewed in Appendix~\ref{sec:sgd}), one updates the parameter vector $\theta$
according to the following rule:
\begin{equation}
\theta_{m+1}=\theta_{m}-\eta\overline{g}(\theta_{m}),
\end{equation}
where $\eta>0$ is the learning rate and $\overline{g}(\theta_{m})$ is a
stochastic gradient evaluated at $\theta_{m}$. The stochastic gradient~$\overline{g}(\theta)$ should
be unbiased, in the sense that $\mathbb{E}[\overline{g}(\theta)]=\nabla
_{\theta}f(\theta)$ for all $\theta\in\mathbb{R}^{J}$,
where the expectation is over all the randomness associated with the generation of~$\overline{g}(\theta)$. As such, it is necessary to have a method for
generating the stochastic gradient $\overline{g}(\theta)$, and for this
purpose, we prescribe a quantum algorithm based on~\eqref{eq-mt:gradient-alt}, which we call the \textit{quantum Boltzmann gradient estimator}.

Consider that~\eqref{eq-mt:gradient-alt} is a linear combination of two terms.
As such, we delineate one procedure that estimates $\frac{1}{2}\left\langle\left\{  H,\Phi_{\theta}(G_{j})\right\}\right \rangle $ and another that estimates $ \left\langle H\right\rangle \left\langle G_{j}\right\rangle$.
The second term $\left\langle H\right\rangle \left\langle G_{j}\right\rangle $
is simpler:\ since it can be written as
\begin{equation}
\left\langle H\right\rangle \left\langle G_{j}\right\rangle =\operatorname{Tr}
[\left(  H\otimes G_{j}\right)  \left(  \rho(\theta)\otimes\rho(\theta
)\right)  ],
\label{eq:simple-term-to-estimate}
\end{equation}
a procedure for estimating it is to generate the state $\rho
(\theta)\otimes\rho(\theta)$ and then measure the observable $H\otimes G_{j}$
on these two copies. Through repetition, the estimate of $\left\langle
H\right\rangle \left\langle G_{j}\right\rangle $ can be made as precise as
desired. This procedure is described in detail in Appendix~\ref{app:2nd-term-estimator}
as Algorithm~\ref{algo:est_second_term}, the result of which is that $O(\left\Vert \alpha\right\Vert
_{1}^{2}\varepsilon^{-2}\ln\delta^{-1})$ samples of $\rho(\theta)$ are
required to have an accuracy of $\varepsilon>0$ with a failure probability of
$\delta\in\left(  0,1\right)  $, when $H$ is of the form in~\eqref{eq:local-Ham-for-GS}.

Our quantum algorithm for estimating the first term in~\eqref{eq-mt:gradient-alt} is more intricate. Under the assumption that $H$
has the form in~\eqref{eq:local-Ham-for-GS}, it follows by direct substitution of \eqref{eq:local-Ham-for-GS} and \eqref{eq:q-channel-phi-theta}, as well as the fact that $p(t)$ is an even function, that
\begin{multline}
\frac{1}{2}\left\langle\left\{  H,\Phi_{\theta}(G_{j})\right\}
\right\rangle = \frac{1}{2}\operatorname{Tr}[\left\{  H,\Phi_{\theta}(G_{j})\right\}
\rho(\theta)]\\
=
\sum_{k}\alpha_{k}\int_{-\infty}^{\infty}dt\ p(t)\ \operatorname{Re}
[\operatorname{Tr}[U_{j,k}(\theta,t)\rho(\theta)]],
\end{multline}
where
\begin{equation}
    U_{j,k}(\theta,t)\coloneqq H_{k}e^{-iG(\theta)t}G_{j}e^{iG(\theta)t} .
\end{equation}
If  $H_{k}$ and $G_{j}$ are also unitaries (which occurs in a rather
standard case that they are Pauli strings, i.e., tensor products of Pauli
matrices), then we can estimate the first term in~\eqref{eq-mt:gradient-alt}
by a combination of classical sampling, Hamiltonian simulation~\cite{lloyd1996universal,childs2018toward}, and the
Hadamard test~\cite{Cleve1998}. This is the key insight behind the quantum Boltzmann gradient estimator. Indeed, the basic idea is to sample $k$ with probability
$\alpha_{k}/\left\Vert \alpha\right\Vert _{1}$ and $t$ with probability
density $p(t)$. Based on these choices, we then execute the quantum circuit in
Figure~\ref{fig:gradient-estimator}, which outputs a binary random
variable~$Y$, that has a realization $y=0$ occurring with probability
\begin{equation}
\frac{1}
{2}\left(  1+\operatorname{Re}[\operatorname{Tr}[U_{j,k}(\theta,t)\rho
(\theta)]]\right)       
\end{equation}
and $y=1$ occurring with probability
\begin{equation}
\frac{1}{2}\left(
1-\operatorname{Re}[\operatorname{Tr}[U_{j,k}(\theta,t)\rho(\theta)]]\right)
.    
\end{equation}
Thus, the expectation of a random variable $Z=\left(  -1\right)  ^{Y+1}$,
conditioned on $k$ and $t$, is equal to $-\operatorname{Re}[\operatorname{Tr}
[U_{j,k}(\theta,t)\rho(\theta)]]$ and including the further averaging over $k$
and $t$, the expectation is equal to
\begin{equation}
-\frac{1}{2\left\Vert \alpha\right\Vert
_{1}}\operatorname{Tr}[\left\{  H,\Phi_{\theta}(G_{j})\right\}  \rho(\theta
)].    
\end{equation}
As such, we can sample $k$, $t$, and $Z$ in this way, and averaging the
outcomes and scaling by $\left\Vert \alpha\right\Vert _{1}$ gives an unbiased
estimate of the first term in~\eqref{eq-mt:gradient-alt}. This procedure is
described in detail in Appendix~\ref{app:1st-term-estimator} as Algorithm~\ref{algo:est_first_term}, the result
of which is that $O(\left\Vert \alpha\right\Vert _{1}^{2}\varepsilon^{-2}
\ln\delta^{-1})$ samples of $\rho(\theta)$ are required to have an accuracy of
$\varepsilon>0$ with a failure probability of $\delta\in\left(  0,1\right)  $,
when $H$ is of the form in~\eqref{eq:local-Ham-for-GS}. We finally note that this construction can straightforwardly be generalized beyond the case of $H_k$ and $G_j$ being Pauli strings, if they instead are block encoded into unitary circuits~\cite{Low2019hamiltonian,Gilyen2019}.

Thus, through this
combination of classical random sampling and quantum circuitry, we can produce an
unbiased estimate of the first term in~\eqref{eq-mt:gradient-alt}. Adding this estimate and the one from the paragraph surrounding~\eqref{eq:simple-term-to-estimate} then leads to the quantum Boltzmann gradient estimator, which realizes an unbiased estimate of~\eqref{eq-mt:gradient-alt}.

\begin{figure}
[ptb]
\begin{center}
\includegraphics[
width=3.3598in
]
{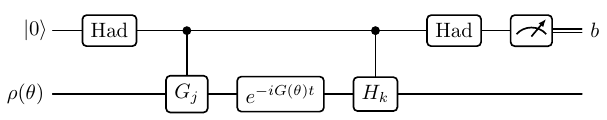}
\caption{Quantum circuit that plays a role in realizing an unbiased estimate of $-\frac{1}{2}\left\langle \left\{  H,\Phi_{\theta}(G_{j})\right\}\right\rangle
$. The Boltzmann gradient estimator combines this estimate with an unbiased estimate of $\left\langle H\right\rangle \left\langle G_{j}\right\rangle $, to realize an unbiased estimate of the gradient $\nabla_\theta f(\theta)$ in 
\eqref{eq-mt:gradient-alt}.}
\label{fig:gradient-estimator}
\end{center}
\end{figure}

\section{Ground-state energy estimation algorithm and its performance}

Finally, we assemble everything presented so far and describe our algorithm for ground-state energy estimation, along with guarantees on its performance. 

\begin{algorithmmain*}[QBM-GSE]
Fix $\varepsilon \in (0,1)$. The algorithm for converging to an $\varepsilon$-stationary point of $f(\theta)$ in~\eqref{eq:obj-f-def} consists of the following steps:
\begin{enumerate}
\item Initialize $\theta_0 \in \mathbb{R}^J$. Set the learning rate $\eta $ and the number $M  $ of iterations as follows:
\begin{equation}
\eta  =  \frac{1}{ \ell }, \qquad 
M \geq  \left\lceil\frac{12 \Delta \ell}{\varepsilon^2}\right\rceil,
\end{equation}
where the smoothness parameter $\ell$ is defined in~\eqref{eq:lips-smooth} and $
\Delta  \coloneqq f(\theta_0) - \inf_{\theta \in \mathbb{R}^J} f(\theta)$. Set $m=0$.
\item Execute Algorithms \ref{algo:est_second_term} and \ref{algo:est_first_term} (detailed in Appendix~\ref{sec:qbge}) to calculate $\overline{g}(\theta_m)$, which is a stochastic gradient satisfying $\mathbb{E}[\overline{g}(\theta)] = \nabla_{\theta
}f(\theta)$.
\item Apply the update: $\theta_{m+1} = \theta_m - \eta \overline{g}(\theta_m)$. Set $m = m+1$.
\item Repeat steps 2-3 $M-1$ more times and output an estimate of $\operatorname{Tr}[H \rho(\theta_M)]$ (the latter obtained by measuring $H$ with respect to the state $\rho(\theta_M)$, i.e., through sampling and averaging).
\end{enumerate}
\end{algorithmmain*}

By invoking~\cite[Corollary~1]{Khaled2020} and further analysis from Appendix~\ref{sec:sample_comp}, we conclude the following convergence guarantee for the QBM-GSE algorithm:
\begin{theorem}
\label{thm:main}
The QBM-GSE algorithm converges to an $\varepsilon$-stationary point of $f(\theta)$ in~\eqref{eq:obj-f-def}, i.e., such that
\begin{equation}
\min_{m \in \{1,\ldots,M\}} \mathbb{E}\left\| \nabla_\theta f(\theta_m)\right\| \leq \varepsilon.
\label{eq:SGD-main}
\end{equation}
\end{theorem}

Note that $\theta_m $ in Theorem~\ref{thm:main} is a random variable, as given in the QBM-GSE algorithm, and  the expectation in~\eqref{eq:SGD-main} is with respect to the randomness associated with generating~$\theta_m$.

The above statement in Theorem~\ref{thm:main} implies that the number $M$ of steps that the QBM-GSE algorithm requires to converge to an $\varepsilon$-stationary point of $f(\theta)$ is polynomial in $\varepsilon^{-1}$, $J$, and $\left \| \alpha \right\|_1$. The number $M$ of steps is then directly related to the sample complexity of the algorithm. By combining Theorem~\ref{thm:main} and the Hoeffding bound, the QBM-GSE algorithm uses at least the following number of parameterized-thermal-state samples:
\begin{equation}
2J\left\lceil\frac{12  \ell \Delta}{\varepsilon^2}\right\rceil \left \lceil \frac{8J\left \Vert \alpha \right \Vert_1^2 }{\varepsilon^2} \ln\!\left(\frac{16J\left\Vert \alpha \right \Vert_1^2}{\varepsilon^2}\right)\right\rceil ,
\end{equation}
where  $\ell$ is defined in~\eqref{eq:lips-smooth}.
As such, the sample complexity of our algorithm is polynomial in $\varepsilon^{-1}$, $J$, and $\left \| \alpha \right\|_1$, as claimed. Thus, if $J$ and $\left \| \alpha \right\|_1$ are polynomial in $n$ (the number of qubits), then  the sample complexity of the QBM-GSE algorithm is also polynomial in~$n$.

\section{Discussion}

Our algorithm is most pertinent for the situation in which low-temperature thermal states of $H$ are computationally difficult to generate but thermal states of $G(\theta)$ are not. In such a situation, one can use our algorithm as a means for approximating the ground-state energy of $H$. This is similar to the scenario considered in VQE: there the VQE algorithm is also most applicable when the ground state of $H$ is computationally difficult to generate but states realized by PQCs are not, which is indeed the case for short-depth PQCs. However, as emphasized previously, in contrast to VQE, QBMs are not known to suffer from the barren-plateau problem, and there is evidence that they do not in certain contexts~\cite{Coopmans2024}. As such, they appear to offer a more promising route for ground-state energy estimation.

\section{Conclusion and outlook}

In this paper, we have analyzed quantum Boltzmann machine learning of ground-state energies. Our main result is an algorithm that converges to an $\varepsilon$-approximate stationary point of $f(\theta)$ in~\eqref{eq:obj-f-def}, along with a rigorous claim about its convergence. Namely, for Hamiltonians of the form in~\eqref{eq:local-Ham-for-GS}, we have proven that the sample complexity of the algorithm is polynomial in $\varepsilon^{-1}$, $J$, and $\left \| \alpha \right\|_1$, where $\left \| \alpha \right\|_1$ is defined in~\eqref{eq:ham-bound-and-norm}. Our algorithm welds together conventional stochastic gradient descent and a novel quantum circuit for estimating the gradient $\nabla_{\theta}f(\theta)$, the latter being the core component of the quantum Boltzmann gradient estimator. Supporting our main claims are calculations of the gradient, Hessian, and smoothness parameter of the objective function $f(\theta)$, along with various observations about the gradient that lead to our quantum algorithm for estimating it. 

We believe that our results have far-reaching consequences for QBM learning. Indeed, given that the quantum Boltzmann gradient estimator efficiently estimates~\eqref{eq-mt:gradient-alt}, this approach now opens the door to using QBMs for efficient learning and optimization in a much wider variety of contexts. For example, one can substitute QBMs for PQCs in recent works on semi-definite programming and constrained Hamiltonian optimization~\cite{Patel2024,chen2023qslackslackvariableapproachvariational} and entropy estimation~\cite{Goldfeld2024} and analyze the sample complexity and convergence when doing so. We suspect that claims similar to those made here, regarding polynomial sample complexity for convergence to an $\varepsilon$-stationary point, can be made in these contexts; however, it remains a topic for future investigation. 

Going forward from here, it is a pressing open question to determine whether the landscape of the objective function $f(\theta)$ suffers from the barren-plateau problem. In a different optimization problem involving QBMs~\cite{Coopmans2024}, evidence was given supporting the conclusion that barren plateaus do not occur there. If the landscape of the objective function $f(\theta)$ does not suffer from the barren-plateau problem, this, along with our algorithm and further progress on thermal state preparation, would imply that QBMs are a viable path toward ground-state energy estimation and other learning and optimization problems.

\let\oldaddcontentsline\addcontentsline
\renewcommand{\addcontentsline}[3]{}

\begin{acknowledgments}
We thank Paul Alsing, Nana Liu, and Soorya Rethinasamy for helpful discussions. 
DP and MMW acknowledge support from
AFRL under agreement no.~FA8750-23-2-0031.

This material is based on research
sponsored by Air Force Research Laboratory under agreement number
FA8750-23-2-0031. The U.S.~Government is authorized to reproduce and
distribute reprints for Governmental purposes notwithstanding any copyright
notation thereon. The views and conclusions contained herein are those of the
authors and should not be interpreted as necessarily representing the official
policies or endorsements, either expressed or implied, of Air Force Research
Laboratory or the U.S.~Government.
\end{acknowledgments}

\section*{Author Contributions}

The following describes the different contributions of all authors of this work, using roles defined by the CRediT
(Contributor Roles Taxonomy) project \cite{NISO}:

\medskip 

\noindent \textbf{DP}: Conceptualization,  Formal Analysis, Investigation, Methodology, Validation, Writing - Original Draft, Writing - Review \& Editing.

\medskip 
\noindent \textbf{DK}: Writing - Review \& Editing

\medskip 
\noindent \textbf{SP}: Writing - Review \& Editing

\medskip 
\noindent \textbf{MMW}: Conceptualization, Formal Analysis, Funding acquisition, Investigation, Methodology,     Project administration,  Supervision, Validation, Writing - Original Draft, Writing - Review \& Editing.

\bibliography{Ref}

\begin{thebibliography}{63}%
\makeatletter
\providecommand \@ifxundefined [1]{%
 \@ifx{#1\undefined}
}%
\providecommand \@ifnum [1]{%
 \ifnum #1\expandafter \@firstoftwo
 \else \expandafter \@secondoftwo
 \fi
}%
\providecommand \@ifx [1]{%
 \ifx #1\expandafter \@firstoftwo
 \else \expandafter \@secondoftwo
 \fi
}%
\providecommand \natexlab [1]{#1}%
\providecommand \enquote  [1]{``#1''}%
\providecommand \bibnamefont  [1]{#1}%
\providecommand \bibfnamefont [1]{#1}%
\providecommand \citenamefont [1]{#1}%
\providecommand \href@noop [0]{\@secondoftwo}%
\providecommand \href [0]{\begingroup \@sanitize@url \@href}%
\providecommand \@href[1]{\@@startlink{#1}\@@href}%
\providecommand \@@href[1]{\endgroup#1\@@endlink}%
\providecommand \@sanitize@url [0]{\catcode `\\12\catcode `\$12\catcode `\&12\catcode `\#12\catcode `\^12\catcode `\_12\catcode `\%12\relax}%
\providecommand \@@startlink[1]{}%
\providecommand \@@endlink[0]{}%
\providecommand \url  [0]{\begingroup\@sanitize@url \@url }%
\providecommand \@url [1]{\endgroup\@href {#1}{\urlprefix }}%
\providecommand \urlprefix  [0]{URL }%
\providecommand \Eprint [0]{\href }%
\providecommand \doibase [0]{https://doi.org/}%
\providecommand \selectlanguage [0]{\@gobble}%
\providecommand \bibinfo  [0]{\@secondoftwo}%
\providecommand \bibfield  [0]{\@secondoftwo}%
\providecommand \translation [1]{[#1]}%
\providecommand \BibitemOpen [0]{}%
\providecommand \bibitemStop [0]{}%
\providecommand \bibitemNoStop [0]{.\EOS\space}%
\providecommand \EOS [0]{\spacefactor3000\relax}%
\providecommand \BibitemShut  [1]{\csname bibitem#1\endcsname}%
\let\auto@bib@innerbib\@empty
\bibitem [{\citenamefont {Lieb}(2005)}]{lieb2005stability}%
  \BibitemOpen
  \bibfield  {author} {\bibinfo {author} {\bibfnamefont {E.~H.}\ \bibnamefont {Lieb}},\ }\href {https://doi.org/10.1007/b138553} {\emph {\bibinfo {title} {The Stability of Matter: From Atoms to Stars}}}\ (\bibinfo  {publisher} {Springer},\ \bibinfo {address} {Springer Berlin, Heidelberg},\ \bibinfo {year} {2005})\BibitemShut {NoStop}%
\bibitem [{\citenamefont {Steinhauser}\ and\ \citenamefont {Hiermaier}(2009)}]{Steinhauser2009}%
  \BibitemOpen
  \bibfield  {author} {\bibinfo {author} {\bibfnamefont {M.~O.}\ \bibnamefont {Steinhauser}}\ and\ \bibinfo {author} {\bibfnamefont {S.}~\bibnamefont {Hiermaier}},\ }\bibfield  {title} {\bibinfo {title} {A review of computational methods in materials science: Examples from shock-wave and polymer physics},\ }\href {https://doi.org/10.3390/ijms10125135} {\bibfield  {journal} {\bibinfo  {journal} {International Journal of Molecular Sciences}\ }\textbf {\bibinfo {volume} {10}},\ \bibinfo {pages} {5135} (\bibinfo {year} {2009})}\BibitemShut {NoStop}%
\bibitem [{\citenamefont {Continentino}(2021)}]{Continentino2021}%
  \BibitemOpen
  \bibfield  {author} {\bibinfo {author} {\bibfnamefont {M.~A.}\ \bibnamefont {Continentino}},\ }\href {https://doi.org/10.1088/978-0-7503-3395-5} {\emph {\bibinfo {title} {Key Methods and Concepts in Condensed Matter Physics}}},\ 2053-2563\ (\bibinfo  {publisher} {IOP Publishing},\ \bibinfo {year} {2021})\BibitemShut {NoStop}%
\bibitem [{\citenamefont {Deglmann}\ \emph {et~al.}(2015)\citenamefont {Deglmann}, \citenamefont {Sch\"afer},\ and\ \citenamefont {Lennartz}}]{DSL15}%
  \BibitemOpen
  \bibfield  {author} {\bibinfo {author} {\bibfnamefont {P.}~\bibnamefont {Deglmann}}, \bibinfo {author} {\bibfnamefont {A.}~\bibnamefont {Sch\"afer}},\ and\ \bibinfo {author} {\bibfnamefont {C.}~\bibnamefont {Lennartz}},\ }\bibfield  {title} {\bibinfo {title} {Application of quantum calculations in the chemical industry---an overview},\ }\href {https://doi.org/https://doi.org/10.1002/qua.24811} {\bibfield  {journal} {\bibinfo  {journal} {International Journal of Quantum Chemistry}\ }\textbf {\bibinfo {volume} {115}},\ \bibinfo {pages} {107} (\bibinfo {year} {2015})}\BibitemShut {NoStop}%
\bibitem [{\citenamefont {Schuch}\ and\ \citenamefont {Verstraete}(2009)}]{Schuch2009}%
  \BibitemOpen
  \bibfield  {author} {\bibinfo {author} {\bibfnamefont {N.}~\bibnamefont {Schuch}}\ and\ \bibinfo {author} {\bibfnamefont {F.}~\bibnamefont {Verstraete}},\ }\bibfield  {title} {\bibinfo {title} {Computational complexity of interacting electrons and fundamental limitations of density functional theory},\ }\href {https://doi.org/10.1038/nphys1370} {\bibfield  {journal} {\bibinfo  {journal} {Nature Physics}\ }\textbf {\bibinfo {volume} {5}},\ \bibinfo {pages} {732} (\bibinfo {year} {2009})}\BibitemShut {NoStop}%
\bibitem [{\citenamefont {Childs}\ \emph {et~al.}(2014)\citenamefont {Childs}, \citenamefont {Gosset},\ and\ \citenamefont {Webb}}]{Childs2014}%
  \BibitemOpen
  \bibfield  {author} {\bibinfo {author} {\bibfnamefont {A.~M.}\ \bibnamefont {Childs}}, \bibinfo {author} {\bibfnamefont {D.}~\bibnamefont {Gosset}},\ and\ \bibinfo {author} {\bibfnamefont {Z.}~\bibnamefont {Webb}},\ }\bibfield  {title} {\bibinfo {title} {The {B}ose-{H}ubbard model is {QMA}-complete},\ }in\ \href {https://doi.org/10.1007/978-3-662-43948-7_26} {\emph {\bibinfo {booktitle} {Automata, Languages, and Programming}}},\ \bibinfo {editor} {edited by\ \bibinfo {editor} {\bibfnamefont {J.}~\bibnamefont {Esparza}}, \bibinfo {editor} {\bibfnamefont {P.}~\bibnamefont {Fraigniaud}}, \bibinfo {editor} {\bibfnamefont {T.}~\bibnamefont {Husfeldt}},\ and\ \bibinfo {editor} {\bibfnamefont {E.}~\bibnamefont {Koutsoupias}}}\ (\bibinfo  {publisher} {Springer Berlin Heidelberg},\ \bibinfo {address} {Berlin, Heidelberg},\ \bibinfo {year} {2014})\ pp.\ \bibinfo {pages} {308--319}\BibitemShut {NoStop}%
\bibitem [{\citenamefont {Huang}(2021)}]{Huang2021}%
  \BibitemOpen
  \bibfield  {author} {\bibinfo {author} {\bibfnamefont {Y.}~\bibnamefont {Huang}},\ }\bibfield  {title} {\bibinfo {title} {Two-dimensional local {H}amiltonian problem with area laws is {QMA}-complete},\ }\href {https://doi.org/https://doi.org/10.1016/j.jcp.2021.110534} {\bibfield  {journal} {\bibinfo  {journal} {Journal of Computational Physics}\ }\textbf {\bibinfo {volume} {443}},\ \bibinfo {pages} {110534} (\bibinfo {year} {2021})}\BibitemShut {NoStop}%
\bibitem [{\citenamefont {Gerjuoy}\ \emph {et~al.}(1983)\citenamefont {Gerjuoy}, \citenamefont {Rau},\ and\ \citenamefont {Spruch}}]{Gerjuoy1983}%
  \BibitemOpen
  \bibfield  {author} {\bibinfo {author} {\bibfnamefont {E.}~\bibnamefont {Gerjuoy}}, \bibinfo {author} {\bibfnamefont {A.~R.~P.}\ \bibnamefont {Rau}},\ and\ \bibinfo {author} {\bibfnamefont {L.}~\bibnamefont {Spruch}},\ }\bibfield  {title} {\bibinfo {title} {A unified formulation of the construction of variational principles},\ }\href {https://doi.org/10.1103/RevModPhys.55.725} {\bibfield  {journal} {\bibinfo  {journal} {Reviews of Modern Physics}\ }\textbf {\bibinfo {volume} {55}},\ \bibinfo {pages} {725} (\bibinfo {year} {1983})}\BibitemShut {NoStop}%
\bibitem [{\citenamefont {Fannes}\ \emph {et~al.}(1992)\citenamefont {Fannes}, \citenamefont {Nachtergaele},\ and\ \citenamefont {Werner}}]{Fannes1992}%
  \BibitemOpen
  \bibfield  {author} {\bibinfo {author} {\bibfnamefont {M.}~\bibnamefont {Fannes}}, \bibinfo {author} {\bibfnamefont {B.}~\bibnamefont {Nachtergaele}},\ and\ \bibinfo {author} {\bibfnamefont {R.~F.}\ \bibnamefont {Werner}},\ }\bibfield  {title} {\bibinfo {title} {Finitely correlated states on quantum spin chains},\ }\href {https://doi.org/10.1007/BF02099178} {\bibfield  {journal} {\bibinfo  {journal} {Communications in Mathematical Physics}\ }\textbf {\bibinfo {volume} {144}},\ \bibinfo {pages} {443} (\bibinfo {year} {1992})}\BibitemShut {NoStop}%
\bibitem [{\citenamefont {Verstraete}\ and\ \citenamefont {Cirac}(2006)}]{Verstraete2006}%
  \BibitemOpen
  \bibfield  {author} {\bibinfo {author} {\bibfnamefont {F.}~\bibnamefont {Verstraete}}\ and\ \bibinfo {author} {\bibfnamefont {J.~I.}\ \bibnamefont {Cirac}},\ }\bibfield  {title} {\bibinfo {title} {Matrix product states represent ground states faithfully},\ }\href {https://doi.org/10.1103/PhysRevB.73.094423} {\bibfield  {journal} {\bibinfo  {journal} {Physical Review B}\ }\textbf {\bibinfo {volume} {73}},\ \bibinfo {pages} {094423} (\bibinfo {year} {2006})}\BibitemShut {NoStop}%
\bibitem [{\citenamefont {Perez-Garcia}\ \emph {et~al.}(2007)\citenamefont {Perez-Garcia}, \citenamefont {Verstraete}, \citenamefont {Wolf},\ and\ \citenamefont {Cirac}}]{PerezGarcia2007}%
  \BibitemOpen
  \bibfield  {author} {\bibinfo {author} {\bibfnamefont {D.}~\bibnamefont {Perez-Garcia}}, \bibinfo {author} {\bibfnamefont {F.}~\bibnamefont {Verstraete}}, \bibinfo {author} {\bibfnamefont {M.~M.}\ \bibnamefont {Wolf}},\ and\ \bibinfo {author} {\bibfnamefont {J.~I.}\ \bibnamefont {Cirac}},\ }\bibfield  {title} {\bibinfo {title} {Matrix product state representations},\ }\href {https://doi.org/10.26421/QIC7.5-6-1} {\bibfield  {journal} {\bibinfo  {journal} {Quantum Information and Computation}\ }\textbf {\bibinfo {volume} {7}},\ \bibinfo {pages} {401} (\bibinfo {year} {2007})}\BibitemShut {NoStop}%
\bibitem [{\citenamefont {Abrams}\ and\ \citenamefont {Lloyd}(1999)}]{Abrams1999}%
  \BibitemOpen
  \bibfield  {author} {\bibinfo {author} {\bibfnamefont {D.~S.}\ \bibnamefont {Abrams}}\ and\ \bibinfo {author} {\bibfnamefont {S.}~\bibnamefont {Lloyd}},\ }\bibfield  {title} {\bibinfo {title} {Quantum algorithm providing exponential speed increase for finding eigenvalues and eigenvectors},\ }\href {https://doi.org/10.1103/PhysRevLett.83.5162} {\bibfield  {journal} {\bibinfo  {journal} {Physical Review Letters}\ }\textbf {\bibinfo {volume} {83}},\ \bibinfo {pages} {5162} (\bibinfo {year} {1999})}\BibitemShut {NoStop}%
\bibitem [{\citenamefont {Aspuru-Guzik}\ \emph {et~al.}(2005)\citenamefont {Aspuru-Guzik}, \citenamefont {Dutoi}, \citenamefont {Love},\ and\ \citenamefont {Head-Gordon}}]{ADLH05}%
  \BibitemOpen
  \bibfield  {author} {\bibinfo {author} {\bibfnamefont {A.}~\bibnamefont {Aspuru-Guzik}}, \bibinfo {author} {\bibfnamefont {A.~D.}\ \bibnamefont {Dutoi}}, \bibinfo {author} {\bibfnamefont {P.~J.}\ \bibnamefont {Love}},\ and\ \bibinfo {author} {\bibfnamefont {M.}~\bibnamefont {Head-Gordon}},\ }\bibfield  {title} {\bibinfo {title} {Simulated quantum computation of molecular energies},\ }\href {https://doi.org/10.1126/science.1113479} {\bibfield  {journal} {\bibinfo  {journal} {Science}\ }\textbf {\bibinfo {volume} {309}},\ \bibinfo {pages} {1704} (\bibinfo {year} {2005})}\BibitemShut {NoStop}%
\bibitem [{\citenamefont {Lin}\ and\ \citenamefont {Tong}(2022)}]{Lin2022}%
  \BibitemOpen
  \bibfield  {author} {\bibinfo {author} {\bibfnamefont {L.}~\bibnamefont {Lin}}\ and\ \bibinfo {author} {\bibfnamefont {Y.}~\bibnamefont {Tong}},\ }\bibfield  {title} {\bibinfo {title} {Heisenberg-limited ground-state energy estimation for early fault-tolerant quantum computers},\ }\href {https://doi.org/10.1103/PRXQuantum.3.010318} {\bibfield  {journal} {\bibinfo  {journal} {PRX Quantum}\ }\textbf {\bibinfo {volume} {3}},\ \bibinfo {pages} {010318} (\bibinfo {year} {2022})}\BibitemShut {NoStop}%
\bibitem [{\citenamefont {Dong}\ \emph {et~al.}(2022)\citenamefont {Dong}, \citenamefont {Lin},\ and\ \citenamefont {Tong}}]{Dong2022}%
  \BibitemOpen
  \bibfield  {author} {\bibinfo {author} {\bibfnamefont {Y.}~\bibnamefont {Dong}}, \bibinfo {author} {\bibfnamefont {L.}~\bibnamefont {Lin}},\ and\ \bibinfo {author} {\bibfnamefont {Y.}~\bibnamefont {Tong}},\ }\bibfield  {title} {\bibinfo {title} {Ground-state preparation and energy estimation on early fault-tolerant quantum computers via quantum eigenvalue transformation of unitary matrices},\ }\href {https://doi.org/10.1103/PRXQuantum.3.040305} {\bibfield  {journal} {\bibinfo  {journal} {PRX Quantum}\ }\textbf {\bibinfo {volume} {3}},\ \bibinfo {pages} {040305} (\bibinfo {year} {2022})}\BibitemShut {NoStop}%
\bibitem [{\citenamefont {Wan}\ \emph {et~al.}(2022)\citenamefont {Wan}, \citenamefont {Berta},\ and\ \citenamefont {Campbell}}]{Wan2022}%
  \BibitemOpen
  \bibfield  {author} {\bibinfo {author} {\bibfnamefont {K.}~\bibnamefont {Wan}}, \bibinfo {author} {\bibfnamefont {M.}~\bibnamefont {Berta}},\ and\ \bibinfo {author} {\bibfnamefont {E.~T.}\ \bibnamefont {Campbell}},\ }\bibfield  {title} {\bibinfo {title} {Randomized quantum algorithm for statistical phase estimation},\ }\href {https://doi.org/10.1103/PhysRevLett.129.030503} {\bibfield  {journal} {\bibinfo  {journal} {Physical Review Letters}\ }\textbf {\bibinfo {volume} {129}},\ \bibinfo {pages} {030503} (\bibinfo {year} {2022})}\BibitemShut {NoStop}%
\bibitem [{\citenamefont {Ding}\ and\ \citenamefont {Lin}(2023)}]{Ding2023}%
  \BibitemOpen
  \bibfield  {author} {\bibinfo {author} {\bibfnamefont {Z.}~\bibnamefont {Ding}}\ and\ \bibinfo {author} {\bibfnamefont {L.}~\bibnamefont {Lin}},\ }\bibfield  {title} {\bibinfo {title} {Even shorter quantum circuit for phase estimation on early fault-tolerant quantum computers with applications to ground-state energy estimation},\ }\href {https://doi.org/10.1103/PRXQuantum.4.020331} {\bibfield  {journal} {\bibinfo  {journal} {PRX Quantum}\ }\textbf {\bibinfo {volume} {4}},\ \bibinfo {pages} {020331} (\bibinfo {year} {2023})}\BibitemShut {NoStop}%
\bibitem [{\citenamefont {Wang}\ \emph {et~al.}(2023{\natexlab{a}})\citenamefont {Wang}, \citenamefont {Fran{\c{c}}a}, \citenamefont {Zhang}, \citenamefont {Zhu},\ and\ \citenamefont {Johnson}}]{Wang2023quantumalgorithm}%
  \BibitemOpen
  \bibfield  {author} {\bibinfo {author} {\bibfnamefont {G.}~\bibnamefont {Wang}}, \bibinfo {author} {\bibfnamefont {D.~S.}\ \bibnamefont {Fran{\c{c}}a}}, \bibinfo {author} {\bibfnamefont {R.}~\bibnamefont {Zhang}}, \bibinfo {author} {\bibfnamefont {S.}~\bibnamefont {Zhu}},\ and\ \bibinfo {author} {\bibfnamefont {P.~D.}\ \bibnamefont {Johnson}},\ }\bibfield  {title} {\bibinfo {title} {Quantum algorithm for ground state energy estimation using circuit depth with exponentially improved dependence on precision},\ }\href {https://doi.org/10.22331/q-2023-11-06-1167} {\bibfield  {journal} {\bibinfo  {journal} {{Quantum}}\ }\textbf {\bibinfo {volume} {7}},\ \bibinfo {pages} {1167} (\bibinfo {year} {2023}{\natexlab{a}})}\BibitemShut {NoStop}%
\bibitem [{\citenamefont {Wang}\ \emph {et~al.}(2023{\natexlab{b}})\citenamefont {Wang}, \citenamefont {Fran{\c{c}}a}, \citenamefont {Rendon},\ and\ \citenamefont {Johnson}}]{wang2023fastergroundstateenergy}%
  \BibitemOpen
  \bibfield  {author} {\bibinfo {author} {\bibfnamefont {G.}~\bibnamefont {Wang}}, \bibinfo {author} {\bibfnamefont {D.~S.}\ \bibnamefont {Fran{\c{c}}a}}, \bibinfo {author} {\bibfnamefont {G.}~\bibnamefont {Rendon}},\ and\ \bibinfo {author} {\bibfnamefont {P.~D.}\ \bibnamefont {Johnson}},\ }\href {https://arxiv.org/abs/2304.09827} {\bibinfo {title} {Faster ground state energy estimation on early fault-tolerant quantum computers via rejection sampling}} (\bibinfo {year} {2023}{\natexlab{b}}),\ \Eprint {https://arxiv.org/abs/2304.09827} {arXiv:2304.09827 [quant-ph]} \BibitemShut {NoStop}%
\bibitem [{\citenamefont {Peruzzo}\ \emph {et~al.}(2014)\citenamefont {Peruzzo}, \citenamefont {McClean}, \citenamefont {Shadbolt}, \citenamefont {Yung}, \citenamefont {Zhou}, \citenamefont {Love}, \citenamefont {Aspuru-Guzik},\ and\ \citenamefont {O'Brien}}]{Peruzzo2014}%
  \BibitemOpen
  \bibfield  {author} {\bibinfo {author} {\bibfnamefont {A.}~\bibnamefont {Peruzzo}}, \bibinfo {author} {\bibfnamefont {J.}~\bibnamefont {McClean}}, \bibinfo {author} {\bibfnamefont {P.}~\bibnamefont {Shadbolt}}, \bibinfo {author} {\bibfnamefont {M.-H.}\ \bibnamefont {Yung}}, \bibinfo {author} {\bibfnamefont {X.-Q.}\ \bibnamefont {Zhou}}, \bibinfo {author} {\bibfnamefont {P.~J.}\ \bibnamefont {Love}}, \bibinfo {author} {\bibfnamefont {A.}~\bibnamefont {Aspuru-Guzik}},\ and\ \bibinfo {author} {\bibfnamefont {J.~L.}\ \bibnamefont {O'Brien}},\ }\bibfield  {title} {\bibinfo {title} {A variational eigenvalue solver on a photonic quantum processor},\ }\href {https://doi.org/10.1038/ncomms5213} {\bibfield  {journal} {\bibinfo  {journal} {Nature Communications}\ }\textbf {\bibinfo {volume} {5}},\ \bibinfo {pages} {4213} (\bibinfo {year} {2014})}\BibitemShut {NoStop}%
\bibitem [{\citenamefont {Tilly}\ \emph {et~al.}(2022)\citenamefont {Tilly}, \citenamefont {Chen}, \citenamefont {Cao}, \citenamefont {Picozzi}, \citenamefont {Setia}, \citenamefont {Li}, \citenamefont {Grant}, \citenamefont {Wossnig}, \citenamefont {Rungger}, \citenamefont {Booth},\ and\ \citenamefont {Tennyson}}]{Tilly2022}%
  \BibitemOpen
  \bibfield  {author} {\bibinfo {author} {\bibfnamefont {J.}~\bibnamefont {Tilly}}, \bibinfo {author} {\bibfnamefont {H.}~\bibnamefont {Chen}}, \bibinfo {author} {\bibfnamefont {S.}~\bibnamefont {Cao}}, \bibinfo {author} {\bibfnamefont {D.}~\bibnamefont {Picozzi}}, \bibinfo {author} {\bibfnamefont {K.}~\bibnamefont {Setia}}, \bibinfo {author} {\bibfnamefont {Y.}~\bibnamefont {Li}}, \bibinfo {author} {\bibfnamefont {E.}~\bibnamefont {Grant}}, \bibinfo {author} {\bibfnamefont {L.}~\bibnamefont {Wossnig}}, \bibinfo {author} {\bibfnamefont {I.}~\bibnamefont {Rungger}}, \bibinfo {author} {\bibfnamefont {G.~H.}\ \bibnamefont {Booth}},\ and\ \bibinfo {author} {\bibfnamefont {J.}~\bibnamefont {Tennyson}},\ }\bibfield  {title} {\bibinfo {title} {The variational quantum eigensolver: A review of methods and best practices},\ }\href {https://doi.org/https://doi.org/10.1016/j.physrep.2022.08.003} {\bibfield  {journal} {\bibinfo  {journal} {Physics Reports}\ }\textbf {\bibinfo {volume} {986}},\ \bibinfo {pages} {1}
  (\bibinfo {year} {2022})}\BibitemShut {NoStop}%
\bibitem [{\citenamefont {McClean}\ \emph {et~al.}(2018)\citenamefont {McClean}, \citenamefont {Boixo}, \citenamefont {Smelyanskiy}, \citenamefont {Babbush},\ and\ \citenamefont {Neven}}]{McClean_2018}%
  \BibitemOpen
  \bibfield  {author} {\bibinfo {author} {\bibfnamefont {J.~R.}\ \bibnamefont {McClean}}, \bibinfo {author} {\bibfnamefont {S.}~\bibnamefont {Boixo}}, \bibinfo {author} {\bibfnamefont {V.~N.}\ \bibnamefont {Smelyanskiy}}, \bibinfo {author} {\bibfnamefont {R.}~\bibnamefont {Babbush}},\ and\ \bibinfo {author} {\bibfnamefont {H.}~\bibnamefont {Neven}},\ }\bibfield  {title} {\bibinfo {title} {Barren plateaus in quantum neural network training landscapes},\ }\href {https://doi.org/10.1038/s41467-018-07090-4} {\bibfield  {journal} {\bibinfo  {journal} {Nature Communications}\ }\textbf {\bibinfo {volume} {9}},\ \bibinfo {pages} {4812} (\bibinfo {year} {2018})}\BibitemShut {NoStop}%
\bibitem [{\citenamefont {Marrero}\ \emph {et~al.}(2021)\citenamefont {Marrero}, \citenamefont {Kieferov{\'a}},\ and\ \citenamefont {Wiebe}}]{marrero2021entanglement}%
  \BibitemOpen
  \bibfield  {author} {\bibinfo {author} {\bibfnamefont {C.~O.}\ \bibnamefont {Marrero}}, \bibinfo {author} {\bibfnamefont {M.}~\bibnamefont {Kieferov{\'a}}},\ and\ \bibinfo {author} {\bibfnamefont {N.}~\bibnamefont {Wiebe}},\ }\bibfield  {title} {\bibinfo {title} {Entanglement-induced barren plateaus},\ }\href {https://doi.org/10.1103/PRXQuantum.2.040316} {\bibfield  {journal} {\bibinfo  {journal} {PRX Quantum}\ }\textbf {\bibinfo {volume} {2}},\ \bibinfo {pages} {040316} (\bibinfo {year} {2021})}\BibitemShut {NoStop}%
\bibitem [{\citenamefont {Arrasmith}\ \emph {et~al.}(2022)\citenamefont {Arrasmith}, \citenamefont {Holmes}, \citenamefont {Cerezo},\ and\ \citenamefont {Coles}}]{arrasmith2022equivalence}%
  \BibitemOpen
  \bibfield  {author} {\bibinfo {author} {\bibfnamefont {A.}~\bibnamefont {Arrasmith}}, \bibinfo {author} {\bibfnamefont {Z.}~\bibnamefont {Holmes}}, \bibinfo {author} {\bibfnamefont {M.}~\bibnamefont {Cerezo}},\ and\ \bibinfo {author} {\bibfnamefont {P.~J.}\ \bibnamefont {Coles}},\ }\bibfield  {title} {\bibinfo {title} {Equivalence of quantum barren plateaus to cost concentration and narrow gorges},\ }\href {https://doi.org/10.1088/2058-9565/ac7d06} {\bibfield  {journal} {\bibinfo  {journal} {Quantum Science and Technology}\ }\textbf {\bibinfo {volume} {7}},\ \bibinfo {pages} {045015} (\bibinfo {year} {2022})}\BibitemShut {NoStop}%
\bibitem [{\citenamefont {Holmes}\ \emph {et~al.}(2022)\citenamefont {Holmes}, \citenamefont {Sharma}, \citenamefont {Cerezo},\ and\ \citenamefont {Coles}}]{holmes2022connecting}%
  \BibitemOpen
  \bibfield  {author} {\bibinfo {author} {\bibfnamefont {Z.}~\bibnamefont {Holmes}}, \bibinfo {author} {\bibfnamefont {K.}~\bibnamefont {Sharma}}, \bibinfo {author} {\bibfnamefont {M.}~\bibnamefont {Cerezo}},\ and\ \bibinfo {author} {\bibfnamefont {P.~J.}\ \bibnamefont {Coles}},\ }\bibfield  {title} {\bibinfo {title} {Connecting ansatz expressibility to gradient magnitudes and barren plateaus},\ }\href {https://doi.org/10.1103/PRXQuantum.3.010313} {\bibfield  {journal} {\bibinfo  {journal} {PRX Quantum}\ }\textbf {\bibinfo {volume} {3}},\ \bibinfo {pages} {010313} (\bibinfo {year} {2022})}\BibitemShut {NoStop}%
\bibitem [{\citenamefont {Fontana}\ \emph {et~al.}(2024)\citenamefont {Fontana}, \citenamefont {Herman}, \citenamefont {Chakrabarti}, \citenamefont {Kumar}, \citenamefont {Yalovetzky}, \citenamefont {Heredge}, \citenamefont {Sureshbabu},\ and\ \citenamefont {Pistoia}}]{Fontana2024}%
  \BibitemOpen
  \bibfield  {author} {\bibinfo {author} {\bibfnamefont {E.}~\bibnamefont {Fontana}}, \bibinfo {author} {\bibfnamefont {D.}~\bibnamefont {Herman}}, \bibinfo {author} {\bibfnamefont {S.}~\bibnamefont {Chakrabarti}}, \bibinfo {author} {\bibfnamefont {N.}~\bibnamefont {Kumar}}, \bibinfo {author} {\bibfnamefont {R.}~\bibnamefont {Yalovetzky}}, \bibinfo {author} {\bibfnamefont {J.}~\bibnamefont {Heredge}}, \bibinfo {author} {\bibfnamefont {S.~H.}\ \bibnamefont {Sureshbabu}},\ and\ \bibinfo {author} {\bibfnamefont {M.}~\bibnamefont {Pistoia}},\ }\bibfield  {title} {\bibinfo {title} {Characterizing barren plateaus in quantum ans\"atze with the adjoint representation},\ }\href {https://doi.org/10.1038/s41467-024-49910-w} {\bibfield  {journal} {\bibinfo  {journal} {Nature Communications}\ }\textbf {\bibinfo {volume} {15}},\ \bibinfo {pages} {7171} (\bibinfo {year} {2024})}\BibitemShut {NoStop}%
\bibitem [{\citenamefont {Ragone}\ \emph {et~al.}(2024)\citenamefont {Ragone}, \citenamefont {Bakalov}, \citenamefont {Sauvage}, \citenamefont {Kemper}, \citenamefont {Ortiz~Marrero}, \citenamefont {Larocca},\ and\ \citenamefont {Cerezo}}]{Ragone2024}%
  \BibitemOpen
  \bibfield  {author} {\bibinfo {author} {\bibfnamefont {M.}~\bibnamefont {Ragone}}, \bibinfo {author} {\bibfnamefont {B.~N.}\ \bibnamefont {Bakalov}}, \bibinfo {author} {\bibfnamefont {F.}~\bibnamefont {Sauvage}}, \bibinfo {author} {\bibfnamefont {A.~F.}\ \bibnamefont {Kemper}}, \bibinfo {author} {\bibfnamefont {C.}~\bibnamefont {Ortiz~Marrero}}, \bibinfo {author} {\bibfnamefont {M.}~\bibnamefont {Larocca}},\ and\ \bibinfo {author} {\bibfnamefont {M.}~\bibnamefont {Cerezo}},\ }\bibfield  {title} {\bibinfo {title} {A {L}ie algebraic theory of barren plateaus for deep parameterized quantum circuits},\ }\href {https://doi.org/10.1038/s41467-024-49909-3} {\bibfield  {journal} {\bibinfo  {journal} {Nature Communications}\ }\textbf {\bibinfo {volume} {15}},\ \bibinfo {pages} {7172} (\bibinfo {year} {2024})}\BibitemShut {NoStop}%
\bibitem [{\citenamefont {Amin}\ \emph {et~al.}(2018)\citenamefont {Amin}, \citenamefont {Andriyash}, \citenamefont {Rolfe}, \citenamefont {Kulchytskyy},\ and\ \citenamefont {Melko}}]{Amin2018}%
  \BibitemOpen
  \bibfield  {author} {\bibinfo {author} {\bibfnamefont {M.~H.}\ \bibnamefont {Amin}}, \bibinfo {author} {\bibfnamefont {E.}~\bibnamefont {Andriyash}}, \bibinfo {author} {\bibfnamefont {J.}~\bibnamefont {Rolfe}}, \bibinfo {author} {\bibfnamefont {B.}~\bibnamefont {Kulchytskyy}},\ and\ \bibinfo {author} {\bibfnamefont {R.}~\bibnamefont {Melko}},\ }\bibfield  {title} {\bibinfo {title} {Quantum {B}oltzmann machine},\ }\href {https://doi.org/10.1103/PhysRevX.8.021050} {\bibfield  {journal} {\bibinfo  {journal} {Physical Review X}\ }\textbf {\bibinfo {volume} {8}},\ \bibinfo {pages} {021050} (\bibinfo {year} {2018})}\BibitemShut {NoStop}%
\bibitem [{\citenamefont {Benedetti}\ \emph {et~al.}(2017)\citenamefont {Benedetti}, \citenamefont {Realpe-G\'omez}, \citenamefont {Biswas},\ and\ \citenamefont {Perdomo-Ortiz}}]{Benedetti2017}%
  \BibitemOpen
  \bibfield  {author} {\bibinfo {author} {\bibfnamefont {M.}~\bibnamefont {Benedetti}}, \bibinfo {author} {\bibfnamefont {J.}~\bibnamefont {Realpe-G\'omez}}, \bibinfo {author} {\bibfnamefont {R.}~\bibnamefont {Biswas}},\ and\ \bibinfo {author} {\bibfnamefont {A.}~\bibnamefont {Perdomo-Ortiz}},\ }\bibfield  {title} {\bibinfo {title} {Quantum-assisted learning of hardware-embedded probabilistic graphical models},\ }\href {https://doi.org/10.1103/PhysRevX.7.041052} {\bibfield  {journal} {\bibinfo  {journal} {Physical Review X}\ }\textbf {\bibinfo {volume} {7}},\ \bibinfo {pages} {041052} (\bibinfo {year} {2017})}\BibitemShut {NoStop}%
\bibitem [{\citenamefont {Kieferov\'a}\ and\ \citenamefont {Wiebe}(2017)}]{Kieferova2017}%
  \BibitemOpen
  \bibfield  {author} {\bibinfo {author} {\bibfnamefont {M.}~\bibnamefont {Kieferov\'a}}\ and\ \bibinfo {author} {\bibfnamefont {N.}~\bibnamefont {Wiebe}},\ }\bibfield  {title} {\bibinfo {title} {Tomography and generative training with quantum {B}oltzmann machines},\ }\href {https://doi.org/10.1103/PhysRevA.96.062327} {\bibfield  {journal} {\bibinfo  {journal} {Physical Review A}\ }\textbf {\bibinfo {volume} {96}},\ \bibinfo {pages} {062327} (\bibinfo {year} {2017})}\BibitemShut {NoStop}%
\bibitem [{\citenamefont {Chen}\ \emph {et~al.}(2023{\natexlab{a}})\citenamefont {Chen}, \citenamefont {Kastoryano}, \citenamefont {Brand{\~a}o},\ and\ \citenamefont {Gily{\'e}n}}]{chen2023quantumthermalstatepreparation}%
  \BibitemOpen
  \bibfield  {author} {\bibinfo {author} {\bibfnamefont {C.-F.}\ \bibnamefont {Chen}}, \bibinfo {author} {\bibfnamefont {M.~J.}\ \bibnamefont {Kastoryano}}, \bibinfo {author} {\bibfnamefont {F.~G. S.~L.}\ \bibnamefont {Brand{\~a}o}},\ and\ \bibinfo {author} {\bibfnamefont {A.}~\bibnamefont {Gily{\'e}n}},\ }\href {https://arxiv.org/abs/2303.18224} {\bibinfo {title} {Quantum thermal state preparation}} (\bibinfo {year} {2023}{\natexlab{a}}),\ \Eprint {https://arxiv.org/abs/2303.18224} {arXiv:2303.18224 [quant-ph]} \BibitemShut {NoStop}%
\bibitem [{\citenamefont {Chen}\ \emph {et~al.}(2023{\natexlab{b}})\citenamefont {Chen}, \citenamefont {Kastoryano},\ and\ \citenamefont {Gily{\'e}n}}]{chen2023efficient}%
  \BibitemOpen
  \bibfield  {author} {\bibinfo {author} {\bibfnamefont {C.-F.}\ \bibnamefont {Chen}}, \bibinfo {author} {\bibfnamefont {M.~J.}\ \bibnamefont {Kastoryano}},\ and\ \bibinfo {author} {\bibfnamefont {A.}~\bibnamefont {Gily{\'e}n}},\ }\href {https://arxiv.org/abs/2311.09207} {\bibinfo {title} {An efficient and exact noncommutative quantum {G}ibbs sampler}} (\bibinfo {year} {2023}{\natexlab{b}}),\ \Eprint {https://arxiv.org/abs/2311.09207} {arXiv:2311.09207 [quant-ph]} \BibitemShut {NoStop}%
\bibitem [{\citenamefont {Bergamaschi}\ \emph {et~al.}(2024)\citenamefont {Bergamaschi}, \citenamefont {Chen},\ and\ \citenamefont {Liu}}]{bergamaschi2024}%
  \BibitemOpen
  \bibfield  {author} {\bibinfo {author} {\bibfnamefont {T.}~\bibnamefont {Bergamaschi}}, \bibinfo {author} {\bibfnamefont {C.-F.}\ \bibnamefont {Chen}},\ and\ \bibinfo {author} {\bibfnamefont {Y.}~\bibnamefont {Liu}},\ }\href {https://arxiv.org/abs/2404.14639} {\bibinfo {title} {Quantum computational advantage with constant-temperature {G}ibbs sampling}} (\bibinfo {year} {2024}),\ \Eprint {https://arxiv.org/abs/2404.14639} {arXiv:2404.14639 [quant-ph]} \BibitemShut {NoStop}%
\bibitem [{\citenamefont {Chen}\ \emph {et~al.}(2024)\citenamefont {Chen}, \citenamefont {Li}, \citenamefont {Lu},\ and\ \citenamefont {Ying}}]{chen2024randomizedmethodsimulatinglindblad}%
  \BibitemOpen
  \bibfield  {author} {\bibinfo {author} {\bibfnamefont {H.}~\bibnamefont {Chen}}, \bibinfo {author} {\bibfnamefont {B.}~\bibnamefont {Li}}, \bibinfo {author} {\bibfnamefont {J.}~\bibnamefont {Lu}},\ and\ \bibinfo {author} {\bibfnamefont {L.}~\bibnamefont {Ying}},\ }\href {https://arxiv.org/abs/2407.06594} {\bibinfo {title} {A randomized method for simulating {L}indblad equations and thermal state preparation}} (\bibinfo {year} {2024}),\ \Eprint {https://arxiv.org/abs/2407.06594} {arXiv:2407.06594 [quant-ph]} \BibitemShut {NoStop}%
\bibitem [{\citenamefont {Rajakumar}\ and\ \citenamefont {Watson}(2024)}]{rajakumar2024gibbssamplinggivesquantum}%
  \BibitemOpen
  \bibfield  {author} {\bibinfo {author} {\bibfnamefont {J.}~\bibnamefont {Rajakumar}}\ and\ \bibinfo {author} {\bibfnamefont {J.~D.}\ \bibnamefont {Watson}},\ }\href {https://arxiv.org/abs/2408.01516} {\bibinfo {title} {Gibbs sampling gives quantum advantage at constant temperatures with {$O(1)$}-local {H}amiltonians}} (\bibinfo {year} {2024}),\ \Eprint {https://arxiv.org/abs/2408.01516} {arXiv:2408.01516 [quant-ph]} \BibitemShut {NoStop}%
\bibitem [{\citenamefont {Rouz{\'e}}\ \emph {et~al.}(2024)\citenamefont {Rouz{\'e}}, \citenamefont {Franca},\ and\ \citenamefont {Alhambra}}]{rouze2024efficient}%
  \BibitemOpen
  \bibfield  {author} {\bibinfo {author} {\bibfnamefont {C.}~\bibnamefont {Rouz{\'e}}}, \bibinfo {author} {\bibfnamefont {D.~S.}\ \bibnamefont {Franca}},\ and\ \bibinfo {author} {\bibfnamefont {{\'{A}}.~M.}\ \bibnamefont {Alhambra}},\ }\href {https://arxiv.org/abs/2403.12691} {\bibinfo {title} {Efficient thermalization and universal quantum computing with quantum {G}ibbs samplers}} (\bibinfo {year} {2024}),\ \Eprint {https://arxiv.org/abs/2403.12691} {arXiv:2403.12691 [quant-ph]} \BibitemShut {NoStop}%
\bibitem [{\citenamefont {Bakshi}\ \emph {et~al.}(2024)\citenamefont {Bakshi}, \citenamefont {Liu}, \citenamefont {Moitra},\ and\ \citenamefont {Tang}}]{bakshi2024hightemperature}%
  \BibitemOpen
  \bibfield  {author} {\bibinfo {author} {\bibfnamefont {A.}~\bibnamefont {Bakshi}}, \bibinfo {author} {\bibfnamefont {A.}~\bibnamefont {Liu}}, \bibinfo {author} {\bibfnamefont {A.}~\bibnamefont {Moitra}},\ and\ \bibinfo {author} {\bibfnamefont {E.}~\bibnamefont {Tang}},\ }\href {https://arxiv.org/abs/2403.16850} {\bibinfo {title} {High-temperature {G}ibbs states are unentangled and efficiently preparable}} (\bibinfo {year} {2024}),\ \Eprint {https://arxiv.org/abs/2403.16850} {arXiv:2403.16850 [quant-ph]} \BibitemShut {NoStop}%
\bibitem [{\citenamefont {Ding}\ \emph {et~al.}(2024)\citenamefont {Ding}, \citenamefont {Li}, \citenamefont {Lin},\ and\ \citenamefont {Zhang}}]{ding2024polynomial}%
  \BibitemOpen
  \bibfield  {author} {\bibinfo {author} {\bibfnamefont {Z.}~\bibnamefont {Ding}}, \bibinfo {author} {\bibfnamefont {B.}~\bibnamefont {Li}}, \bibinfo {author} {\bibfnamefont {L.}~\bibnamefont {Lin}},\ and\ \bibinfo {author} {\bibfnamefont {R.}~\bibnamefont {Zhang}},\ }\href {https://arxiv.org/abs/2410.01206} {\bibinfo {title} {Polynomial-time preparation of low-temperature {G}ibbs states for 2{D} toric code}} (\bibinfo {year} {2024}),\ \Eprint {https://arxiv.org/abs/2410.01206} {arXiv:2410.01206 [quant-ph]} \BibitemShut {NoStop}%
\bibitem [{\citenamefont {Bravyi}\ \emph {et~al.}(2022)\citenamefont {Bravyi}, \citenamefont {Chowdhury}, \citenamefont {Gosset},\ and\ \citenamefont {Wocjan}}]{Bravyi2022}%
  \BibitemOpen
  \bibfield  {author} {\bibinfo {author} {\bibfnamefont {S.}~\bibnamefont {Bravyi}}, \bibinfo {author} {\bibfnamefont {A.}~\bibnamefont {Chowdhury}}, \bibinfo {author} {\bibfnamefont {D.}~\bibnamefont {Gosset}},\ and\ \bibinfo {author} {\bibfnamefont {P.}~\bibnamefont {Wocjan}},\ }\bibfield  {title} {\bibinfo {title} {Quantum {H}amiltonian complexity in thermal equilibrium},\ }\href {https://doi.org/10.1038/s41567-022-01742-5} {\bibfield  {journal} {\bibinfo  {journal} {Nature Physics}\ }\textbf {\bibinfo {volume} {18}},\ \bibinfo {pages} {1367} (\bibinfo {year} {2022})}\BibitemShut {NoStop}%
\bibitem [{\citenamefont {Anshu}\ \emph {et~al.}(2021)\citenamefont {Anshu}, \citenamefont {Arunachalam}, \citenamefont {Kuwahara},\ and\ \citenamefont {Soleimanifar}}]{Anshu2021}%
  \BibitemOpen
  \bibfield  {author} {\bibinfo {author} {\bibfnamefont {A.}~\bibnamefont {Anshu}}, \bibinfo {author} {\bibfnamefont {S.}~\bibnamefont {Arunachalam}}, \bibinfo {author} {\bibfnamefont {T.}~\bibnamefont {Kuwahara}},\ and\ \bibinfo {author} {\bibfnamefont {M.}~\bibnamefont {Soleimanifar}},\ }\bibfield  {title} {\bibinfo {title} {Sample-efficient learning of interacting quantum systems},\ }\href {https://doi.org/10.1038/s41567-021-01232-0} {\bibfield  {journal} {\bibinfo  {journal} {Nature Physics}\ }\textbf {\bibinfo {volume} {17}},\ \bibinfo {pages} {931} (\bibinfo {year} {2021})}\BibitemShut {NoStop}%
\bibitem [{\citenamefont {Garc\'{\i}a-Pintos}\ \emph {et~al.}(2024)\citenamefont {Garc\'{\i}a-Pintos}, \citenamefont {Bharti}, \citenamefont {Bringewatt}, \citenamefont {Dehghani}, \citenamefont {Ehrenberg}, \citenamefont {Yunger~Halpern},\ and\ \citenamefont {Gorshkov}}]{GarciaPintos2024}%
  \BibitemOpen
  \bibfield  {author} {\bibinfo {author} {\bibfnamefont {L.~P.}\ \bibnamefont {Garc\'{\i}a-Pintos}}, \bibinfo {author} {\bibfnamefont {K.}~\bibnamefont {Bharti}}, \bibinfo {author} {\bibfnamefont {J.}~\bibnamefont {Bringewatt}}, \bibinfo {author} {\bibfnamefont {H.}~\bibnamefont {Dehghani}}, \bibinfo {author} {\bibfnamefont {A.}~\bibnamefont {Ehrenberg}}, \bibinfo {author} {\bibfnamefont {N.}~\bibnamefont {Yunger~Halpern}},\ and\ \bibinfo {author} {\bibfnamefont {A.~V.}\ \bibnamefont {Gorshkov}},\ }\bibfield  {title} {\bibinfo {title} {Estimation of {H}amiltonian parameters from thermal states},\ }\href {https://doi.org/10.1103/PhysRevLett.133.040802} {\bibfield  {journal} {\bibinfo  {journal} {Physical Review Letters}\ }\textbf {\bibinfo {volume} {133}},\ \bibinfo {pages} {040802} (\bibinfo {year} {2024})}\BibitemShut {NoStop}%
\bibitem [{\citenamefont {Coopmans}\ and\ \citenamefont {Benedetti}(2024)}]{Coopmans2024}%
  \BibitemOpen
  \bibfield  {author} {\bibinfo {author} {\bibfnamefont {L.}~\bibnamefont {Coopmans}}\ and\ \bibinfo {author} {\bibfnamefont {M.}~\bibnamefont {Benedetti}},\ }\bibfield  {title} {\bibinfo {title} {On the sample complexity of quantum {B}oltzmann machine learning},\ }\href {https://doi.org/10.1038/s42005-024-01763-x} {\bibfield  {journal} {\bibinfo  {journal} {Communications Physics}\ }\textbf {\bibinfo {volume} {7}},\ \bibinfo {pages} {274} (\bibinfo {year} {2024})}\BibitemShut {NoStop}%
\bibitem [{\citenamefont {Ortiz~Marrero}\ \emph {et~al.}(2021)\citenamefont {Ortiz~Marrero}, \citenamefont {Kieferov\'a},\ and\ \citenamefont {Wiebe}}]{OrtizMarrero2021}%
  \BibitemOpen
  \bibfield  {author} {\bibinfo {author} {\bibfnamefont {C.}~\bibnamefont {Ortiz~Marrero}}, \bibinfo {author} {\bibfnamefont {M.}~\bibnamefont {Kieferov\'a}},\ and\ \bibinfo {author} {\bibfnamefont {N.}~\bibnamefont {Wiebe}},\ }\bibfield  {title} {\bibinfo {title} {Entanglement-induced barren plateaus},\ }\href {https://doi.org/10.1103/PRXQuantum.2.040316} {\bibfield  {journal} {\bibinfo  {journal} {PRX Quantum}\ }\textbf {\bibinfo {volume} {2}},\ \bibinfo {pages} {040316} (\bibinfo {year} {2021})}\BibitemShut {NoStop}%
\bibitem [{\citenamefont {Hastings}(2007)}]{Hastings2007}%
  \BibitemOpen
  \bibfield  {author} {\bibinfo {author} {\bibfnamefont {M.~B.}\ \bibnamefont {Hastings}},\ }\bibfield  {title} {\bibinfo {title} {Quantum belief propagation: An algorithm for thermal quantum systems},\ }\href {https://doi.org/10.1103/PhysRevB.76.201102} {\bibfield  {journal} {\bibinfo  {journal} {Physical Review B}\ }\textbf {\bibinfo {volume} {76}},\ \bibinfo {pages} {201102} (\bibinfo {year} {2007})}\BibitemShut {NoStop}%
\bibitem [{\citenamefont {Kim}(2012)}]{Kim2012}%
  \BibitemOpen
  \bibfield  {author} {\bibinfo {author} {\bibfnamefont {I.~H.}\ \bibnamefont {Kim}},\ }\bibfield  {title} {\bibinfo {title} {Perturbative analysis of topological entanglement entropy from conditional independence},\ }\href {https://doi.org/10.1103/PhysRevB.86.245116} {\bibfield  {journal} {\bibinfo  {journal} {Physical Review B}\ }\textbf {\bibinfo {volume} {86}},\ \bibinfo {pages} {245116} (\bibinfo {year} {2012})}\BibitemShut {NoStop}%
\bibitem [{\citenamefont {Kato}\ and\ \citenamefont {Brand{\~a}o}(2019)}]{Kato2019}%
  \BibitemOpen
  \bibfield  {author} {\bibinfo {author} {\bibfnamefont {K.}~\bibnamefont {Kato}}\ and\ \bibinfo {author} {\bibfnamefont {F.~G. S.~L.}\ \bibnamefont {Brand{\~a}o}},\ }\bibfield  {title} {\bibinfo {title} {Quantum approximate {M}arkov chains are thermal},\ }\href {https://doi.org/10.1007/s00220-019-03485-6} {\bibfield  {journal} {\bibinfo  {journal} {Communications in Mathematical Physics}\ }\textbf {\bibinfo {volume} {370}},\ \bibinfo {pages} {117} (\bibinfo {year} {2019})}\BibitemShut {NoStop}%
\bibitem [{\citenamefont {Wiebe}\ and\ \citenamefont {Wossnig}(2019)}]{wiebe2019generative}%
  \BibitemOpen
  \bibfield  {author} {\bibinfo {author} {\bibfnamefont {N.}~\bibnamefont {Wiebe}}\ and\ \bibinfo {author} {\bibfnamefont {L.}~\bibnamefont {Wossnig}},\ }\href {https://arxiv.org/abs/1905.09902} {\bibinfo {title} {Generative training of quantum {B}oltzmann machines with hidden units}} (\bibinfo {year} {2019}),\ \Eprint {https://arxiv.org/abs/1905.09902} {arXiv:1905.09902 [quant-ph]} \BibitemShut {NoStop}%
\bibitem [{\citenamefont {Anschuetz}\ and\ \citenamefont {Cao}(2019)}]{anschuetz2019realizing}%
  \BibitemOpen
  \bibfield  {author} {\bibinfo {author} {\bibfnamefont {E.~R.}\ \bibnamefont {Anschuetz}}\ and\ \bibinfo {author} {\bibfnamefont {Y.}~\bibnamefont {Cao}},\ }\href {https://arxiv.org/abs/1903.01359} {\bibinfo {title} {Realizing quantum {B}oltzmann machines through eigenstate thermalization}} (\bibinfo {year} {2019}),\ \Eprint {https://arxiv.org/abs/1903.01359} {arXiv:1903.01359 [quant-ph]} \BibitemShut {NoStop}%
\bibitem [{\citenamefont {Kappen}(2020)}]{Kappen_2020}%
  \BibitemOpen
  \bibfield  {author} {\bibinfo {author} {\bibfnamefont {H.~J.}\ \bibnamefont {Kappen}},\ }\bibfield  {title} {\bibinfo {title} {Learning quantum models from quantum or classical data},\ }\href {https://doi.org/10.1088/1751-8121/ab7df6} {\bibfield  {journal} {\bibinfo  {journal} {Journal of Physics A: Mathematical and Theoretical}\ }\textbf {\bibinfo {volume} {53}},\ \bibinfo {pages} {214001} (\bibinfo {year} {2020})}\BibitemShut {NoStop}%
\bibitem [{\citenamefont {Zoufal}\ \emph {et~al.}(2021)\citenamefont {Zoufal}, \citenamefont {Lucchi},\ and\ \citenamefont {Woerner}}]{Zoufal2021}%
  \BibitemOpen
  \bibfield  {author} {\bibinfo {author} {\bibfnamefont {C.}~\bibnamefont {Zoufal}}, \bibinfo {author} {\bibfnamefont {A.}~\bibnamefont {Lucchi}},\ and\ \bibinfo {author} {\bibfnamefont {S.}~\bibnamefont {Woerner}},\ }\bibfield  {title} {\bibinfo {title} {Variational quantum {B}oltzmann machines},\ }\href {https://doi.org/10.1007/s42484-020-00033-7} {\bibfield  {journal} {\bibinfo  {journal} {Quantum Machine Intelligence}\ }\textbf {\bibinfo {volume} {3}},\ \bibinfo {pages} {7} (\bibinfo {year} {2021})}\BibitemShut {NoStop}%
\bibitem [{\citenamefont {Khaled}\ and\ \citenamefont {Richt{\'a}rik}(2020)}]{Khaled2020}%
  \BibitemOpen
  \bibfield  {author} {\bibinfo {author} {\bibfnamefont {A.}~\bibnamefont {Khaled}}\ and\ \bibinfo {author} {\bibfnamefont {P.}~\bibnamefont {Richt{\'a}rik}},\ }\href {https://doi.org/10.48550/ARXIV.2002.03329} {\bibinfo {title} {Better theory for {SGD} in the nonconvex world}} (\bibinfo {year} {2020})\BibitemShut {NoStop}%
\bibitem [{\citenamefont {Harrow}\ and\ \citenamefont {Napp}(2021)}]{Harrow2021}%
  \BibitemOpen
  \bibfield  {author} {\bibinfo {author} {\bibfnamefont {A.~W.}\ \bibnamefont {Harrow}}\ and\ \bibinfo {author} {\bibfnamefont {J.~C.}\ \bibnamefont {Napp}},\ }\bibfield  {title} {\bibinfo {title} {Low-depth gradient measurements can improve convergence in variational hybrid quantum-classical algorithms},\ }\href {https://doi.org/10.1103/PhysRevLett.126.140502} {\bibfield  {journal} {\bibinfo  {journal} {Physical Review Letters}\ }\textbf {\bibinfo {volume} {126}},\ \bibinfo {pages} {140502} (\bibinfo {year} {2021})}\BibitemShut {NoStop}%
\bibitem [{\citenamefont {Patel}\ \emph {et~al.}(2024)\citenamefont {Patel}, \citenamefont {Coles},\ and\ \citenamefont {Wilde}}]{Patel2024}%
  \BibitemOpen
  \bibfield  {author} {\bibinfo {author} {\bibfnamefont {D.}~\bibnamefont {Patel}}, \bibinfo {author} {\bibfnamefont {P.~J.}\ \bibnamefont {Coles}},\ and\ \bibinfo {author} {\bibfnamefont {M.~M.}\ \bibnamefont {Wilde}},\ }\bibfield  {title} {\bibinfo {title} {Variational quantum algorithms for semidefinite programming},\ }\href {https://doi.org/10.22331/q-2024-06-17-1374} {\bibfield  {journal} {\bibinfo  {journal} {Quantum}\ }\textbf {\bibinfo {volume} {8}},\ \bibinfo {pages} {1374} (\bibinfo {year} {2024})}\BibitemShut {NoStop}%
\bibitem [{\citenamefont {Lloyd}(1996)}]{lloyd1996universal}%
  \BibitemOpen
  \bibfield  {author} {\bibinfo {author} {\bibfnamefont {S.}~\bibnamefont {Lloyd}},\ }\bibfield  {title} {\bibinfo {title} {Universal quantum simulators},\ }\href {https://doi.org/10.1126/science.273.5278.1073} {\bibfield  {journal} {\bibinfo  {journal} {Science}\ }\textbf {\bibinfo {volume} {273}},\ \bibinfo {pages} {1073} (\bibinfo {year} {1996})}\BibitemShut {NoStop}%
\bibitem [{\citenamefont {Childs}\ \emph {et~al.}(2018)\citenamefont {Childs}, \citenamefont {Maslov}, \citenamefont {Nam}, \citenamefont {Ross},\ and\ \citenamefont {Su}}]{childs2018toward}%
  \BibitemOpen
  \bibfield  {author} {\bibinfo {author} {\bibfnamefont {A.~M.}\ \bibnamefont {Childs}}, \bibinfo {author} {\bibfnamefont {D.}~\bibnamefont {Maslov}}, \bibinfo {author} {\bibfnamefont {Y.}~\bibnamefont {Nam}}, \bibinfo {author} {\bibfnamefont {N.~J.}\ \bibnamefont {Ross}},\ and\ \bibinfo {author} {\bibfnamefont {Y.}~\bibnamefont {Su}},\ }\bibfield  {title} {\bibinfo {title} {Toward the first quantum simulation with quantum speedup},\ }\href {https://doi.org/10.1073/pnas.1801723115} {\bibfield  {journal} {\bibinfo  {journal} {Proceedings of the National Academy of Sciences}\ }\textbf {\bibinfo {volume} {115}},\ \bibinfo {pages} {9456} (\bibinfo {year} {2018})}\BibitemShut {NoStop}%
\bibitem [{\citenamefont {Cleve}\ \emph {et~al.}(1998)\citenamefont {Cleve}, \citenamefont {Ekert}, \citenamefont {Macchiavello},\ and\ \citenamefont {Mosca}}]{Cleve1998}%
  \BibitemOpen
  \bibfield  {author} {\bibinfo {author} {\bibfnamefont {R.}~\bibnamefont {Cleve}}, \bibinfo {author} {\bibfnamefont {A.}~\bibnamefont {Ekert}}, \bibinfo {author} {\bibfnamefont {C.}~\bibnamefont {Macchiavello}},\ and\ \bibinfo {author} {\bibfnamefont {M.}~\bibnamefont {Mosca}},\ }\bibfield  {title} {\bibinfo {title} {Quantum algorithms revisited},\ }\href {https://doi.org/10.1098/rspa.1998.0164} {\bibfield  {journal} {\bibinfo  {journal} {Proceedings of the Royal Society A}\ }\textbf {\bibinfo {volume} {454}},\ \bibinfo {pages} {339} (\bibinfo {year} {1998})}\BibitemShut {NoStop}%
\bibitem [{\citenamefont {Low}\ and\ \citenamefont {Chuang}(2019)}]{Low2019hamiltonian}%
  \BibitemOpen
  \bibfield  {author} {\bibinfo {author} {\bibfnamefont {G.~H.}\ \bibnamefont {Low}}\ and\ \bibinfo {author} {\bibfnamefont {I.~L.}\ \bibnamefont {Chuang}},\ }\bibfield  {title} {\bibinfo {title} {Hamiltonian simulation by qubitization},\ }\href {https://doi.org/10.22331/q-2019-07-12-163} {\bibfield  {journal} {\bibinfo  {journal} {{Quantum}}\ }\textbf {\bibinfo {volume} {3}},\ \bibinfo {pages} {163} (\bibinfo {year} {2019})}\BibitemShut {NoStop}%
\bibitem [{\citenamefont {Gily\'{e}n}\ \emph {et~al.}(2019)\citenamefont {Gily\'{e}n}, \citenamefont {Su}, \citenamefont {Low},\ and\ \citenamefont {Wiebe}}]{Gilyen2019}%
  \BibitemOpen
  \bibfield  {author} {\bibinfo {author} {\bibfnamefont {A.}~\bibnamefont {Gily\'{e}n}}, \bibinfo {author} {\bibfnamefont {Y.}~\bibnamefont {Su}}, \bibinfo {author} {\bibfnamefont {G.~H.}\ \bibnamefont {Low}},\ and\ \bibinfo {author} {\bibfnamefont {N.}~\bibnamefont {Wiebe}},\ }\bibfield  {title} {\bibinfo {title} {Quantum singular value transformation and beyond: exponential improvements for quantum matrix arithmetics},\ }in\ \href {https://doi.org/10.1145/3313276.3316366} {\emph {\bibinfo {booktitle} {Proceedings of the 51st Annual ACM SIGACT Symposium on Theory of Computing}}},\ \bibinfo {series and number} {STOC 2019}\ (\bibinfo  {publisher} {Association for Computing Machinery},\ \bibinfo {address} {New York, NY, USA},\ \bibinfo {year} {2019})\ pp.\ \bibinfo {pages} {193--204}\BibitemShut {NoStop}%
\bibitem [{\citenamefont {Chen}\ \emph {et~al.}(2023{\natexlab{c}})\citenamefont {Chen}, \citenamefont {Westerheim}, \citenamefont {Holmes}, \citenamefont {Luo}, \citenamefont {Nuradha}, \citenamefont {Patel}, \citenamefont {Rethinasamy}, \citenamefont {Wang},\ and\ \citenamefont {Wilde}}]{chen2023qslackslackvariableapproachvariational}%
  \BibitemOpen
  \bibfield  {author} {\bibinfo {author} {\bibfnamefont {J.}~\bibnamefont {Chen}}, \bibinfo {author} {\bibfnamefont {H.}~\bibnamefont {Westerheim}}, \bibinfo {author} {\bibfnamefont {Z.}~\bibnamefont {Holmes}}, \bibinfo {author} {\bibfnamefont {I.}~\bibnamefont {Luo}}, \bibinfo {author} {\bibfnamefont {T.}~\bibnamefont {Nuradha}}, \bibinfo {author} {\bibfnamefont {D.}~\bibnamefont {Patel}}, \bibinfo {author} {\bibfnamefont {S.}~\bibnamefont {Rethinasamy}}, \bibinfo {author} {\bibfnamefont {K.}~\bibnamefont {Wang}},\ and\ \bibinfo {author} {\bibfnamefont {M.~M.}\ \bibnamefont {Wilde}},\ }\href {https://arxiv.org/abs/2312.03830} {\bibinfo {title} {Qslack: A slack-variable approach for variational quantum semi-definite programming}} (\bibinfo {year} {2023}{\natexlab{c}}),\ \Eprint {https://arxiv.org/abs/2312.03830} {arXiv:2312.03830 [quant-ph]} \BibitemShut {NoStop}%
\bibitem [{\citenamefont {Goldfeld}\ \emph {et~al.}(2024)\citenamefont {Goldfeld}, \citenamefont {Patel}, \citenamefont {Sreekumar},\ and\ \citenamefont {Wilde}}]{Goldfeld2024}%
  \BibitemOpen
  \bibfield  {author} {\bibinfo {author} {\bibfnamefont {Z.}~\bibnamefont {Goldfeld}}, \bibinfo {author} {\bibfnamefont {D.}~\bibnamefont {Patel}}, \bibinfo {author} {\bibfnamefont {S.}~\bibnamefont {Sreekumar}},\ and\ \bibinfo {author} {\bibfnamefont {M.~M.}\ \bibnamefont {Wilde}},\ }\bibfield  {title} {\bibinfo {title} {Quantum neural estimation of entropies},\ }\href {https://doi.org/10.1103/PhysRevA.109.032431} {\bibfield  {journal} {\bibinfo  {journal} {Physical Review A}\ }\textbf {\bibinfo {volume} {109}},\ \bibinfo {pages} {032431} (\bibinfo {year} {2024})}\BibitemShut {NoStop}%
\bibitem [{\citenamefont {NISO}()}]{NISO}%
  \BibitemOpen
  \bibfield  {author} {\bibinfo {author} {\bibnamefont {NISO}},\ }\href@noop {} {\bibinfo {title} {Credit – contributor roles taxonomy}},\ \bibinfo {note} {\url{https://credit.niso.org/}, Accessed 2024-09-27}\BibitemShut {NoStop}%
\bibitem [{\citenamefont {Danilova}\ \emph {et~al.}(2022)\citenamefont {Danilova}, \citenamefont {Dvurechensky}, \citenamefont {Gasnikov}, \citenamefont {Gorbunov}, \citenamefont {Guminov}, \citenamefont {Kamzolov},\ and\ \citenamefont {Shibaev}}]{Danilova2022}%
  \BibitemOpen
  \bibfield  {author} {\bibinfo {author} {\bibfnamefont {M.}~\bibnamefont {Danilova}}, \bibinfo {author} {\bibfnamefont {P.}~\bibnamefont {Dvurechensky}}, \bibinfo {author} {\bibfnamefont {A.}~\bibnamefont {Gasnikov}}, \bibinfo {author} {\bibfnamefont {E.}~\bibnamefont {Gorbunov}}, \bibinfo {author} {\bibfnamefont {S.}~\bibnamefont {Guminov}}, \bibinfo {author} {\bibfnamefont {D.}~\bibnamefont {Kamzolov}},\ and\ \bibinfo {author} {\bibfnamefont {I.}~\bibnamefont {Shibaev}},\ }\bibinfo {title} {Recent theoretical advances in non-convex optimization},\ in\ \href {https://doi.org/10.1007/978-3-031-00832-0_3} {\emph {\bibinfo {booktitle} {Springer Optimization and Its Applications}}},\ Vol.\ \bibinfo {volume} {191}\ (\bibinfo  {publisher} {Springer International Publishing},\ \bibinfo {year} {2022})\ pp.\ \bibinfo {pages} {79--163},\ \Eprint {https://arxiv.org/abs/2012.06188} {arXiv:2012.06188} \BibitemShut {NoStop}%
\bibitem [{\citenamefont {{Wikipedia contributors}}(2024)}]{enwiki:1249353484}%
  \BibitemOpen
  \bibfield  {author} {\bibinfo {author} {\bibnamefont {{Wikipedia contributors}}},\ }\href {https://en.wikipedia.org/w/index.php?title=Fourier_transform&oldid=1249353484} {\bibinfo {title} {Fourier transform --- {Wikipedia}{,} the free encyclopedia}} (\bibinfo {year} {2024}),\ \bibinfo {note} {[Online; accessed 5-October-2024] \url{https://en.wikipedia.org/w/index.php?title=Fourier_transform&oldid=1249353484}}\BibitemShut {NoStop}%
\end{thebibliography}%

\let\addcontentsline\oldaddcontentsline

\onecolumngrid

\newpage\newgeometry{margin=1in}
\appendix
\thispagestyle{empty}

\large

\section{Preliminaries}\label{sec:preliminaries}

In this section, we introduce our notation and present some definitions and standard results from the optimization literature, which we use later in this document. Specifically, we revisit one of the known convergence results~\cite[Corollary~1]{Khaled2020} associated with the stochastic gradient descent algorithm.

\subsection{Notation}\label{sec:notation}

Let $\mathbb{R}, \mathbb{R}_{\geq 0}$, $\mathbb{N}$, and $\mathbb{C}$ denote the set of real, non-negative real, natural, and complex numbers, respectively. We use the notation $[M]$ to denote the set $\{1, 2, \ldots, M\}$. Let $\mathcal{H}$ denote a $2^n$-dimensional Hilbert space associated with a quantum system of $n$ qubits. 
We denote the set of quantum states acting on~$\mathcal{H}$ by~$\mathcal{D}$. Let $\operatorname{Tr}[X]$ denote the trace of a matrix $X$, i.e., the sum of its diagonal elements. Also, let $X^{\dagger}$ denote the Hermitian conjugate (or
adjoint) of the matrix $X$. The Schatten $p$-norm of a matrix $X$ is defined for $p \in [1, \infty)$ as follows:
\begin{equation}
    \left\Vert X\right\Vert_{p} \coloneqq \left(\operatorname{Tr}\!\left[\left(X^{\dagger}X\right)^{\frac{p}{2}}\right]\right)^{\frac{1}{p}}.\label{eq:schatten-norm-def}
\end{equation}
For our purposes, we use Schatten norms with $p =1$ (also called trace norm), $p=2$ (Hilbert--Schmidt norm), and $p=\infty$ (operator norm). Note that the operator norm of a matrix corresponds to its maximum singular value. For notational convenience, we omit the subscript `$\infty$' when referring to the operator norm. Additionally, $\left \Vert x \right \Vert$ denotes the $\ell_2$ norm of a vector $x$.
 Moreover, let $\{X, Y\} \coloneqq XY + YX$ denote the anti-commutator of the matrices $X$ and~$Y$. For a multivariate function $f \colon \mathbb{R}^{n} \rightarrow \mathbb{R}$, we use $\nabla f$ and $\nabla^{2} f$ to denote its gradient and Hessian, respectively. Let $\sfrac{\partial f (\cdot)}{\partial x_{i}}$ denote the partial derivative of $f$ with respect to the $i$th component of the vector $x$. For brevity, we use the notation $\partial_{i}f(\cdot)\equiv\sfrac{\partial f(\cdot)}{\partial x_{i}}$ throughout the appendices. The $(i, j)$-th element of the Hessian is then denoted by $\partial_i\partial_{j}f(\cdot)$. Finally, we use the notation $O(\cdot)$ for hiding constants that do not depend on any problem parameter.

\subsection{Definitions and standard results}\label{sec:def_stand_res}

We now review some definitions and known results related to the Lipschitz continuity and smoothness of a function.

\begin{definition}[Lipschitz Continuity]\label{eq:lipcongen}
A function $f\colon  \mathbb{R}^{n} \rightarrow \mathbb{R}^{m}$ is $L$-Lipschitz continuous if there exists a non-negative real number $L$ such that, for all $ x, x' \in  \mathbb{R}^{n}$, the following holds:
\begin{equation}
    \left\Vert f(x) - f(x') \right\Vert \leq L \left\Vert x - x' \right\Vert.\label{eq:lip-cond}
\end{equation} 
We say that $L$ is a Lipschitz constant for~$f$.
\end{definition}

\begin{definition}[Smoothness]\label{eq:smoothness_def}
A function $f \colon \mathbb{R}^{n} \rightarrow \mathbb{R}^{m}$ is $\ell$-smooth if its gradient is $\ell$-Lipschitz continuous. In other words, there exists a non-negative real number $\ell$ such that,  for all $ x, x' \in  \mathbb{R}^{n}$, the following holds:
\begin{align}
    \left\Vert \nabla f(x) - \nabla f(x') \right\Vert \leq \ell \left\Vert x - x' \right\Vert.
\end{align}
\end{definition}

We now recall some known results related to the Lipschitz continuity of a function. These results also directly apply to the smoothness of a function because, as per the definition above, if the gradient of the function is Lipschitz continuous, then the function is smooth. Having said that, we begin with a simple case in which the function $f \colon  \mathbb{R} \rightarrow  \mathbb{R}$ is univariate and differentiable. Now, in order to prove that this function is Lipschitz continuous, one approach is to show that it satisfies the condition given by~\eqref{eq:lip-cond} for some non-negative constant $L$. This gives the Lipschitz constant~$L$ for this function. Alternatively, we can show Lipschitz continuity by bounding the gradient of $f$ from above. This is a direct consequence of the mean value theorem. 
More formally, if there exists a non-negative real number $L$ such that the absolute value of its gradient is bounded from above by $L$, that is, the following holds for all $x \in \mathbb{R}$
\begin{equation}
    \left|\frac{df(x)}{dx} \right| \leq L,
\end{equation}
then $L$ is a Lipschitz constant for $f$. Using this simple case of a univariate function as a base case, the following lemmas provide Lipschitz constants for multivariate and multivariate vector-valued functions.

\begin{lemma}[Lipschitz Constant for a Multivariate Function]\label{lemma:lipmul}
Let $f \colon \mathbb{R}^{n} \rightarrow \mathbb{R}$
be a differentiable multivariate function with bounded partial derivatives. Then the value 
\begin{equation}
    L = \sqrt{n} \max_{i \in [n]} \left \{\sup_{x} \left| \frac{\partial f(x)}{\partial x_{i}} \right| \right \}
\end{equation}
is a Lipschitz constant for $f$.
\end{lemma}

\begin{proof}
See \cite[Appendix~A.1]{Patel2024}. 
\end{proof}

\begin{lemma}[Lipschitz Constant for a Multivariate Vector-Valued Function]\label{lemma:lipvec}
Let $f \colon \mathbb{R}^{n} \rightarrow \mathbb{R}^{m}$
be a differentiable multivariate vector-valued function such that each of its components, $f_{i}$, is $L_{i}$-Lipschitz continuous. Then 
\begin{equation}
    L = \left(\sum_{i=1}^m L_{i}^2 \right)^{\frac{1}{2}} = n^{\frac{1}{2}} \left(\sum_{i=1}^m \max_{j \in [n]} \left \{\sup_{x} \left| \frac{\partial f_i(x)}{\partial x_{j}} \right|^2 \right \} \right)^{\frac{1}{2}}
\end{equation}
is a Lipschitz constant for $f$.
\end{lemma}

\begin{proof}
For the proof of the first equality of the lemma statement, see \cite[Appendix~A.2]{Patel2024}. The second equality then directly follows from Lemma~\ref{lemma:lipmul}.
\end{proof}

\medskip

The objective function that we deal with in this paper is non-convex in general. In optimization theory, it is well known that finding a globally optimal point of a non-convex function is generally intractable (NP-hard)~\cite[Section~2.1]{Danilova2022}. Therefore, an important question arises about the types of solutions that can be guaranteed in such a scenario. In optimization theory, when dealing with non-convex objective functions, the notion of $\varepsilon$-stationary points is often considered.  Intuitively, a point is $\varepsilon$-stationary if the norm of the gradient of the function at that point is very small. Formally, we define an $\varepsilon$-stationary point as follows:
\begin{definition}[$\varepsilon$-Stationary Point]\label{def:estationary}
Let $f \colon \mathbb{R}^{n} \rightarrow \mathbb{R}^{m}$ be a differentiable function, and let $\varepsilon \geq 0$. A point $x \in \mathbb{R}^{n}$ is an $\varepsilon$-stationary point of  $f$ if $\left \Vert \nabla f(x) \right \Vert \leq \varepsilon $. \end{definition}

We conclude by recalling Hoeffding’s inequality, which we employ later to examine the sample complexity of our algorithm.

\begin{lemma}[Hoeffding’s Inequality]
    Suppose that $X_1, \ldots, X_n$ are $n$ independent random variables, and that there exist $a_i, b_i \in \mathbb{R}$ such that $a_i\leq X_i \leq b_i$ for all $i \in [n]$. Then, for all $\varepsilon \geq 0$, we have that
    \begin{align}
        \Pr\!\left(\left | \overline{X} - \mathbb{E}[\overline{X}]\right| \geq \varepsilon\right) \leq 2 \exp\!\left (- \frac{2n^2\varepsilon^2}{\sum_{i=1}^n(b_i - a_i)^2}\right),
    \end{align}
    where $\overline{X} \coloneqq \frac{1}{n}\sum_{i=1}^n X_i$ and $\mathbb{E}[\overline{X}]$ is the expected value of $\overline{X}$.\label{lem:hoeffding}
\end{lemma}

\subsection{Stochastic gradient descent}

\label{sec:sgd}

Consider the following minimization problem:
\begin{align}
    f^* \coloneqq \inf_{x\in \mathbb{R}^n} f(x),
\end{align}
where $f \colon  \mathbb{R}^n \rightarrow \mathbb{R}$ is an $\ell$-smooth function and $f^*$ is the global minimum. The stochastic gradient descent (SGD) algorithm uses the following rule to update the iterate:
\begin{align}
    x_{m+1} =  x_{m} - \eta \overline{g}(x_{m}),\label{eq:update_rule_sgd} 
\end{align}
where $\overline{g}(x)$ is a stochastic gradient, evaluated at some point $x$, and $\eta > 0$ is the learning rate parameter. Furthermore, the SGD algorithm requires the stochastic gradient $\overline{g}(x)$ to be unbiased, i.e., $\mathbb{E}[ \overline{g}(x)] = \nabla f(x)$, for all $x \in \mathbb{R}^n$. Here, the expectation $\mathbb{E}[\cdot]$ is with respect to the randomness inherent in $\overline{g}(x)$.  In addition, for all $x \in \mathbb{R}^n$, $\overline{g}(x)$ should also satisfy the following condition: there exist constants $A, B, C \geq 0$ such that 
\begin{align}\label{eq:norm_grad_exp}
    \mathbb{E}\!\left[\left \Vert \overline{g}(x)\right\Vert^2\right] \leq 2A(f(x) - f^*) + B \left \Vert \nabla f(x)\right \Vert^2 + C.
\end{align}

To this end, the following lemma demonstrates the rate at which the SGD algorithm converges to an $\varepsilon$-stationary point of $f$. This lemma is a restatement of~\cite[Corollary~1]{Khaled2020}, and we include it here for completeness.

\begin{lemma}[SGD Convergence]\label{lem:sgd_conv} Let $M$ be the total number of iterations of the SGD algorithm with update rule given by~\eqref{eq:update_rule_sgd}. Also, let $\eta \coloneqq \min\left \{ \frac{1}{\sqrt{\ell A M}},  \frac{1}{\ell B},  \frac{\varepsilon}{2 \ell C}\right\}$ and $\Delta \coloneqq f(x_0) - f^*$. Then provided that
\begin{align}
    M \geq \frac{12  \ell \Delta}{\varepsilon^2}\max\!\left \{ B, \frac{12  A \Delta}{\varepsilon^2}, \frac{2 C}{\varepsilon^2}\right\},\label{eq:M_total_iterations}
\end{align}
 the SGD algorithm converges in such a way that
\begin{align}
    \min_{1\leq m\leq M} \mathbb{E}\!\left[\left \Vert \nabla f(x_m) \right \Vert\right] \leq \varepsilon,
\end{align}
where the expectation $\mathbb{E}[\cdot]$ is over the randomness of the SGD algorithm.
\end{lemma}

\section{Justification for main results}

\label{sec:main_results}

\subsection{Organization}

\label{sec:organization}

The rest of the appendices are organized as follows. In Section~\ref{sec:prob_setup}, we begin by formally defining the quantum Boltzmann machine (QBM) learning problem of ground-state energies (Definition~\ref{def:qbm_learn_def}). The goal of this learning problem is to find an $\varepsilon$-stationary point of a particular optimization problem that we define later in~\eqref{eq:qbm_learn_def}. To accomplish this, we employ SGD;  for more details on SGD, refer to Section~\ref{sec:sgd}. Motivated by the fact that SGD is a first-order optimization algorithm,  in Section~\ref{sec:grad}, we derive an analytical expression for the gradient of the objective function and further show that this gradient is bounded. Then, in Section~\ref{sec:hessian}, we derive analytical expressions for the matrix elements of the Hessian of our objective function and show that these elements are also bounded. This property of the Hessian is important for establishing that our objective function is smooth, which we analyze in Section~\ref{sec:smoothness}. Furthermore, this smoothness property then ensures convergence of SGD to an $\varepsilon$-stationary point. In Section~\ref{sec:qbge}, we present a quantum algorithm that estimates the gradient by returning an unbiased estimator of it. We refer to this algorithm as the quantum Boltzmann gradient estimator (QBGE) in what follows. Then, in Algorithm~\ref{algo:qbm-gse}, we present the full SGD-based algorithm for QBM learning of ground-state energies, which employs QBGE for computing the stochastic gradient at each iteration. We refer to this algorithm as QBM-GSE. Finally, in Section~\ref{sec:sample_comp}, we analyze the sample complexity of the QBM-GSE algorithm.

\subsection{Problem setup}

\label{sec:prob_setup}

Here we formally present the QBM learning problem that we introduced in the main text. Let $H \in \mathbb{C}^{2^n \times 2^n}$ be the Hamiltonian of interest, acting on the space of $n$ qubits, and suppose that it is given in the following form:
\begin{equation}\label{eq:def_h}
H\coloneqq\sum_{k=1}^{K}\alpha_{k}H_{k},
\end{equation}
where, for all $k \in [K]$, the coefficient $\alpha_{k} \in \mathbb{R}$ and $H_{k}$ is a local Hamiltonian that acts on a constant number of qubits.
Note that, for all $k \in [K]$, the parameter $\alpha_{k}$ in general can be any real number, but we can absorb any negative sign of $\alpha_k$ into $H_{k}$. Therefore, without loss of generality, we assume that $\alpha_{k} > 0$, for all~$k \in [K]$. Furthermore, we assume that $\left\Vert H_k \right\Vert \leq 1$, for all $k \in [K]$. We do this without loss of generality by absorbing the norm of $H_k$ into $\alpha_k$, so that
\begin{align}
    \left\Vert H\right\Vert \leq\sum_{k=1}^K\left\vert \alpha_{k}\right\vert
\left\Vert H_{k}\right\Vert \leq\sum_{k=1}^K \left\vert \alpha_{k}\right\vert
\eqqcolon \left\Vert \alpha\right\Vert _{1}.
\end{align}
For our purposes, we assume that $\left\Vert \alpha\right\Vert _{1}$ is polynomial in $n$, i.e., $\left\Vert \alpha\right\Vert _{1} = O(\operatorname{poly}(n))$). We also assume that $H$ can be measured efficiently on a quantum computer. This assumption is reasonable because  physically relevant Hamiltonians consist of only $O(\operatorname{poly}(n))$  summands in their linear combinations (see~\eqref{eq:def_h}), and thus they can be efficiently measured on a quantum computer. Having said that, we now state the above assumptions more formally below.

\begin{assumption}\label{assump:hamil}
The Hamiltonian $H$, defined in~\eqref{eq:def_h}, can be efficiently measured on a quantum computer, and $\left\Vert \alpha\right\Vert _{1} = O(\operatorname{poly}(n))$.
\end{assumption}

We now present the problem of determining the ground-state energy of $H$. This problem can be formulated as an optimization problem as follows:
\begin{equation}
    \inf_{\rho \in \mathcal{D}}\operatorname{Tr}
[H\rho],\label{eq:gr-state-op}
\end{equation}
where $\mathcal{D}$ represents the set of all possible quantum states acting
on the same Hilbert space on which $H$ acts. As previously mentioned in the main text, the problem above is generally difficult to solve due to the large search space of size $2^n \times 2^n$, which is exponential in the number of qubits. Despite this computational difficulty, various approaches have been proposed (see main text for references). These approaches typically involve reducing the search space by making an informed guess and then parameterizing this reduced search space. Following this reduction, these approaches utilize classical optimization techniques, such as gradient descent, to find the optimal value within this reduced search space.

Here, we parameterize this search space in the following way. Consider a parameterized QBM Hamiltonian
\begin{equation}
G(\theta)\coloneqq\sum_{j=1}^{J}\theta_{j}G_{j},\label{eq:para-qbm-hamil}
\end{equation}
where, for all $j \in [J]$, $\theta_{j}\in\mathbb{R}$ is a tunable parameter and $G_{j}$ is a local
Hamiltonian that acts on a constant number of qubits. Additionally, $\theta\coloneqq(\theta_{1},\ldots,\theta_{J})$ is a parameter vector. This parameterized Hamiltonian $G(\theta)$ further defines the following parameterized thermal state:
\begin{equation}
    \rho(\theta)  \coloneqq\frac{e^{-G(\theta)}}{Z(\theta)},\label{eq:para-state}
\end{equation}
where $Z(\theta) \coloneqq\operatorname{Tr}[e^{-G(\theta)}] $ is the partition function.
Consequently, using this parameterization, we can rewrite the original optimization problem, given by~\eqref{eq:gr-state-op}, in the following way:
\begin{equation}\label{eq:obj_func}
\inf_{\theta\in
\mathbb{R}^{J}}\operatorname{Tr}
[H\rho(\theta)].
\end{equation}
Notably, the original optimization problem (defined in~\eqref{eq:gr-state-op}) and the above parameterized problem are equivalent because any quantum state can be expressed as a thermal state with a suitable Hamiltonian. This means that the dimension $J$ of the parameterized search space is exponential in the number of qubits since the dimension of the original search space (the set of quantum states,~$\mathcal{D}$) is also exponential in the number of qubits. This is evident from a simple counting argument. To address this computational difficulty, we assume some knowledge about the problem structure, allowing us to reduce the parameterized search space such that $J$ is $O(\operatorname{poly}(n))$.

Due to this reduction in the search space, the following inequality is a basic consequence of the variational principle:
\begin{equation}
    \inf_{\rho\in\mathcal{D}}\operatorname{Tr}[H\rho] \leq\inf_{\theta\in
\mathbb{R}^{J}} \operatorname{Tr}[H\rho(\theta)],
\end{equation}
The above inequality indicates that the true ground-state energy of $H$ is bounded from above by the minimal energy of the Hamiltonian $H$ over every possible trial state $\rho(\theta)$.

Moreover, due to this reduction, the function $f(\theta)$ is now a non-convex function. Indeed, one can check that even the following basic instance of $f(\theta)$ is non-convex (refer to Figure~\ref{fig:example}): 
\begin{equation}
    (\theta_1, \theta_2) \mapsto \operatorname{Tr}\!\left[\sigma_Y \frac{e^{-G(\theta_1,\theta_2)}}{\operatorname{Tr}[e^{-G(\theta_1,\theta_2)}]}\right]\label{eq:example}
\end{equation}
with $G(\theta_1, \theta_2) =\theta_1 \sigma_X + \theta_2 \sigma_Y$. As mentioned previously in Section~\ref{sec:def_stand_res}, it is well known in optimization theory that finding a globally optimal point of a non-convex function is generally NP-hard~\cite[Section~2.1]{Danilova2022}. Consequently, optimizing $\operatorname{Tr}[H\rho(\theta)]$ to find its global optimal point is challenging. Therefore, an important question arises about the types of solutions that can be guaranteed in such a scenario. In optimization theory, when dealing with non-convex objective functions, the notion of $\varepsilon$-stationary points (Definition~\ref{def:estationary}) is often considered. Given the difficulty of finding the globally optimal point, in this paper, we focus instead on finding an $\varepsilon$-stationary point of $\operatorname{Tr}[H\rho(\theta)]$.

\begin{figure}
    \centering
    \includegraphics[scale=0.4]{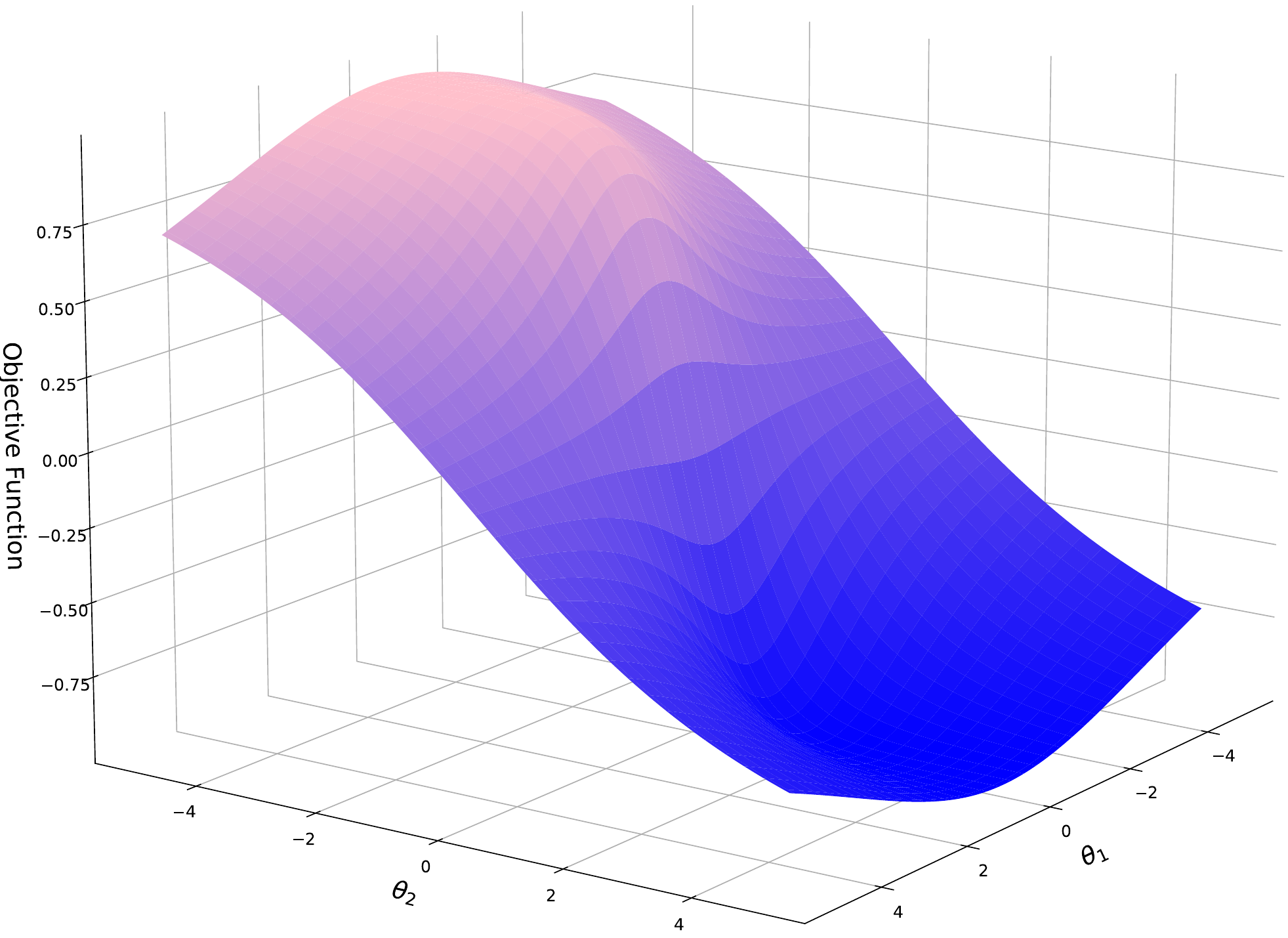}
    \caption{The non-convex landscape of the objective function given by~\eqref{eq:example}. 
    }
    \label{fig:example}
\end{figure}

We measure the cost of our algorithm by the number of samples, $N$, of parameterized thermal states (as defined in~\eqref{eq:para-qbm-hamil}) required to find an $\varepsilon$-stationary point of $\operatorname{Tr}[H\rho(\theta)]$. We refer to this metric as the sample complexity of our algorithm, which we analyze in more detail in Section~\ref{sec:sample_comp}. 

With the above notions in place, we now formally define our problem as follows:

\begin{definition}[QBM Learning of Ground-State Energy]\label{def:qbm_learn_def} Let $n$ be the number of qubits. Given a Hamiltonian $H \in \mathbb{C}^{2^n \times 2^n}$ as defined in~\eqref{eq:def_h} such that it satisfies Assumption~\ref{assump:hamil} and a positive integer $J$ such that $J \in O(\operatorname{poly}(n))$,
consider the following optimization problem:
\begin{align}
     \inf_{\theta\in
\mathbb{R}^{J}} \operatorname{Tr}[H\rho(\theta)],\label{eq:qbm_learn_def}
\end{align}
where $\rho(\theta)$ is a parameterized thermal state as defined in~\eqref{eq:para-qbm-hamil}. Then the goal of QBM learning of ground-state energy is to find an $\varepsilon$-stationary point of the above optimization problem, given access to multiple copies of $\rho(\theta)$, for all $\theta \in \mathbb{R}^J$.
\end{definition}

\subsection{Gradient}

\label{sec:grad}

We employ SGD (see Section~\ref{sec:sgd} for more details on SGD) to find an $\varepsilon$-stationary point of $\operatorname{Tr}[H\rho(\theta)]$ (refer to Definition~\ref{def:qbm_learn_def} for a more formal definition of the problem). Therefore, our first step is to derive an analytical expression for the gradient, $\nabla_\theta \operatorname{Tr}[H\rho(\theta)]$. Note that this gradient is a multivariate vector-valued function with partial derivatives as its components, i.e.,
\begin{equation}
    \left\{\frac{\partial \operatorname{Tr}[H\rho(\theta)]}{\partial\theta_{j}}\right\}_{j=1}^J.
\end{equation}
Recall that for brevity, we use the notation $\partial_{j}\equiv\frac{\partial}{\partial\theta_{j}}$. Since $\rho(\theta)$ is a thermal state, for all $j \in [J]$, its partial derivative, $\partial_{j}\rho(\theta)$, involves the partial derivative of the matrix exponential $e^{-G(\theta)}$:
\begin{align}
    \partial_{j}e^{-G(\theta)}.
\end{align}
Therefore, we first focus on deriving an explicit expression for the above quantity, which we do in the following lemma. Note that this derivation is not original, and it follows from the developments in~\cite{Hastings2007}, \cite[Appendix~B]{Anshu2021}, and \cite[Lemma~5]{Coopmans2024} (see also \cite[Section~III-C]{Kim2012} and \cite[Section~IV-A]{Kato2019}). We include it here for completeness.

\begin{lemma}[Partial Derivatives of the Matrix Exponential \cite{Hastings2007,Kim2012,Kato2019,Anshu2021,Coopmans2024}]\label{lem:pd_me}
    Let $J\in \mathbb{N}$, and let $\theta \in \mathbb{R}^{J}$ be a parameter vector, and let $G(\theta)$ be the corresponding parameterized QBM Hamiltonian as defined in~\eqref{eq:para-qbm-hamil}. Then, the partial derivative of $e^{-G(\theta)}$ with respect to $\theta_j$ is given as follows:
    \begin{align}
        \partial_{j}e^{-G(\theta)} = -\frac{1}{2}\left\{\Phi_{\theta}(G_j), e^{-G(\theta)}\right\},
    \end{align}
    where the quantum channel $\Phi_{\theta}$ is defined as
    \begin{equation}
        \Phi_{\theta}(X)\coloneqq\int_{\mathbb{R}}dt\ p(t)\, e^{-iG(\theta)t}Xe^{iG(\theta)t},\label{eq:op-ft-def}
    \end{equation}
    and $p(t)$ satisfies
    \begin{equation}
\int_{\mathbb{R}}dt\ p(t)e^{-i\omega t}=\frac{\tanh{\sfrac{\omega}{2}}}{\sfrac{\omega}{2}}.
\label{eq:pt-defined-through-FT}
    \end{equation}
\end{lemma}

\begin{proof}
   According to Duhamel’s formula, the partial derivative of a matrix exponential $e^{A(x)}$ with respect to some parameter $x$ is given as follows:
   \begin{equation}
       \partial_{x}e^{A(x)} = \int_{0}^{1} e^{(1-u)A(x)}\left(\partial_x A(x)\right) e^{uA(x)}\ du.\label{eq:duh_for}
   \end{equation}
   Using this formula for $\partial_{j}e^{-G(\theta)}$, we obtain:
   \begin{align}
       \partial_{j}e^{-G(\theta)} & = \int_{0}^{1} e^{(1-u)(-G(\theta))}\left(\partial_j (-G(\theta)\right)) e^{u(-G(\theta))}\ du\\
       & = -\int_{0}^{1} e^{(u-1)G(\theta)}G_j e^{-uG(\theta)}\ du.\label{eq:duh-simplify}
   \end{align}
   Now, suppose that the spectral decomposition of $G(\theta)$ is as follows:
\begin{equation}
    G(\theta) = \sum_k \lambda_k |k\rangle\!\langle k|,
\end{equation}
where $\{\lambda_k\}_k$ are the eigenvalues and $\{|k\rangle\}_k$ are the corresponding eigenvectors. Substituting the above equation into~\eqref{eq:duh-simplify}, we find that
\begin{align}
    \partial_{j}e^{-G(\theta)}  & = -\int_{0}^{1} \left(\sum_{k}  e^{(u-1)\lambda_k} |k\rangle\!\langle k| \right) G_j \left(\sum_l e^{-u\lambda_l} |l\rangle\!\langle l|\right) du\\
    & = -\int_{0}^{1} \sum_{k, l}  e^{(u-1)\lambda_k} |k\rangle\!\langle k| \left(G_j\right)  e^{-u\lambda_l} |l\rangle\!\langle l|\ du\\
    & = -\sum_{k, l}   |k\rangle\!\langle k| G_j   |l\rangle\!\langle l|\left(\int_{0}^{1} e^{(u-1)\lambda_k} e^{-u\lambda_l}\ du\right)\\
    & = -\sum_{k, l}   |k\rangle\!\langle k| G_j   |l\rangle\!\langle l| \left( e^{-\lambda_k}\int_{0}^{1} e^{u\left(\lambda_k - \lambda_l\right)}\ du\right)\\
    & = -\sum_{k, l}   |k\rangle\!\langle k| G_j   |l\rangle\!\langle l| \left(e^{-\lambda_k} \frac{e^{\lambda_k - \lambda_l} - 1}{\lambda_k - \lambda_l}\right).\label{eq:pd-duh-sd}
\end{align}
Now, consider the following:
\begin{equation}
    e^{-\lambda_k} \frac{e^{\lambda_k - \lambda_l} - 1}{\lambda_k - \lambda_l} =  e^{-\lambda_k} \frac{e^{\lambda_k - \lambda_l} - 1}{e^{\lambda_k - \lambda_l} + 1} \frac{e^{\lambda_k - \lambda_l} + 1}{\lambda_k - \lambda_l} = \frac{\tanh\!{\left(\frac{\lambda_k - \lambda_l}{2}\right)}}{\frac{\lambda_k - \lambda_l}{2}} \frac{e^{- \lambda_l} + e^{- \lambda_k}}{2}.\label{eq:tanh-intro}
\end{equation}
Let $p(t)$ be a function such that its Fourier transform is the following:
\begin{equation}
    \int_{\mathbb{R}} dt\ p(t)e^{-i\omega t} = \frac{\tanh{\sfrac{\omega}{2}}}{\sfrac{\omega}{2}}.\label{eq:pt-fourier} 
\end{equation}
Using this equation and~\eqref{eq:tanh-intro}, we rewrite~\eqref{eq:pd-duh-sd} in the following way:
\begin{align}
   & \partial_{j}e^{-G(\theta)}\notag\\
   & = -\sum_{k, l}   |k\rangle\!\langle k| G_j   |l\rangle\!\langle l|  \left(\left(\int_{\mathbb{R}}dt\  p(t)e^{-i\left(\lambda_k - \lambda_l\right)t}\right) \frac{e^{- \lambda_l} + e^{- \lambda_k}}{2}\right)\\
    & = -\frac{1}{2}\sum_{k, l}   |k\rangle\!\langle k| G_j   |l\rangle\!\langle l| \left(\int_{\mathbb{R}}dt\ p(t) \left(e^{-i\lambda_k t + i\lambda_lt-\lambda_l} + e^{-i\lambda_k t - \lambda_k + i\lambda_lt} \right)\right)\\
    & = -\frac{1}{2}\left(\int_{\mathbb{R}}dt\ p(t) \left( \sum_{k, l}   |k\rangle\!\langle k| G_j   |l\rangle\!\langle l|\ e^{-i\lambda_k t + i\lambda_lt-\lambda_l} + \sum_{k, l}   |k\rangle\!\langle k| G_j   |l\rangle\!\langle l|\ e^{-i\lambda_k t - \lambda_k + i\lambda_lt} \right)\right)\\
    & = -\frac{1}{2}\Bigg(\int_{\mathbb{R}}dt\ p(t) \Bigg( \sum_{k}   e^{-i\lambda_k t} |k\rangle\!\langle k|G_j   \sum_{l} e^{i\lambda_lt-\lambda_l} |l\rangle\!\langle l|\notag \\
    & \hspace{8cm}+ \sum_{k}  e^{-i\lambda_k t - \lambda_k} |k\rangle\!\langle k| G_j   \sum_{l} e^{i\lambda_lt} |l\rangle\!\langle l|\Bigg)\Bigg)\\
    & = -\frac{1}{2}\left(\int_{\mathbb{R}}dt\ p(t) \left( e^{-iG(\theta) t}G_j   e^{iG(\theta) t} e^{-G(\theta)}+ e^{-G(\theta)} e^{-iG(\theta) t} G_j   e^{iG(\theta) t} \right)\right)\\
    & = -\frac{1}{2}\left(\left(\int_{\mathbb{R}}dt\ p(t) e^{-iG(\theta) t}G_j   e^{iG(\theta) t} \right) e^{-G(\theta)} + e^{-G(\theta)} \left(\int_{\mathbb{R}}dt\ p(t) e^{-iG(\theta) t} G_j   e^{iG(\theta) t} \right)\right)\\
    & = -\frac{1}{2}\left(\Phi_{\theta}(G_j) e^{-G(\theta)} + e^{-G(\theta)} \Phi_{\theta}(G_j)\right)\\
     & = -\frac{1}{2}\left\{\Phi_{\theta}(G_j), e^{-G(\theta)}\right\},
\end{align}
where, in the second last equality, we use the definition of the quantum channel $\Phi_{\theta}$ introduced in the lemma statement (see~\eqref{eq:op-ft-def}).
\end{proof}

\begin{figure}
    \centering
    \includegraphics[width=0.7\linewidth]{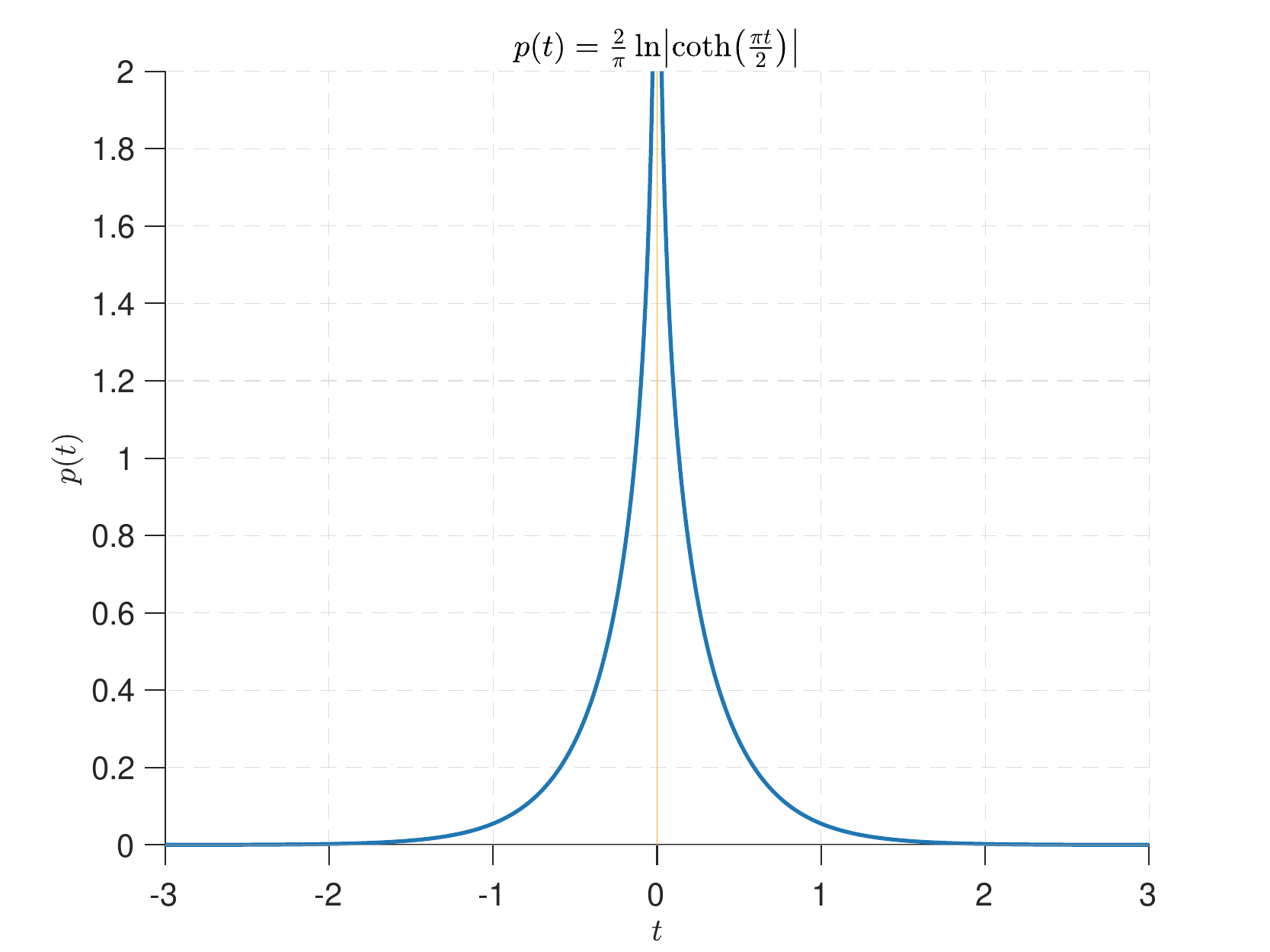}
    \caption{The high-peak-tent probability density function $p(t)$, defined in~\eqref{eq:high-peak-tent-density}.}
    \label{fig:high-peak-tent-density}
\end{figure}

\begin{remark}
\label{rem:FT-high-peak-tent}
The function $p(t)$ that satisfies \eqref{eq:pt-defined-through-FT} is the following ``high-peak-tent'' probability density:
\begin{equation}
p(t) = \frac{2}{\pi}\ln\! \left |  \coth\!\left(  \frac{\pi t}{2}\right) \right|.
\label{eq:high-peak-tent-density}
\end{equation}
See Figure~\ref{fig:high-peak-tent-density} for a plot of $p(t)$.
For a proof of \eqref{eq:high-peak-tent-density}, refer to~\cite[Appendix~B]{Anshu2021} or see Lemma~\ref{lem:FT-high-peak-tent} below. Furthermore, it is important to note that $p(t)$ is indeed a probability density function because $p(t)> 0$, for all $t \in \mathbb{R}$, and
\begin{align}\label{eq:p_pdf}
    \int_{\mathbb{R}} p(t)\ dt = 1.
\end{align}
The inequality $p(t) > 0 $ follows because $|\coth(x)| > 1$ for all $x\in \mathbb{R}$ and \eqref{eq:p_pdf} by plugging in $\omega = 0 $ in \eqref{eq:pt-defined-through-FT} and noting that $\lim_{\omega \to 0} \frac{\tanh{\sfrac{\omega}{2}}}{\sfrac{\omega}{2}} = 1$.
\end{remark}

\begin{lemma}[Fourier transform of high-peak-tent probability density]\label{lem:FT-high-peak-tent}
The following equality holds:
\begin{equation}
\int_{\mathbb{R}}dt\ \frac{2}{\pi}\ln\left\vert \coth\!\left(  \frac{\pi t}
{2}\right)  \right\vert e^{-i\omega t}=\frac{\tanh\!\left(  \frac{\omega}
{2}\right)  }{\frac{\omega}{2}}.
\end{equation}

\end{lemma}

\begin{proof} We provide two different proofs of this lemma. To begin with,
let us define the Fourier operator $\mathcal{F}$, as applied to a function
$f(t)$, as
\begin{equation}
\mathcal{F}\left\{  f(t)\right\}  \coloneqq \int_{\mathbb{R}}dt\ f(t)e^{-i\omega
t}\eqqcolon F(\omega).\label{eq:fourier-transform}
\end{equation}
As such, our aim is to prove that
\begin{equation}
\mathcal{F}\left\{  p(t)\right\}  =\frac{\tanh\!\left(  \frac{\omega}{2}\right)
}{\frac{\omega}{2}},
\end{equation}
where
\begin{equation}
p(t)=\frac{2}{\pi}\ln\left\vert \coth\!\left(  \frac{\pi t}{2}\right)
\right\vert .
\end{equation}
Consider that the derivative of $p(t)$ is as follows:
\begin{align}
\frac{d}{dt}p(t)  & =\frac{d}{dt}\frac{2}{\pi}\ln\left\vert \coth\!\left(
\frac{\pi t}{2}\right)  \right\vert \\
& =\frac{2}{\pi}\frac{1}{\coth\!\left(  \frac{\pi t}{2}\right)  }\frac{d}
{dt}\left[  \coth\!\left(  \frac{\pi t}{2}\right)  \right]  \\
& =\frac{2}{\pi}\frac{1}{\coth\!\left(  \frac{\pi t}{2}\right)  }
\left[-\operatorname{csch}^{2}\!\left(  \frac{\pi t}{2}\right) \right] \frac{\pi}{2}\\
& =-\frac{2}{\sinh\!\left(  \pi t\right)  }.
\end{align}
Then, by the fundamental theorem of calculus, and the fact that $\lim
_{t\rightarrow-\infty}p(t)=0$, we find that
\begin{equation}
p(t)=-2\int_{-\infty}^{t}d\tau\ \frac{1}{\sinh\!\left(  \pi\tau\right)
}.\label{eq:pt-FTOC}
\end{equation}
It is then a well known fact that the following functions are Fourier
transform pairs (see \cite{enwiki:1249353484}):
\begin{equation}
\int_{-\infty}^{t}d\tau\ f(\tau)\qquad\leftrightarrow\qquad\frac{F(\omega
)}{i\omega}+2\pi F(0)\delta(\omega),\label{eq:int-ft-pairs}
\end{equation}
where $F(\omega)$ is the Fourier transform of $f(t)$. Given that the following
are also known Fourier transform pairs
\begin{equation}
\frac{1}{\sinh\!\left(  \pi t\right)  }\qquad\leftrightarrow\qquad -i\tanh\!\left(
\frac{\omega}{2}\right)  ,\label{eq:inv-sinh-tanh-FT-pairs}
\end{equation}
and given that $\tanh(0)=0$ and thus that $\tanh\!\left(  0\right)
\delta(\omega)=0$, we conclude from \eqref{eq:pt-FTOC},
\eqref{eq:int-ft-pairs}, and \eqref{eq:inv-sinh-tanh-FT-pairs} that
\begin{align}
\mathcal{F}\left\{  p(t)\right\}    
& =-2\, \mathcal{F}\left\{  \int_{-\infty}^{t}d\tau\ \frac{1}{\sinh\!\left(  \pi
\tau\right)  }\right\}  \\
& =2\left[  \frac{i\tanh\!\left(  \frac{\omega}{2}\right)  }{i\omega}+\pi
i\tanh\!\left(  0\right)  \delta(\omega)\right]  \\
& =\frac{\tanh\!\left(  \frac{\omega}{2}\right)  }{\frac{\omega}{2}},
\end{align}
as claimed.

Our alternate proof of Lemma~\ref{lem:FT-high-peak-tent} is as follows.
Let us define the inverse Fourier operator $\mathcal{F}^{-1}$, as applied to a function
$F(\omega)$ (defined in~\eqref{eq:fourier-transform}), as
\begin{equation}
    \mathcal{F}^{-1}\!\left\{  F(\omega)\right\}  \coloneqq \frac{1}{2\pi} \int_{\mathbb{R}}d\omega\ F(\omega)e^{i\omega t} = f(t) .\label{def:inv-fourier-transform}
\end{equation}
Also, let the operator $f * g$ denote the convolution between two functions $f$ and $g$:
\begin{equation}
    (f*g)(t) \coloneqq \int_{\mathbb{R}} d\tau\ f(\tau) g(t-\tau) = f(t)*g(t).
\end{equation}
Then from the convolution theorem for inverse Fourier transform, it follows that:
\begin{align}
\mathcal{F}^{-1}\!\left(\frac{\tanh\!\left(  \frac{\omega}{2}\right)  }{\frac{\omega}{2}}\right) & = \mathcal{F}^{-1}\!\left(\tanh\!\left(  \frac{\omega}{2}\right)\right) * \mathcal{F}^{-1}\!\left( \frac{2}{\omega}\right)\\
    & = \frac{-i}{\sinh\!\left(  \pi t\right)  } * \left(-i\operatorname{sgn}(t)\right)\\
    & = \int_{\mathbb{R}} d\tau\ \frac{-i}{\sinh\!\left(  \pi \tau\right)} \cdot  (-i) \operatorname{sgn}(t-\tau)\\
    & = -\int_{\mathbb{R}} d\tau\ \frac{1}{\sinh\!\left(  \pi \tau\right)} \cdot \operatorname{sgn}(t-\tau)\\
     & = -\left(\int_{-\infty}^{t} d\tau\ \frac{1}{\sinh\!\left(  \pi \tau\right)}\cdot (1) + \int_{t}^{\infty} d\tau\ \frac{1}{\sinh\!\left(  \pi \tau\right)} \cdot (-1) \right)\\
     & = -\int_{-\infty}^{t} d\tau\ \frac{1}{\sinh\!\left(  \pi \tau\right)} + \int_{t}^{\infty} d\tau\ \frac{1}{\sinh\!\left(  \pi \tau\right)}\\
       &= -\int_{-\infty}^{t} d\tau\ \frac{1}{\sinh\!\left(  \pi \tau\right)} + \underbrace{\int_{-\infty}^{\infty} d\tau\ \frac{1}{\sinh\!\left(  \pi \tau\right)}}_{=0} - \int_{-\infty}^{t} d\tau\ \frac{1}{\sinh\!\left(  \pi \tau\right)}\\
     & = \int_{-\infty}^{t} d\tau\ \frac{-2}{\sinh\!\left(  \pi \tau\right)} \\
     & = \frac{2}{\pi}\ln\left\vert \coth\!\left(  \frac{\pi t}{2}\right)
\right\vert,
\end{align}
where the last equality follows from~\eqref{eq:pt-FTOC}.
This concludes the alternate proof.
\end{proof}

\medskip

Using the development in Lemma~\ref{lem:pd_me} and Remark~\ref{rem:FT-high-peak-tent}, we now direct our focus on deriving an analytical expression for the partial derivatives, that is, $\partial_j  \operatorname{Tr}[H\rho(\theta)]$, for all $j \in [J]$.

\begin{proposition}[Partial Derivatives]\label{prop:grad}
Let $H$ be a Hamiltonian as defined in~\eqref{eq:def_h}. Let $J\in\mathbb{N}$, let $\theta \in \mathbb{R}^{J}$ be a parameter vector, and let $\rho(\theta)$ be the corresponding parameterized thermal state as defined in~\eqref{eq:para-state}. Then the partial derivative of the function $\operatorname{Tr}[H\rho(\theta)]$ with respect to the parameter~$\theta_j$ can be expressed as follows: 
\begin{align}
\partial_{j}\operatorname{Tr}[H\rho(\theta)]  & 
 = - \frac{1}{2}\left\langle  \left\{  H,\Phi_{\theta}
(G_{j})\right\} \right\rangle_{ \rho(\theta)} + \left\langle H\right\rangle
_{\rho(\theta)}\left\langle G_{j}\right\rangle _{\rho(\theta)}\label{eq:grad_1}\\
& =-\frac{1}{2}\operatorname{Tr}\!\left[  \left\{ H- \left\langle H\right\rangle_{\rho(\theta)}
, \Phi_{\theta}(G_{j}) - \left\langle G_{j}\right\rangle_{\rho(\theta)}\right\}  \rho(\theta)\right]\label{eq:grad_2},
\end{align}
where the quantum channel $\Phi_{\theta}$ is defined in~\eqref{eq:op-ft-def} and $\left\langle C \right\rangle_\sigma \equiv \operatorname{Tr}[C \sigma]$.
\end{proposition}

\begin{proof}
Consider that
\begin{align}
\partial_{j}\operatorname{Tr}[H\rho(\theta)] &  =\operatorname{Tr}[H\left(
\partial_{j}\rho(\theta)\right)  ] =\operatorname{Tr}\!\left[  H \partial_{j}\!\left(  \frac{e^{-G(\theta)}
}{Z(\theta)}\right)  \right]  \\
&  =\operatorname{Tr}\!\left[  H\left(  -\frac{e^{-G(\theta)}}{Z(\theta)^{2}
}\partial_{j}Z(\theta)+\left(  \partial_{j}e^{-G(\theta)}\right)  \frac
{1}{Z(\theta)}\right)  \right],\label{eq:pd_ob_fe} 
\end{align}
where the last equality follows from the product rule.
Now, to evaluate $\partial_{j}Z(\theta)$, consider the following:
\begin{align}
    \partial_{j}Z(\theta) & = \partial_{j}\operatorname{Tr}\!\left[ e^{-G(\theta)}\right] = \operatorname{Tr}\!\left[ \partial_{j}e^{-G(\theta)}\right] \\
    & = -\frac{1}{2}\operatorname{Tr}\!\left[\left\{\Phi_{\theta}(G_j), e^{-G(\theta)}\right\}\right] = -\operatorname{Tr}\!\left[\Phi_{\theta}(G_j) e^{-G(\theta)}\right]\\
    & = -\operatorname{Tr}\!\left[\left(\int_{\mathbb{R}}dt\ p(t) e^{-iG(\theta) t}G_j   e^{iG(\theta) t} \right) e^{-G(\theta)}\right]\\
    & = -\int_{\mathbb{R}}dt\ p(t) \operatorname{Tr}\!\left[e^{-iG(\theta) t}G_j e^{iG(\theta) t} e^{-G(\theta)}\right]\\
    & = -\int_{\mathbb{R}}dt\ p(t) \operatorname{Tr}\!\left[e^{-iG(\theta) t}G_j e^{-G(\theta)} e^{iG(\theta) t} \right] \\
    & = -\int_{\mathbb{R}}dt\ p(t) \operatorname{Tr}\!\left[e^{iG(\theta) t} e^{-iG(\theta) t}G_j e^{-G(\theta)}  \right]\\
    & = -\int_{\mathbb{R}}dt\ p(t) \operatorname{Tr}\!\left[G_j e^{-G(\theta)} \right]\\
    & = -\operatorname{Tr}\!\left[G_j e^{-G(\theta)} \right] \int_{\mathbb{R}}dt\ p(t)\\
    & = -\operatorname{Tr}\!\left[G_j e^{-G(\theta)} \right].
\end{align}
The third equality follows from Lemma~\ref{lem:pd_me}. The seventh equality follows from the fact that $e^{-G(\theta)}$ commutes with $e^{-iG(\theta)t}$, and the eighth equality follows from the cyclicity property of trace. Finally, the penultimate equality follows from the fact that $p(t)$ is a probability density function (see~\eqref{eq:p_pdf}). 

Now, substituting the last equality into~\eqref{eq:pd_ob_fe} and using Lemma~\ref{lem:pd_me} for the second term of~\eqref{eq:pd_ob_fe}, we obtain:
\begin{align}
\partial_{j}\operatorname{Tr}[H\rho(\theta)] &  =\operatorname{Tr}\!\left[  H\left(  -\frac{e^{-G(\theta)}}{Z(\theta)^{2}
}\left(-\operatorname{Tr}[G_{j}e^{-G(\theta)}]\right)- \frac{1}{2}\left\{  \Phi_{\theta}
(G_{j}),e^{-G(\theta)}\right\}  \frac{1}{Z(\theta)}\right)  \right]  \\
& =\operatorname{Tr}\!\left[  H\left( \frac{e^{-G(\theta)}}{Z(\theta)^{2}
}\operatorname{Tr}[G_{j}e^{-G(\theta)}]- \frac{1}{2}\left\{  \Phi_{\theta}
(G_{j}),e^{-G(\theta)}\right\}  \frac{1}{Z(\theta)}\right)  \right]  \\
&  =\operatorname{Tr}\!\left[  H\frac{e^{-G(\theta)}}{Z(\theta)}\right]
\operatorname{Tr}\!\left[  G_{j}\frac{e^{-G(\theta)}}{Z(\theta)}\right]
-\frac{1}{2}\operatorname{Tr}\!\left[  H\left\{  \Phi_{\theta}(G_{j}
),\frac{e^{-G(\theta)}}{Z(\theta)}\right\}  \right]  \\
&  =-\frac{1}{2}\operatorname{Tr}[\left\{  H,\Phi_{\theta}
(G_{j})\right\}  \rho(\theta)]+\operatorname{Tr}[H\rho(\theta
)]\operatorname{Tr}[G_{j}\rho(\theta)] \\
&  =- \frac{1}{2}\operatorname{Tr}\!\left[  \left\{  H,\Phi_{\theta}
(G_{j})\right\}  \rho(\theta)\right] + \left\langle H\right\rangle
_{\rho(\theta)}\left\langle G_{j}\right\rangle _{\rho(\theta)}\\
& =- \frac{1}{2}\left\langle  \left\{  H,\Phi_{\theta}
(G_{j})\right\} \right\rangle_{ \rho(\theta)} + \left\langle H\right\rangle
_{\rho(\theta)}\left\langle G_{j}\right\rangle _{\rho(\theta)}.
\end{align}
This concludes the proof of the first equality claimed in the statement of the proposition.

Now, to show that the second equality in \eqref{eq:grad_2} also holds, consider the following steps:
\begin{align}
&  \frac{1}{2}\operatorname{Tr}\!\left[  \left\{  H-\left\langle H\right\rangle
_{\rho(\theta)},\Phi_{\theta}(G_{j})-\left\langle G_{j}\right\rangle
_{\rho(\theta)}\right\}  \rho(\theta)\right]  \nonumber\\
&  =\frac{1}{2}\operatorname{Tr}\!\left[  \left(  H-\left\langle H\right\rangle
_{\rho(\theta)}\right)  \left(  \Phi_{\theta}(G_{j})-\left\langle
G_{j}\right\rangle _{\rho(\theta)}\right)  \rho(\theta)\right]  \nonumber\\
&  \qquad+\frac{1}{2}\operatorname{Tr}\!\left[  \left(  \Phi_{\theta}
(G_{j})-\left\langle G_{j}\right\rangle _{\rho(\theta)}\right)  \left(
H-\left\langle H\right\rangle _{\rho(\theta)}\right)  \rho(\theta)\right]  \\
&  =\frac{1}{2}\operatorname{Tr}\!\left[  H\Phi_{\theta}(G_{j})\rho
(\theta)\right]  -\frac{1}{2}\operatorname{Tr}\!\left[  H\left\langle
G_{j}\right\rangle _{\rho(\theta)}\rho(\theta)\right] -\frac{1}{2}\operatorname{Tr}\!\left[  \left\langle H\right\rangle
_{\rho(\theta)}\Phi_{\theta}(G_{j})\rho(\theta)\right]  \nonumber\\
&  \qquad  +\frac{1}
{2}\operatorname{Tr}\!\left[  \left\langle H\right\rangle _{\rho(\theta
)}\left\langle G_{j}\right\rangle _{\rho(\theta)}\rho(\theta)\right]
+\frac{1}{2}\operatorname{Tr}\!\left[  \Phi_{\theta}(G_{j})H\rho
(\theta)\right]  -\frac{1}{2}\operatorname{Tr}\!\left[  \Phi_{\theta}
(G_{j})\left\langle H\right\rangle _{\rho(\theta)}\rho(\theta)\right]
\nonumber\\
&  \qquad-\frac{1}{2}\operatorname{Tr}\!\left[  \left\langle G_{j}\right\rangle
_{\rho(\theta)}H\rho(\theta)\right]  +\frac{1}{2}\operatorname{Tr}\!\left[
\left\langle G_{j}\right\rangle _{\rho(\theta)}\left\langle H\right\rangle
_{\rho(\theta)}\rho(\theta)\right]  \\
&  =\frac{1}{2}\operatorname{Tr}\!\left[  H\Phi_{\theta}(G_{j})\rho
(\theta)\right]  -\frac{1}{2}\left\langle H\right\rangle _{\rho(\theta
)}\left\langle G_{j}\right\rangle _{\rho(\theta)} -\frac{1}{2}\left\langle H\right\rangle _{\rho(\theta)}
\operatorname{Tr}\!\left[  \Phi_{\theta}(G_{j})\rho(\theta)\right]  \nonumber\\
&  \qquad+\frac{1}{2}\left\langle H\right\rangle _{\rho(\theta)}\left\langle
G_{j}\right\rangle _{\rho(\theta)}+\frac{1}{2}\operatorname{Tr}\!\left[
\Phi_{\theta}(G_{j})H\rho(\theta)\right] -\frac{1}{2}\left\langle H\right\rangle _{\rho(\theta)}
\operatorname{Tr}\!\left[  \Phi_{\theta}(G_{j})\rho(\theta)\right] \nonumber\\
&  \qquad -\frac{1}
{2}\left\langle G_{j}\right\rangle _{\rho(\theta)}\left\langle H\right\rangle
_{\rho(\theta)} +\frac{1}{2}\left\langle G_{j}\right\rangle _{\rho(\theta
)}\left\langle H\right\rangle _{\rho(\theta)}\\
&  =\frac{1}{2}\operatorname{Tr}\!\left[  \left\{  H,\Phi_{\theta}
(G_{j})\right\}  \rho(\theta)\right]  -\left\langle H\right\rangle
_{\rho(\theta)}\left\langle G_{j}\right\rangle _{\rho(\theta)}\\
& = \frac{1}{2}\left\langle  \left\{  H,\Phi_{\theta}
(G_{j})\right\} \right\rangle_{ \rho(\theta)} - \left\langle H\right\rangle
_{\rho(\theta)}\left\langle G_{j}\right\rangle _{\rho(\theta)} .
\end{align}
By multiplying both sides of the given equation by $-1$, we arrive at the second equality in~\eqref{eq:grad_2}. 
\end{proof}

 \begin{remark}
     We employ the first equality stated in Proposition~\ref{prop:grad}, i.e.,~\eqref{eq:grad_1}, to propose a quantum algorithm that computes an unbiased estimator of the gradient. The purpose of stating the second equality is to illustrate that it resembles an entry in a covariance matrix.
 \end{remark}
 
 Having said that, we now show that for all $j\in[J]$, the absolute value of the partial derivative, i.e.,  $\left|\partial_j\operatorname{Tr}[H\rho(\theta)]\right|$, is bounded from above by a quantity that does not depend on the parameter vector~$\theta$. This demonstrates that the gradient is not arbitrarily large at any given point.

\begin{proposition}[Bounds on Partial Derivatives]\label{prop:grad_bound}
Let $H$ be a Hamiltonian as defined in~\eqref{eq:def_h}. Let $J\in\mathbb{N}$, let $\theta \in \mathbb{R}^{J}$ be a parameter vector, and let $\rho(\theta)$ be the corresponding parameterized thermal state as defined in~\eqref{eq:para-state}. Then, for all $j \in [J]$, the following holds:
\begin{equation}
\left\vert \partial_{j}\operatorname{Tr}[H\rho(\theta)]\right\vert
\leq2\left\Vert H\right\Vert \left\Vert G_{j}\right\Vert.
\end{equation}

\end{proposition}

\begin{proof}
Consider that
\begin{align}
\left\vert \partial_{j}\operatorname{Tr}[H\rho(\theta)]\right\vert  &
=\left\vert -\frac{1}{2}\operatorname{Tr}[\left\{  H,\Phi_{\theta}
(G_{j})\right\}  \rho(\theta)]+\operatorname{Tr}[H\rho(\theta
)]\operatorname{Tr}[G_{j}\rho(\theta)]\right\vert \\
& \leq\frac{1}{2}\left\vert \operatorname{Tr}[\left\{  H,\Phi_{\theta}
(G_{j})\right\}  \rho(\theta)]\right\vert +\left\vert \operatorname{Tr}
[H\rho(\theta)]\right\vert \left\vert \operatorname{Tr}[G_{j}\rho
(\theta)]\right\vert \\
& \leq\frac{1}{2}\left\Vert \left\{  H,\Phi_{\theta}(G_{j})\right\}
\right\Vert +\left\Vert H\right\Vert \left\Vert G_{j}\right\Vert \\
& \leq\left\Vert H\right\Vert \left\Vert G_{j}\right\Vert +\left\Vert
H\right\Vert \left\Vert G_{j}\right\Vert \\
& =2\left\Vert H\right\Vert \left\Vert G_{j}\right\Vert .
\end{align}
The first equality follows from Proposition~\ref{prop:grad}. The first inequality follows from the triangle inequality, and the second inequality follows from H\"{o}lder's inequality. The third inequality is a result of the anticommutator bound, which states that for any two matrices $A$ and $B$, we have $\left\Vert \left\{  A,B\right\}  \right\Vert \leq2\left\Vert A\right\Vert
\left\Vert B\right\Vert$. 
Moreover, we also employed  contractivity under a mixture-of-unitaries channel:
\begin{align}
\left\Vert \Phi_{\theta}(X)\right\Vert  & =\left\Vert \int_{-\infty}^{\infty
}dt\ p(t)e^{-iG(\theta)t}Xe^{iG(\theta)t}\right\Vert \\
& \leq\int_{\mathbb{R}}dt\ p(t)\left\Vert e^{-iG(\theta)t}\right\Vert
\left\Vert X\right\Vert \left\Vert e^{iG(\theta)t}\right\Vert \\
& =\left\Vert X\right\Vert.
\end{align}
This concludes the proof.
\end{proof}

\subsection{Hessian}

\label{sec:hessian}

In this section, we focus on the matrix elements of the Hessian of the objective function $\operatorname{Tr}[H\rho(\theta)]$. We start by obtaining analytical expressions for these elements. Then, we demonstrate that these elements are bounded from above, ensuring that none of them can grow arbitrarily large at any given point. This property is crucial for establishing the smoothness of the objective function, which we will utilize later in our analysis.

\begin{proposition}[Hessian]
\label{prop:hessian-calc}Let $H$ be a Hamiltonian as defined in~\eqref{eq:def_h}. Let $J\in\mathbb{N}$, let $\theta \in \mathbb{R}^{J}$ be a parameter vector, and let $\rho(\theta)$ be the corresponding parameterized thermal state as defined in~\eqref{eq:para-state}. Then the matrix elements of
the Hessian of $\operatorname{Tr}[H\rho(\theta)]$ are given by
\begin{multline}
\partial_{k}\partial_{j}\operatorname{Tr}[H\rho(\theta)]=-\frac{1}
{2}\operatorname{Tr}[\left\{  H,\left[  \partial_{k}\Phi_{\theta}
(G_{j})\right]  \right\}  \rho(\theta)]+\frac{1}{4}\operatorname{Tr}\left[
\left\{  H,\Phi_{\theta}(G_{j})\right\}  \left\{  \rho(\theta),\Phi_{\theta
}(G_{k})\right\}  \right]  \\
-\frac{1}{2}\operatorname{Tr}\left[  \left\{  H,\Phi_{\theta}(G_{j})\right\}
\rho(\theta)\right]  \left\langle G_{k}\right\rangle _{\rho(\theta)}-\frac
{1}{2}\operatorname{Tr}\left[  \left\{  H,\Phi_{\theta}(G_{k})\right\}
\rho(\theta)\right]  \left\langle G_{j}\right\rangle _{\rho(\theta)}\\
-\frac{1}{2}\operatorname{Tr}\left[  \left\{  G_{j},\Phi_{\theta}
(G_{k})\right\}  \rho(\theta)\right]  \left\langle H\right\rangle
_{\rho(\theta)}+2\left\langle H\right\rangle _{\rho(\theta)}\left\langle
G_{k}\right\rangle _{\rho(\theta)}\left\langle G_{j}\right\rangle
_{\rho(\theta)},
\end{multline}
where
\begin{multline}
\partial_{k}\Phi_{\theta}(G_{j})=\int_{\mathbb{R}}dt\ \int_{0}
^{1}du\ it\ p(t)\times\\
\left(  e^{\left(  1-u\right)  iG(\theta)t}G_{k}e^{uiG(\theta)t}
G_{j}e^{iG(\theta)t}-e^{-iG(\theta)t}G_{j}e^{-\left(  1-u\right)  iG(\theta
)t}G_{k}e^{-uiG(\theta)t}\right)  ,
\end{multline}
with $p(t)$ the probability density function defined in~\eqref{eq:p_pdf}. 
\end{proposition}

\begin{proof}
From Proposition~\ref{prop:grad}, we have that
\begin{equation}
\partial_{k}\partial_{j}\operatorname{Tr}[H\rho(\theta)]=\partial_{k}\left(
-\frac{1}{2}\operatorname{Tr}[\left\{  H,\Phi_{\theta}(G_{j})\right\}
\rho(\theta)]+\operatorname{Tr}[H\rho(\theta)]\operatorname{Tr}[G_{j}
\rho(\theta)]\right)  ,\label{eq:starting-pnt-hessian}
\end{equation}
Now, for the first term in the above equation, consider the following:
\begin{align}
& \partial_{k}\operatorname{Tr}[\left\{  H,\Phi_{\theta}(G_{j})\right\}
\rho(\theta)] \notag \\
&  =\operatorname{Tr}[\left\{  H,\left[  \partial_{k}\Phi_{\theta}
(G_{j})\right]  \right\}  \rho(\theta)]+\operatorname{Tr}[\left\{
H,\Phi_{\theta}(G_{j})\right\}  \partial_{k}\rho(\theta)]\\
&  =\operatorname{Tr}[\left\{  H,\left[  \partial_{k}\Phi_{\theta}
(G_{j})\right]  \right\}  \rho(\theta)]\nonumber\\
&  \qquad+\operatorname{Tr}\left[  \left\{  H,\Phi_{\theta}(G_{j})\right\}
\left(  -\frac{1}{2}\left\{  \Phi_{\theta}(G_{k}),\rho(\theta)\right\}
+\rho(\theta)\left\langle G_{k}\right\rangle _{\rho(\theta)}\right)  \right]
\\
&  =\operatorname{Tr}[\left\{  H,\left[  \partial_{k}\Phi_{\theta}
(G_{j})\right]  \right\}  \rho(\theta)]-\frac{1}{2}\operatorname{Tr}\left[
\left\{  H,\Phi_{\theta}(G_{j})\right\}  \left\{  \Phi_{\theta}(G_{k}
),\rho(\theta)\right\}  \right]  \nonumber\\
&  \qquad+\operatorname{Tr}\left[  \left\{  H,\Phi_{\theta}(G_{j})\right\}
\rho(\theta)\left\langle G_{k}\right\rangle _{\rho(\theta)}\right]  \\
&  =\operatorname{Tr}[\left\{  H,\left[  \partial_{k}\Phi_{\theta}
(G_{j})\right]  \right\}  \rho(\theta)]-\frac{1}{2}\operatorname{Tr}\left[
\left\{  H,\Phi_{\theta}(G_{j})\right\}  \left\{  \rho(\theta),\Phi_{\theta
}(G_{k})\right\}  \right]  \nonumber\\
&  \qquad+\operatorname{Tr}\left[  \left\{  H,\Phi_{\theta}(G_{j})\right\}
\rho(\theta)\right]  \left\langle G_{k}\right\rangle _{\rho(\theta)},
\end{align}
where we again used Proposition~\ref{prop:grad} in the second equality. Now consider that
\begin{align}
\partial_{k}\Phi_{\theta}(G_{j}) &  =\partial_{k}\int_{\mathbb{R}}dt\ p(t)e^{-iG(\theta)t}G_{j}e^{iG(\theta)t}\\
&  =\int_{\mathbb{R}}dt\ p(t)\left(  \left[  \partial_{k}e^{-iG(\theta
)t}\right]  G_{j}e^{iG(\theta)t}+e^{-iG(\theta)t}G_{j}\left[  \partial
_{k}e^{iG(\theta)t}\right]  \right)
\end{align}
By applying Duhamel's formula, given by~\eqref{eq:duh_for}, we find
that
\begin{align}
\partial_{k}e^{iG(\theta)t} &  =it\int_{0}^{1}du\ e^{\left(  1-u\right)
iG(\theta)t}G_{k}e^{uiG(\theta)t},\\
\partial_{k}e^{-iG(\theta)t} &  =-it\int_{0}^{1}du\ e^{-\left(  1-u\right)
iG(\theta)t}G_{k}e^{-uiG(\theta)t},
\end{align}
so that
\begin{align}
\partial_{k}\Phi_{\theta}(G_{j}) &  =\int_{\mathbb{R}}dt\ p(t)\left(
\begin{array}
[c]{c}
\left(  it\int_{0}^{1}du\ e^{\left(  1-u\right)  iG(\theta)t}G_{k}
e^{uiG(\theta)t}\right)  G_{j}e^{iG(\theta)t}\\
+e^{-iG(\theta)t}G_{j}\left(  -it\int_{0}^{1}du\ e^{-\left(  1-u\right)
iG(\theta)t}G_{k}e^{-uiG(\theta)t}\right)
\end{array}
\right)  \\
&  =\int_{\mathbb{R}}dt\ \int_{0}^{1}du\ it\ p(t)\left(
\begin{array}
[c]{c}
e^{\left(  1-u\right)  iG(\theta)t}G_{k}e^{uiG(\theta)t}G_{j}e^{iG(\theta)t}\\
-e^{-iG(\theta)t}G_{j}e^{-\left(  1-u\right)  iG(\theta)t}G_{k}e^{-uiG(\theta
)t}
\end{array}
\right).
\end{align}
Finally, for the second term in~\eqref{eq:starting-pnt-hessian}, consider that
\begin{align}
&  \partial_{k}\left(  \operatorname{Tr}[H\rho(\theta)]\operatorname{Tr}
[G_{j}\rho(\theta)]\right)  \nonumber\\
&  =\operatorname{Tr}[H\left(  \partial_{k}\rho(\theta)\right)
]\operatorname{Tr}[G_{j}\rho(\theta)]+\operatorname{Tr}[H\rho(\theta
)]\operatorname{Tr}[G_{j}\left(  \partial_{k}\rho(\theta)\right)  ]\\
&  =\operatorname{Tr}\left[  H\left(  -\frac{1}{2}\left\{  \Phi_{\theta}
(G_{k}),\rho(\theta)\right\}  +\rho(\theta)\left\langle G_{k}\right\rangle
_{\rho(\theta)}\right)  \right]  \operatorname{Tr}[G_{j}\rho(\theta
)]\nonumber\\
&  \qquad+\operatorname{Tr}[H\rho(\theta)]\operatorname{Tr}\left[
G_{j}\left(  -\frac{1}{2}\left\{  \Phi_{\theta}(G_{k}),\rho(\theta)\right\}
+\rho(\theta)\left\langle G_{k}\right\rangle _{\rho(\theta)}\right)  \right]
\\
&  =-\frac{1}{2}\operatorname{Tr}\left[  H\left\{  \Phi_{\theta}(G_{k}
),\rho(\theta)\right\}  \right]  \operatorname{Tr}[G_{j}\rho(\theta
)] +\operatorname{Tr}[H\rho(\theta)]\left\langle G_{k}\right\rangle
_{\rho(\theta)}\operatorname{Tr}[G_{j}\rho(\theta)]\nonumber\\
&  \qquad-\frac{1}{2}\operatorname{Tr}[H\rho(\theta)]\operatorname{Tr}\left[
G_{j}\left\{  \Phi_{\theta}(G_{k}),\rho(\theta)\right\}  \right] +\operatorname{Tr}[H\rho(\theta)]\operatorname{Tr}\left[  G_{j}
\rho(\theta)\right]  \left\langle G_{k}\right\rangle _{\rho(\theta)}\\
&  =-\frac{1}{2}\operatorname{Tr}\left[  \left\{  H,\Phi_{\theta}
(G_{k})\right\}  \rho(\theta)\right]  \left\langle G_{j}\right\rangle
_{\rho(\theta)}+2\left\langle H\right\rangle _{\rho(\theta)}\left\langle
G_{k}\right\rangle _{\rho(\theta)}\left\langle G_{j}\right\rangle
_{\rho(\theta)}\nonumber\\
&  \qquad-\frac{1}{2}\left\langle H\right\rangle _{\rho(\theta)}
\operatorname{Tr}\left[  \left\{  G_{j},\Phi_{\theta}(G_{k})\right\}
\rho(\theta)\right]  .
\end{align}
Then, we finally see that
\begin{multline}
\partial_{k}\partial_{j}\operatorname{Tr}[H\rho(\theta)]
=-\frac{1}{2}\operatorname{Tr}[\left\{  H,\left[  \partial_{k}\Phi_{\theta
}(G_{j})\right]  \right\}  \rho(\theta)]+\frac{1}{4}\operatorname{Tr}\left[
\left\{  H,\Phi_{\theta}(G_{j})\right\}  \left\{  \rho(\theta),\Phi_{\theta
}(G_{k})\right\}  \right]  \\
-\frac{1}{2}\operatorname{Tr}\left[  \left\{  H,\Phi_{\theta}(G_{j})\right\}
\rho(\theta)\right]  \left\langle G_{k}\right\rangle _{\rho(\theta)}-\frac
{1}{2}\operatorname{Tr}\left[  \left\{  H,\Phi_{\theta}(G_{k})\right\}
\rho(\theta)\right]  \left\langle G_{j}\right\rangle _{\rho(\theta)}\\
-\frac{1}{2}\operatorname{Tr}\left[  \left\{  G_{j},\Phi_{\theta}
(G_{k})\right\}  \rho(\theta)\right]  \left\langle H\right\rangle
_{\rho(\theta)}+2\left\langle H\right\rangle _{\rho(\theta)}\left\langle
G_{k}\right\rangle _{\rho(\theta)}\left\langle G_{j}\right\rangle
_{\rho(\theta)},
\end{multline}
thus concluding the proof.
\end{proof}

\begin{proposition}[Bounds on the Hessian Elements]\label{prop:hess-bound}
Let $H$ be a Hamiltonian as defined in~\eqref{eq:def_h}. Let $J \in \mathbb{N}$, let $\theta \in \mathbb{R}^{J}$ be a parameter vector, and let $\rho(\theta)$ be the corresponding parameterized thermal state as defined in~\eqref{eq:para-state}. Then, for all $j, k \in [J]$, the following holds:
\begin{equation}
\left\vert \partial_{k}\partial_{j}\operatorname{Tr}[H\rho(\theta)]\right\vert
\leq8\left\Vert H\right\Vert \left\Vert G_{j}\right\Vert \left\Vert
G_{k}\right\Vert .
\end{equation}

\end{proposition}

\begin{proof}
Consider that
\begin{align}
&  \left\vert \partial_{k}\partial_{j}\operatorname{Tr}[H\rho(\theta
)]\right\vert \nonumber\\
&  \leq\frac{1}{2}\left\vert \operatorname{Tr}[\left\{  H,\left[  \partial
_{k}\Phi_{\theta}(G_{j})\right]  \right\}  \rho(\theta)]\right\vert
\nonumber\\
&  \qquad+\frac{1}{4}\left\vert \operatorname{Tr}\left[  \left\{
H,\Phi_{\theta}(G_{j})\right\}  \left\{  \rho(\theta),\Phi_{\theta}
(G_{k})\right\}  \right]  \right\vert +\frac{1}{2}\left\vert \operatorname{Tr}
\left[  \left\{  H,\Phi_{\theta}(G_{j})\right\}  \rho(\theta)\right]
\left\langle G_{k}\right\rangle _{\rho(\theta)}\right\vert \nonumber\\
&  \qquad+\frac{1}{2}\left\vert \operatorname{Tr}\left[  \left\{
H,\Phi_{\theta}(G_{k})\right\}  \rho(\theta)\right]  \left\langle
G_{j}\right\rangle _{\rho(\theta)}\right\vert +\frac{1}{2}\left\vert
\operatorname{Tr}\left[  \left\{  G_{j},\Phi_{\theta}(G_{k})\right\}
\rho(\theta)\right]  \left\langle H\right\rangle _{\rho(\theta)}\right\vert
\nonumber\\
&  \qquad+2\left\vert \left\langle H\right\rangle _{\rho(\theta)}\right\vert
\left\vert \left\langle G_{k}\right\rangle _{\rho(\theta)}\right\vert
\left\vert \left\langle G_{j}\right\rangle _{\rho(\theta)}\right\vert \\
&  \leq\frac{1}{2}\left\Vert \left\{  H,\left[  \partial_{k}\Phi_{\theta
}(G_{j})\right]  \right\}  \right\Vert +\frac{1}{2}\left\Vert \left\{
H,\Phi_{\theta}(G_{j})\right\}  \Phi_{\theta}(G_{k})\right\Vert +\frac{1}
{2}\left\Vert \left\{  H,\Phi_{\theta}(G_{j})\right\}  \right\Vert \left\Vert
G_{k}\right\Vert \nonumber\\
&  \qquad+\frac{1}{2}\left\Vert \left\{  H,\Phi_{\theta}(G_{k})\right\}
\right\Vert \left\Vert G_{j}\right\Vert +\frac{1}{2}\left\Vert \left\{
G_{j},\Phi_{\theta}(G_{k})\right\}  \right\Vert \left\Vert H\right\Vert +2\left\Vert H\right\Vert \left\Vert G_{k}\right\Vert \left\Vert
G_{j}\right\Vert \\
&  \leq\frac{1}{2}\left\Vert \left\{  H,\left[  \partial_{k}\Phi_{\theta
}(G_{j})\right]  \right\}  \right\Vert +\left\Vert H\right\Vert \left\Vert
G_{j}\right\Vert \left\Vert G_{k}\right\Vert +\left\Vert H\right\Vert
\left\Vert G_{j}\right\Vert \left\Vert G_{k}\right\Vert \nonumber\\
&  \qquad+\left\Vert H\right\Vert \left\Vert G_{k}\right\Vert \left\Vert
G_{j}\right\Vert +\left\Vert G_{j}\right\Vert \left\Vert G_{k}\right\Vert
\left\Vert H\right\Vert +2\left\Vert H\right\Vert \left\Vert G_{k}\right\Vert
\left\Vert G_{j}\right\Vert \\
&  =\frac{1}{2}\left\Vert \left\{  H,\left[  \partial_{k}\Phi_{\theta}
(G_{j})\right]  \right\}  \right\Vert +6\left\Vert H\right\Vert \left\Vert
G_{j}\right\Vert \left\Vert G_{k}\right\Vert \nonumber\\
&  \leq\left\Vert H\right\Vert \left\Vert \partial_{k}\Phi_{\theta}
(G_{j})\right\Vert +6\left\Vert H\right\Vert \left\Vert G_{j}\right\Vert
\left\Vert G_{k}\right\Vert .\label{eq:hessian-bound-middle-proof}
\end{align}
For the first inequality, we first obtain the matrix elements of the Hessian from Proposition~\ref{prop:hessian-calc}. Then, this inequality directly follows from the triangle inequality. The second
inequality follows from H\"{o}lder's inequality and submultiplicativity of the
spectral norm. The third inequality follows from the anticommutator bound and
contracitivity under a mixture-of-unitaries channel, both of which we mention now. The anticommutator bound is the following: given two matrices $A$ and $B$, we have $\left\Vert \left\{  A,B\right\}  \right\Vert \leq2\left\Vert A\right\Vert
\left\Vert B\right\Vert$. 
Now consider that
\begin{align}
&  \left\Vert \partial_{k}\Phi_{\theta}(G_{j})\right\Vert \nonumber\\
&  \leq\left\Vert \int_{\mathbb{R}}dt\ \int_{0}^{1}du\ it\ p(t)\left(
\begin{array}
[c]{c}
e^{\left(  1-u\right)  iG(\theta)t}G_{k}e^{uiG(\theta)t}G_{j}e^{iG(\theta)t}\\
-e^{-iG(\theta)t}G_{j}e^{-\left(  1-u\right)  iG(\theta)t}G_{k}e^{-uiG(\theta
)t}
\end{array}
\right)  \right\Vert \\
&  \leq\int_{\mathbb{R}}dt\ \int_{0}^{1}du\ \left\vert t\right\vert
p(t)\left(
\begin{array}
[c]{c}
\left\Vert e^{\left(  1-u\right)  iG(\theta)t}G_{k}e^{uiG(\theta)t}
G_{j}e^{iG(\theta)t}\right\Vert \\
+\left\Vert e^{-iG(\theta)t}G_{j}e^{-\left(  1-u\right)  iG(\theta)t}
G_{k}e^{-uiG(\theta)t}\right\Vert
\end{array}
\right)  \\
&  \leq\int_{\mathbb{R}}dt\ \int_{0}^{1}du\ \left\vert t\right\vert
p(t)\left(
\begin{array}
[c]{c}
\left\Vert e^{\left(  1-u\right)  iG(\theta)t}\right\Vert \left\Vert
G_{k}\right\Vert \left\Vert e^{uiG(\theta)t}\right\Vert \left\Vert
G_{j}\right\Vert \left\Vert e^{iG(\theta)t}\right\Vert \\
+\left\Vert e^{-iG(\theta)t}\right\Vert \left\Vert G_{j}\right\Vert \left\Vert
e^{-\left(  1-u\right)  iG(\theta)t}\right\Vert \left\Vert G_{k}\right\Vert
\left\Vert e^{-uiG(\theta)t}\right\Vert
\end{array}
\right)  \\
&  =2\int_{\mathbb{R}}dt\ \int_{0}^{1}du\ \left\vert t\right\vert
p(t)\left\Vert G_{j}\right\Vert \left\Vert G_{k}\right\Vert \\
&  =2\left\Vert G_{j}\right\Vert \left\Vert G_{k}\right\Vert \int_{-\infty
}^{\infty}dt\ \left\vert t\right\vert p(t)\\
&  \leq2\left\Vert G_{j}\right\Vert \left\Vert G_{k}\right\Vert .
\end{align}
Applying this to~\eqref{eq:hessian-bound-middle-proof}, we finally conclude
that
\begin{align}
\left\vert \partial_{k}\partial_{j}\operatorname{Tr}[H\rho(\theta
)]\right\vert  &  \leq\left\Vert H\right\Vert \left\Vert \partial_{k}
\Phi_{\theta}(G_{j})\right\Vert +6\left\Vert H\right\Vert \left\Vert
G_{j}\right\Vert \left\Vert G_{k}\right\Vert \\
&  \leq8\left\Vert H\right\Vert \left\Vert G_{j}\right\Vert \left\Vert
G_{k}\right\Vert ,
\end{align}
thus completing the proof.
\end{proof}

\subsection{Smoothness}

\label{sec:smoothness}

From the development above, we are now in a position to prove that the objective function $\operatorname{Tr}[H\rho(\theta)]$ is smooth. We do so in the proof of the following proposition.

\begin{proposition}[Smoothness]
Let $H$ be a Hamiltonian as defined in~\eqref{eq:def_h}. 
Let $J \in \mathbb{N}$, let $\theta \in \mathbb{R}^{J}$ be a parameter vector, and let $\rho(\theta)$ be the corresponding parameterized thermal state as defined in~\eqref{eq:para-state}. 
Then, the objective function $\operatorname{Tr}[H\rho(\theta)]$ is $\ell$-smooth, where
\begin{equation}
    \ell = 8 J \left\Vert H\right \Vert \max \left \{ \left\Vert
G_{k}\right\Vert^2 \right \}_k .\label{eq:smoothness_para}
\end{equation}
\end{proposition}

\begin{proof}
    To prove that the objective function $\operatorname{Tr}[H\rho(\theta)]$  is $\ell$-smooth, we need to show that its gradient $\nabla_{\theta}\operatorname{Tr}[H\rho(\theta)]$ is $\ell$-Lipschitz continuous (see Definition~\ref{eq:smoothness_def}). To this end, note that this gradient is a multivariate vector-valued function. Therefore, it is $\ell$-Lipschitz continuous if all its components, i.e., $\{\partial_{j}\operatorname{Tr}[H\rho(\theta)]\}_j$, are Lipschitz continuous. 
    Let $\ell_j$ be a Lipschitz constant of the function $\partial_{j}\operatorname{Tr}[H\rho(\theta)]$. Then, from Lemma~\ref{lemma:lipvec}, it directly follows that a choice for $\ell$ is
    \begin{align}
        \ell = \left(\sum_{j=1}^{J} \ell_{j}^2\right)^{\frac{1}{2}}.\label{eq:lip-con-l}
    \end{align} 
    Next, we get the Lipschitz constant $\ell_j$ by using Lemma~\ref{lemma:lipmul} and the bounds that we obtained on the elements of the Hessian in Proposition~\ref{prop:hess-bound}:
\begin{align}
    \ell_j & = \sqrt{J} \max_{k} \left \{\sup_{\theta} \left|\partial_k \partial_j\operatorname{Tr}[H\rho(\theta)] \right| \right \}_k\\
    & \leq 8\sqrt{J} \max_{k} \left \{\left\Vert H\right\Vert \left\Vert G_{j}\right\Vert \left\Vert
G_{k}\right\Vert\right \}_k\\
& \leq 8\sqrt{J} \left\Vert H\right\Vert \left\Vert G_{j}\right\Vert \max_{k} \left \{ \left\Vert
G_{k}\right\Vert\right \}_k.
\end{align}
Substituting the above equation into~\eqref{eq:lip-con-l}, we find that
\begin{align}
    \ell & \leq \left(\sum_{j=1}^{J} \left[8\sqrt{J} \left\Vert H\right\Vert \left\Vert G_{j}\right\Vert \max_{k} \left \{ \left\Vert
G_{k}\right\Vert\right \}_k\right]^2 \right)^{\frac{1}{2}}\\
& = 8 J^{\frac{1}{2}} \left( \left\Vert H\right \Vert^2 \max_{k} \left \{ \left\Vert
G_{k}\right\Vert^2 \right \} _k \underbrace{\sum_{j=1}^{J}\left\Vert G_{j}\right\Vert^2}_{\leq\,  J \max_{j} \left \{ \left\Vert
G_{j}\right\Vert^2 \right \}_j} \right)^{\frac{1}{2}}\\
& \leq 8 J \left\Vert H\right \Vert \max_{k} \left \{ \left\Vert
G_{k}\right\Vert^2 \right \}_k.
\end{align}
Finally, we take the right-hand side of the above inequality as the Lipschitz constant $\ell$. 
\end{proof}

\subsection{Quantum Boltzmann gradient estimator}

\label{sec:qbge}

In this section, we present a quantum algorithm for estimating the gradient $\nabla_{\theta}\operatorname{Tr}[H\rho(\theta)]$. As mentioned in the main text and previously in Section~\ref{sec:organization}, we refer to this algorithm as the quantum Boltzmann gradient estimator (QBGE).
For the sake of simplicity, we assume in what follows that for all $k \in [K]$ and for all $j \in [J]$, the Hamiltonians $H_k$ and $G_j$ are local unitaries. However, let us note that our algorithm can straightforwardly be generalized beyond the case of $H_k$ and $G_j$ being local unitaries, if they instead are block encoded into unitary circuits~\cite{Low2019hamiltonian,Gilyen2019}. 

The gradient is a multivariate vector-valued function with components corresponding to the partial derivatives $\{\partial_{j}\operatorname{Tr}[H\rho(\theta)]\}_j$. Therefore, QBGE estimates the gradient by estimating these partial derivatives individually using Algorithms~\ref{algo:est_first_term} and~\ref{algo:est_second_term} as subroutines. We begin by presenting these subroutines first and then provide pseudocode for QBGE at the end of this section, showcasing how QBGE employs these subroutines. To this end, let us recall the explicit form of $\partial_{j}\operatorname{Tr}[H\rho(\theta)]$  from Proposition~\ref{prop:grad}:
\begin{align}\label{eq:recall_pd_ex}
\partial_j\!\operatorname{Tr}\!\left[H\rho(\theta) \right] =
-\frac{1}{2}\left\langle  \left\{  H,\Phi_{\theta}(G_{j})\right\}
\right\rangle_{\rho(\theta)}  + \left\langle H\right\rangle _{\rho(\theta)}\left\langle
G_{j}\right\rangle _{\rho(\theta)}.
\end{align}
Observe that the above equation is a linear combination of two terms. Therefore, we present algorithms (Algorithms~\ref{algo:est_first_term} and~\ref{algo:est_second_term}) for estimating these terms separately.

\subsubsection{Estimating the first term} 

\label{app:1st-term-estimator}

Expanding the first term of \eqref{eq:recall_pd_ex} yields:
\begin{align}
    & -\frac{1}{2}\operatorname{Tr}\!\left[  \left\{  H,\Phi_{\theta}(G_{j})\right\}
\rho(\theta)\right]\notag\\
& = -\frac{1}{2}\operatorname{Tr}\!\left[\left(H\Phi_{\theta}(G_{j}) + \Phi_{\theta}(G_{j})H\right)
\rho(\theta)\right]\\
& = -\frac{1}{2}\operatorname{Tr}\!\Bigg[\Bigg(\left(\sum_{k}\alpha_{k}H_{k}\right)\left(\int_{\mathbb{R}} dt\, p(t) e^{i G(\theta) t} G_j e^{-i G(\theta) t} \right)\notag \\
&\hspace{3cm} + \left(\int_{\mathbb{R}} dt\, p(t) e^{i G(\theta) t} G_j e^{-i G(\theta) t} \right)\left(\sum_{k}\alpha_{k}H_{k}\right)\Bigg)
\rho(\theta)\Bigg]\\
& = -\sum_{k}\alpha_{k} \int_{\mathbb{R}} dt \, p(t) \left(\frac{1}{2}\operatorname{Tr}\!\left[\left(\underbrace{H_{k} e^{i G(\theta) t} G_j e^{-i G(\theta) t}}_{\eqcolon U_{jk}(\theta, t)} + \underbrace{e^{i G(\theta) t} G_j e^{-i G(\theta) t} H_{k}}_{= U_{jk}^{\dagger}(\theta, t)}\right)
\rho(\theta)\right]\right)\\
& = -\sum_{k}\alpha_{k} \int_{\mathbb{R}} dt \, p(t) \left(\frac{1}{2}\operatorname{Tr}\!\left[\left(U_{jk}(\theta, t) +  U_{jk}^{\dagger}(\theta, t)\right)
\rho(\theta)\right]\right)\label{eq:first_term_quant}.
\end{align}
Note that, in the third equality, $H_{k} e^{i G(\theta) t} G_j e^{-i G(\theta) t}$ is a unitary because we are assuming that $\{H_k\}_k$ and $\{G_{j}\}_j$ are local unitaries. 

\begin{figure}
    \centering
    \begin{quantikz}
    \lstick{$\ket{0}\!$} & \gate{\text{Had}} & \ctrl{1} & \gate{\text{Had}} & \meter{} & \setwiretype{n} \cw \rstick{$b$}  \\
    \lstick{$\rho$} & \qw & \gate{U_i} & \qw & \qw & \qw
\end{quantikz}
    \caption{Quantum primitive for estimating  $\frac{1}{2}\operatorname{Tr}\!\left[\left(U_1^\dag U_0  +  U_0^{\dagger}U_1\right)\rho\right]$. Note that the ``Had'' gate denotes the Hadamard gate.}
    \label{fig:qc-primitive}
\end{figure}
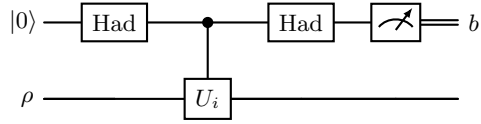

The key idea here is to estimate the term on the right-hand side of \eqref{eq:first_term_quant}. However, before presenting a quantum algorithm to accomplish this, let us first recall a fundamental primitive on which this algorithm is based. This primitive estimates a quantity of the form
\begin{align}
    \frac{1}{2}\operatorname{Tr}\!\left[\left(U_1^\dag U_0  +  U_0^{\dagger}U_1\right)\rho\right]\label{eq:U+Udag},
\end{align}
where $U_0$  and $U_1$ are unitaries and $\rho$ is a quantum state. Observe that the expression above is similar to the expression on the right-hand side of~\eqref{eq:first_term_quant}, so understanding how this primitive works is crucial. The quantum circuit for this primitive is depicted in Figure~\ref{fig:qc-primitive}, where the controlled gate is given by $|0\rangle\!\langle 0 | \otimes U_0 + |1\rangle\!\langle 1 | \otimes U_1$. Moreover, this circuit consists of the following two quantum registers: 1) a control register, initialized in the state $|0\rangle$, and 2) a system register, which is in the state~$\rho$. After executing this circuit and obtaining a measurement outcome~$b$ in the control register, the final state $\sigma_{\operatorname{sub}}^{(b)}$ (sub-normalized) of the system register is as follows, where $b\in\{0, 1\}$:
\begin{align}
    \sigma_{\operatorname{sub}}^{(b)} & = \left(\frac{U_0 + (-1)^b U_1}{2}\right) \rho \left(\frac{U_0^\dag + (-1)^b U_1^{\dagger}}{2}\right)\\
    & = \frac{1}{4}\left( U_0\rho U_0^\dag +(-1)^b U_0\rho U_1^\dag + (-1)^b U_1\rho U_0^{\dagger} + U_1\rho U_1^{\dagger} \right).
\end{align}
The probability $p_{b}$ of obtaining the measurement outcome $b$ is then
\begin{align}
    p_{b} = \operatorname{Tr}\!\left[\sigma_{\operatorname{sub}}^{(b)}\right] & = \frac{2 + (-1)^b\operatorname{Tr}\!\left[U_0\rho U_1^\dag \right] + (-1)^b\operatorname{Tr}\!\left[U_1\rho U_0^{\dagger} \right]}{4}\\
    & = \frac{2 + (-1)^b\operatorname{Tr}\!\left[\left(U_1^\dag U_0  +  U_0^{\dagger}U_1\right) \rho \right]}{4}\label{eq:fun_pri_prob}.
\end{align}

In order to estimate the quantity, given by~\eqref{eq:U+Udag}, we use the following approach. Let $b_1, \ldots, b_N$ represent the measurement results obtained from $N$ independent executions of the aforementioned quantum circuit. We define a random variable $X$ as $X \coloneqq (-1)^b$. Then the sample mean
\begin{equation}
    \overline{X} \coloneqq \frac{1}{N}\sum_{i=1}^{N} X_{i} 
\end{equation}
serves as an unbiased estimator for the quantity of interest, because
\begin{align}
    \mathbb{E}\!\left[\,\overline{X}\, \right] & = \mathbb{E}\!\left[\frac{1}{N}\sum_{i=1}^{N} X_{i}\right] = \frac{1}{N}\sum_{i=1}^{N} \mathbb{E}\!\left[X_{i}\right]= \frac{1}{N}\sum_{i=1}^{N} \sum_{b_{i} \in \{0, 1\}}p_{b_{i}} (-1)^{b_{i}} \\
    & = \frac{1}{N}\sum_{i=1}^{N}\left ( \frac{2 + \operatorname{Tr}\!\left[\left(U_1^\dag U_0  +  U_0^{\dagger}U_1\right)\rho \right]}{4} - \frac{2 - \operatorname{Tr}\!\left[\left(U_1^\dag U_0  +  U_0^{\dagger}U_1\right)\rho \right]}{4}\right)\\
    & = \frac{\operatorname{Tr}\!\left[\left(U_1^\dag U_0  +  U_0^{\dagger}U_1\right)\rho \right]}{2}  .
\end{align}

\begin{figure}[t]
    \centering
    \begin{quantikz}
    \lstick{$\ket{0}\!$} & \gate{\text{Had}}   &\ctrl{1} &  & \ctrl{1} & \gate{\text{Had}} & \meter{} & \setwiretype{n} \cw \rstick{$b_n$}  \\
    \lstick{$\rho(\theta)$} & \qw 
    &\gate{G_j} & \gate{e^{-i G(\theta)t}} & \gate{H_k} &\qw & \qw & \qw
\end{quantikz}
    \caption{Quantum circuit corresponding to the $n^{\operatorname{th}}$ iteration of Algorithm~\ref{algo:est_first_term}.}
    \label{fig:qc-algo-1}
\end{figure}
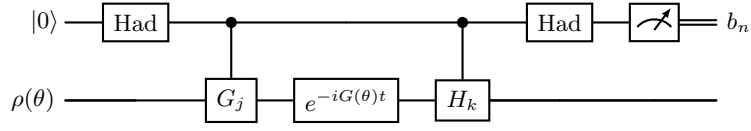

\begin{algorithm}[H]
\caption{\texorpdfstring{$\mathtt{estimate\_first\_term}(\alpha, \{H_k\}_{k=1}^{K}, \theta, \{G_\ell\}_{\ell=1}^{J}, j, \varepsilon_1, \delta_1)$}{estimate first term}}
\label{algo:est_first_term}
\begin{algorithmic}[1]
\STATE \textbf{Input:} Hamiltonian coefficients vector $\alpha = \left( \alpha_{1}, \ldots,  \alpha_{K}\right)^{\mathsf{T}} \in \mathbb{R}_{\geq 0}^{K}$, local Hamiltonians $\{H_k\}_{k=1}^{K}$, parameter vector $\theta = \left( \theta_{1}, \ldots,  \theta_{J}\right)^{\mathsf{T}} \in \mathbb{R}^{J}$, Gibbs local Hamiltonians $\{G_\ell\}_{\ell=1}^{J}$, index $j \in [J]$, precision $\varepsilon_1 > 0$, error probability $\delta_1 > 0$
\STATE $N_1 \leftarrow \lceil\sfrac{2\left \Vert \alpha \right \Vert_1^2 \ln(\sfrac{2}{\delta_1})}{\varepsilon_1^2}\rceil$
\FOR{$n = 0$ to $N_1-1$}
\STATE Initialize the control register to $|0\rangle$
\STATE Prepare the system register in the state $\rho(\theta)$
\STATE Sample $k$ and $t$ with probabilities $\sfrac{\alpha_k}{\left \Vert \alpha \right \Vert_1}$ and $p(t)$, respectively
\STATE Apply the Hadamard gate to the control register
\STATE Apply the following unitaries to the control and system registers:
\STATE \hspace{0.6cm} \textbullet~Controlled-$G_j$: $G_j$ is a local unitary with the control on the control register
\STATE \hspace{0.6cm} \textbullet~$e^{-iG(\theta)t}$: Hamiltonian simulation for time $t$ on the system register
\STATE \hspace{0.6cm} \textbullet~Controlled-$H_k$: $H_k$ is a local unitary with the control on the control register
\STATE Apply the Hadamard gate to the control register
\STATE Measure the control register in the computational basis and store the measurement outcome~$b_n$
\STATE $Y_{n}^{(1)} \leftarrow \left \Vert \alpha \right \Vert_1 (-1)^{b_n+1}$
\ENDFOR
\RETURN $\overline{Y}^{(1)} \leftarrow \frac{1}{N_1}\sum_{n=0}^{N_1-1}Y_{n}^{(1)}$
\end{algorithmic}
\end{algorithm}

With the above primitive in mind, we are now in a position to present an algorithm (Algorithm~\ref{algo:est_first_term}) to estimate the first term of~\eqref{eq:recall_pd_ex}, i.e., $-\frac{1}{2}\!\operatorname{Tr}\!\left[  \left\{  H,\Phi_{\theta}(G_{j})\right\}
\rho(\theta)\right]$, using its equivalent form given by~\eqref{eq:first_term_quant}. Let $\alpha \coloneqq \left( \alpha_{1}, \ldots,  \alpha_{K}\right)^{\mathsf{T}}\!\in \mathbb{R}_{\geq 0}^{K}$. For this algorithm, we assume that we have access to an oracle that samples an index $k \in [K]$ and time $t \in \mathbb{R}$ with probabilities $\sfrac{\alpha_k}{\left \Vert \alpha \right \Vert_1}$ and $p(t)$, respectively. Additionally, we assume that we have access to multiple copies of $\rho(\theta)$. At the core of our algorithm lies the aforementioned primitive with $U_0 = e^{-i G(\theta)t}$ and $U_1 = H_k e^{-i G(\theta)t} G_j $ (see Figure~\ref{fig:qc-algo-1}), so that $U_1^\dag U_0 = G_j  e^{i G(\theta)t} H_k e^{-i G(\theta)t}$ and
\begin{align}
\frac{1}{2}\operatorname{Tr}\!\left[\left(U_1^\dag U_0  +  U_0^{\dagger}U_1\right)\rho(\theta) \right] &  = \frac{1}{2}\operatorname{Tr}\!\left[\left(G_j  e^{i G(\theta)t} H_k e^{-i G(\theta)t}  +  e^{i G(\theta)t} H_k  e^{-i G(\theta)t} G_j \right)\rho(\theta) \right]\\
& = \frac{1}{2}\operatorname{Tr}\!\left[\left(e^{-i G(\theta)t}G_j  e^{i G(\theta)t} H_k   +   H_k  e^{-i G(\theta)t} G_j e^{i G(\theta)t}\right)\rho(\theta) \right] \\
& = \frac{1}{2}\operatorname{Tr}\!\left[\left\{H_k,e^{-i G(\theta)t}G_j  e^{i G(\theta)t} \right\}\rho(\theta) \right],
\end{align}
where used the fact that $[G(\theta),\rho(\theta)]=0$ and cyclicity of trace to obtain the second equality.

The output of our algorithm is a random variable $\overline{Y}^{(1)}$, which we show is an unbiased estimator of $-\frac{1}{2}\operatorname{Tr}[\{H, \Phi_\theta(G_j)\}\rho(\theta)]$:
\begin{align}
    \mathbb{E}\!\left [\overline{Y}^{(1)}\right]&  = \mathbb{E}\!\left [\frac{1}{N_1}\sum_{n=0}^{N_1-1}Y_n^{(1)}\right] = \mathbb{E}\!\left [\frac{1}{N_1}\sum_{n=0}^{N_1-1} \left \Vert \alpha \right \Vert_1 (-1)^{b_n+1} \right]\\
    & = -\frac{\left \Vert \alpha \right \Vert_1}{N_1}\sum_{n=0}^{N_1-1}\mathbb{E}\!\left [  (-1)^{b_n} \right] = -\frac{\left \Vert \alpha \right \Vert_1}{N_1}\sum_{n=0}^{N_1-1}\sum_{b_{n} \in \{0, 1\}}p_{b_{n}}\left[ (-1)^{b_{n}} \right]\label{eq:un_bias_est_proof_first_term},
\end{align}
where 
\begin{align}
    p_{b_n} \coloneqq  \sum_{k}\frac{\alpha_{k}}{\left \Vert \alpha \right \Vert_1} \int_{\mathbb{R}} dt \, p(t) \left ( \frac{2 + (-1)^{b_n}\operatorname{Tr}\!\left[\left (U_{jk}(\theta, t) + U_{jk}^{\dagger}(\theta, t) \right)\rho(\theta) \right]}{4}  \right). 
\end{align}
The above expression for $p_{b_n}$ follows from the fact that we first sample an index $k$ and time $t$ with probabilities $\sfrac{\alpha_{k}}{\left \Vert \alpha \right \Vert_1}$ and $p(t)$, respectively, and then we apply the primitive introduced before with $U_0 = e^{-i G(\theta)t}$ and $U_1 = H_k e^{-i G(\theta)t} G_j $, whose probability of outputting a bit $b$ is given by~\eqref{eq:fun_pri_prob}.
Plugging the expression above for $p_{b_n}$ into~\eqref{eq:un_bias_est_proof_first_term} and simplifying, we finally get
\begin{align}
    \mathbb{E}\!\left[ \overline{Y}^{(1)} \right] & = -\sum_{k}\alpha_{k} \int_{\mathbb{R}} dt \, p(t) \left(\frac{1}{2}\operatorname{Tr}\!\left[\left(U_{jk}(\theta, t) +  U_{jk}^{\dagger}(\theta, t)\right)
\rho(\theta)\right]\right)\\
& = -\frac{1}{2}\operatorname{Tr}\!\left[  \left\{  H,\Phi_{\theta}(G_{j})\right\}
\rho(\theta)\right],
\end{align}
where the second equality follows due to~\eqref{eq:first_term_quant}.

\subsubsection{Estimating the second term} 

\label{app:2nd-term-estimator}

We now expand the second term on the right-hand side of~\eqref{eq:recall_pd_ex}:
\begin{align}
    \left\langle H\right\rangle _{\rho(\theta)}\left\langle
G_{j}\right\rangle _{\rho(\theta)} & = \operatorname{Tr}[H\rho(\theta)]\operatorname{Tr}[G_{j}\rho(\theta)]\\
& = \operatorname{Tr}\!\left[\left (\sum_k \alpha_k H_k\right)\rho(\theta)\right]\operatorname{Tr}\!\left[G_{j}\rho(\theta)\right]\\
& = \sum_k \alpha_k \operatorname{Tr}\!\left [ H_k\rho(\theta)\right]\operatorname{Tr}\!\left[G_{j}\rho(\theta)\right].
\end{align}
The second equality follows directly from the definition of $H$, given by~\eqref{eq:def_h}. 

\begin{algorithm}[H]
\caption{$\mathtt{estimate\_second\_term}(\alpha, \{H_k\}_{k=1}^{K}, G_j, \varepsilon_2, \delta_2)$}
\label{algo:est_second_term} 
\begin{algorithmic}[1]
\STATE \textbf{Input:} Hamiltonian coefficients vector $\alpha = \left( \alpha_{1}, \ldots,  \alpha_{K}\right)^{\mathsf{T}} \in \mathbb{R}_{\geq 0}^{K}$, local Hamiltonians $\{H_k\}_{k=1}^{K}$, parameter vector $\theta = \left( \theta_{1}, \ldots,  \theta_{J}\right)^{\mathsf{T}} \in \mathbb{R}^{J}$, Gibbs local Hamiltonian $G_j$, precision $\varepsilon_2 > 0$, error probability $\delta_2 > 0$
\STATE $N_2 \leftarrow \sfrac{2\left \Vert \alpha \right \Vert_1^2 \ln(\sfrac{2}{\delta_2})}{\varepsilon_2^2}$
\FOR{$n = 0$ to $N_2-1$}
\STATE Sample $k$ with probability $\sfrac{\alpha_k}{\left \Vert \alpha \right \Vert_1}$
\STATE Prepare a register in the state $\rho_{\theta}$, measure $H_k$, and store the measurement outcome $h_n$
\STATE Prepare a register in the state $\rho_{\theta}$, measure $G_j$, and store the measurement outcome $g_n$
\STATE Set $Y_n^{(2)} \leftarrow \left \Vert \alpha \right \Vert_1 (-1)^{h_n + g_n}$
\ENDFOR
\RETURN $\overline{Y}^{(2)} \leftarrow \frac{1}{N_2}\sum_{n=0}^{N_2-1}Y_n^{(2)}$
\end{algorithmic}
\end{algorithm}

Then Algorithm~\ref{algo:est_second_term}  estimates the above quantity. As we did previously for Algorithm~\ref{algo:est_first_term}, we assume that we have access to an oracle that samples an index $k \in [K]$ with probability $\sfrac{\alpha_k}{\left \Vert \alpha \right \Vert_1}$ and that we have access to multiple copies of $\rho(\theta)$. The output of this algorithm is a random variable $\overline{Y}^{(2)}$, which can be easily shown to be an unbiased estimator of our quantity of interest, i.e., $\left\langle H\right\rangle _{\rho(\theta)}\left\langle
G_{j}\right\rangle _{\rho(\theta)}$.

\subsubsection{Estimating the full gradient using QBGE}

With the algorithms (Algorithms~\ref{algo:est_first_term} and~\ref{algo:est_second_term}) for estimating the first and second terms of the partial derivatives in place, we now provide pseudocode for QBGE in Algorithm~\ref{algo:quant_Boltzmann_grad_estimator}. As previously mentioned, this algorithm estimates the full gradient $\nabla_{\theta} \operatorname{Tr}[H\rho(\theta)]$ and outputs an estimator $\overline{g}$, the latter of which can be easily verified to be an unbiased estimator of $\nabla_{\theta}\operatorname{Tr}[H\rho(\theta)]$.

\begin{algorithm}[H]
\caption{$\mathtt{QBGE}(\alpha, \{H_k\}_{k=1}^{K}, \theta, \{G_\ell\}_{\ell=1}^{J}, \varepsilon_1, \varepsilon_2, \delta_1, \delta_2)$}\label{algo:quant_Boltzmann_grad_estimator} 
\begin{algorithmic}[1]
\STATE \textbf{Input:} Hamiltonian coefficients vector $\alpha = \left( \alpha_{1}, \ldots,  \alpha_{K}\right)^{\mathsf{T}} \in \mathbb{R}_{\geq 0}^{K}$, local Hamiltonians $\{H_k\}_{k=1}^{K}$, parameter vector $\theta = \left( \theta_{1}, \ldots,  \theta_{J}\right)^{\mathsf{T}} \in \mathbb{R}^{J}$, Gibbs local Hamiltonians $\{G_\ell\}_{\ell=1}^{J}$, precisions $\varepsilon_1, \varepsilon_1  > 0$, error probabilities $\delta_1, \delta_2 > 0$
\FOR{$j = 1$ to $J$}
\STATE $ \overline{Y}^{(1)}_j \leftarrow \mathtt{estimate\_first\_term}(\alpha, \{H_k\}_{k=1}^{K}, \theta, \{G_\ell\}_{\ell=1}^{J}, j, \varepsilon_1, \delta_1)$
\STATE $ \overline{Y}^{(2)}_j \leftarrow \mathtt{estimate\_second\_term}(\alpha, \{H_k\}_{k=1}^{K}, G_j, \varepsilon_2, \delta_2)$
\STATE $\overline{g}_j \leftarrow \overline{Y}^{(1)}_j + \overline{Y}^{(2)}_j$
\ENDFOR
\RETURN $\overline{g} \leftarrow \left( \overline{g}_{1}, \ldots,  \overline{g}_{J}\right)^{\mathsf{T}}$
\end{algorithmic}
\end{algorithm}

\begin{remark}
    Given the form of the Hessian in Proposition~\ref{prop:hessian-calc}, along with the fact that the function $|t|\, p(t)$ is normalizable as a probability density function (i.e., $\int_{\mathbb{R}} dt\, |t|\,  p(t) \approx 0.2714 $), we can devise a quantum algorithm for estimating the matrix elements of the Hessian. The algorithm employs  ideas similar to those used in this section to derive the algorithm  for estimating the elements of the gradient. This quantum algorithm for estimating the matrix elements of the Hessian can be used in a second-order stochastic Newton search method, thus extending the SGD algorithm used in our paper. We leave the detailed exploration of this approach for future work.
\end{remark}

\subsection{SGD for QBM learning of ground-state energies}

\label{sec:sgd_qbm_learning_algo}

In the previous section, we introduced QBGE (see Algorithm~\ref{algo:quant_Boltzmann_grad_estimator}), an algorithm for estimating the gradient $\nabla_{\theta}\operatorname{Tr}[H\rho(\theta)]$. This algorithm outputs an unbiased estimator $\overline{g}(\theta)$ of the gradient at a given point $\theta \in \mathbb{R}^{J}$. Note that from now on, we use the notation $\overline{g}(\theta)$ instead of simply using $\overline{g}$ (as in Algorithm~\ref{algo:quant_Boltzmann_grad_estimator}) to emphasize the explicit dependence on $\theta$. Having said that, in this section, we present an algorithm that uses SGD for QBM learning of ground-state energies (Definition~\ref{def:qbm_learn_def}).
\begin{algorithm}[H]
\caption{$\mathtt{QBM\_GSE}(\alpha, \{H_k\}_{k=1}^{K}\{G_\ell\}_{\ell=1}^{J}, \varepsilon)$}\label{algo:qbm-gse}
\begin{algorithmic}[1]
\STATE \textbf{Input:} Hamiltonian coefficients vector $\alpha = \left( \alpha_{1}, \ldots,  \alpha_{K}\right)^{\mathsf{T}} \in \mathbb{R}_{\geq 0}^{K}$, local Hamiltonians $\{H_k\}_{k=1}^{K}$, Gibbs local Hamiltonians $\{G_\ell\}_{\ell=1}^{J}$, precision $\varepsilon  > 0$
\STATE $\theta \leftarrow$ Random Initialization
\STATE $M \leftarrow \left\lceil\frac{12  \ell \Delta}{\varepsilon^2}\right\rceil $
\STATE $\varepsilon_1, \varepsilon_2 \leftarrow \frac{\varepsilon}{2\sqrt{2J}}$
\STATE $\delta_1, \delta_2 \leftarrow \frac{\varepsilon^2}{8J\left\Vert \alpha \right \Vert_1^2}$
\FOR{$m = 1$ to $M$}
\STATE $\overline{g}(\theta_m) \leftarrow \mathtt{QBGE}(\alpha, \{H_k\}_{k=1}^{K}, \theta_m, \{G_\ell\}_{\ell=1}^{J}, \varepsilon_1, \varepsilon_2, \delta_1, \delta_2)$
\STATE $\theta_{m+1} \leftarrow \theta_{m+1} - \eta \overline{g}(\theta_m)$
\ENDFOR
\RETURN $\min_{m\in[M]}\operatorname{Tr}[H\rho(\theta_{m})]$
\end{algorithmic}
\end{algorithm}

\subsection{Sample complexity}

\label{sec:sample_comp}

In this section, we investigate the sample complexity -- the number of samples of the thermal state~$\rho(\theta)$ -- required by the QBM-GSE algorithm (Algorithm~\ref{algo:qbm-gse}) to reach an $\varepsilon$-stationary point of the optimization problem defined in~\eqref{eq:qbm_learn_def}. To simplify the discussion, we divide the analysis into two parts. First, we investigate the sample complexities of Algorithms~\ref{algo:est_first_term} and~\ref{algo:est_second_term} and then investigate the sample complexity of the QBM-GSE algorithm itself. This is because the QBM-GSE algorithm employs QBGE for gradient estimation, which, in turn, employs Algorithms~\ref{algo:est_first_term} and~\ref{algo:est_second_term} for estimating partial derivatives.

\subsubsection{Sample complexities of Algorithms~\ref{algo:est_first_term} and~\ref{algo:est_second_term}}

Recall that Algorithms~\ref{algo:est_first_term} and~\ref{algo:est_second_term} output estimates of the first and second terms of the partial derivative $\partial_j\!\operatorname{Tr}\!\left[H\rho(\theta) \right]$, where
\begin{align}
\partial_j\!\operatorname{Tr}\!\left[H\rho(\theta) \right] =
-\frac{1}{2}\left\langle  \left\{  H,\Phi_{\theta}(G_{j})\right\}
\right\rangle_{\rho(\theta)}  + \left\langle H\right\rangle _{\rho(\theta)}\left\langle
G_{j}\right\rangle _{\rho(\theta)}.\label{eq:par_deri_exp}
\end{align}
We demonstrated in Section~\ref{sec:qbge} that these estimators are unbiased. The next step is to investigate how fast these estimators converge to their respective expected values, which we do formally in the proofs of the following two lemmas. 

\begin{lemma}[Sample Complexity of Algorithm~\ref{algo:est_first_term}]\label{lem:first_term}
    Let $\varepsilon_1 > 0$, $J, K \in \mathbb{N}$, $\delta_1 \in (0, 1)$, $\theta \in \mathbb{R}^J$, $j\in[J]$, and $\alpha \in \mathbb{R}_{\geq 0}^K$. Then, the number of samples, $N_1$, of $\rho(\theta)$ used by Algorithm~\ref{algo:est_first_term} to produce an $\varepsilon_1$-close estimate of $-\frac{1}{2}\operatorname{Tr}\!\left[  \left\{  H,\Phi_{\theta}(G_{j})\right\} \rho(\theta)\right]$ with a success probability not less than $1 - \delta_1$ is
\begin{equation}
    N_1 = \left \lceil \frac{2\left \Vert \alpha \right \Vert_1^2 \ln(\sfrac{2}{\delta_1})}{\varepsilon_1^2} \right\rceil.
\end{equation}
\end{lemma}
\begin{proof}
    Recall that Algorithm~\ref{algo:est_first_term} outputs an unbiased estimator $\overline{Y}^{(1)}$ of $-\frac{1}{2}\operatorname{Tr}\!\left[  \left\{  H,\Phi_{\theta}(G_{j})\right\} \rho(\theta)\right]$: 
\begin{equation}
    \overline{Y}^{(1)} = \frac{1}{N_1}\sum_{n=0}^{N_1-1}Y_n^{(1)},\label{eq:Z1}
\end{equation} 
where $Y_n^{(1)} = \left \Vert \alpha \right \Vert_1 (-1)^{b_n+1}$ and $b_n \in \{0, 1\}$, for all $n \in \{0, \ldots, N_1-1\}$. This implies that $Y_n^{(1)}$ lies in the range
\begin{equation}
   -\left \Vert \alpha \right \Vert_1 \leq Y_n^{(1)} \leq \left \Vert \alpha \right \Vert_1. 
\end{equation}
Now, using the Hoeffding inequality (Lemma~\ref{lem:hoeffding}), we can say that for $\varepsilon_1 > 0$, we have
\begin{equation}
    \Pr\!\left( \left |\overline{Y}^{(1)} - \mathbb{E}\!\left[ \overline{Y}^{(1)}\right] \right|  \geq \varepsilon_1\right) \leq  \exp\!\left( \frac{-N_1\varepsilon_1^2}{2\left \Vert \alpha \right \Vert_1^2}\right).
\end{equation}
This implies that for 
\begin{equation}
    N_1 \geq \frac{2\left \Vert \alpha \right \Vert_1^2 \ln(\sfrac{2}{\delta_1})}{\varepsilon_1^2},
\end{equation}
we have
\begin{equation}
    \Pr \!\left (\left |\overline{Y}^{(1)} - \mathbb{E}\!\left[ \overline{Y}^{(1)}\right] \right| \leq \varepsilon_1 \right) \geq 1 - \delta_1 ,\label{eq:Z(1)_bound}
\end{equation}
thus concluding the proof.
\end{proof}

\begin{lemma}[Sample Complexity of Algorithm~\ref{algo:est_second_term}]\label{lem:second_term}
    Let $\varepsilon_2 > 0$,  $J, K \in \mathbb{N}$, $\delta_2 \in (0, 1)$, $\theta \in \mathbb{R}^J$, $j\in[J]$, and $\alpha \in \mathbb{R}_{\geq 0}^K$. Then, the number of samples, $N_2$, of $\rho(\theta)$ used by Algorithm~\ref{algo:est_second_term} to produce an $\varepsilon_2$-close estimate of $\left\langle H\right\rangle _{\rho(\theta)}\left\langle
G_{j}\right\rangle _{\rho(\theta)}$ with a success probability not less than $1 - \delta_2$ is
\begin{equation}
    N_2 = \left \lceil \frac{2\left \Vert \alpha \right \Vert_1^2 \ln(\sfrac{2}{\delta_2})}{\varepsilon_2^2} \right\rceil.
\end{equation}
\end{lemma}
\begin{proof}
    The proof follows a similar line of reasoning as that of Lemma~\ref{lem:first_term}, and we provide it here for completeness. Recall that Algorithm~\ref{algo:est_second_term} outputs an unbiased estimator $\overline{Y}^{(2)}$ of $\left\langle H\right\rangle _{\rho(\theta)}\left\langle
G_{j}\right\rangle _{\rho(\theta)}$: 
\begin{equation}
    \overline{Y}^{(2)} = \frac{1}{N_2}\sum_{n=0}^{N_2-1}Y_n^{(2)},\label{eq:Z2}
\end{equation} 
where $Y_n^{(2)} = \left \Vert \alpha \right \Vert_1 (-1)^{h_n + g_n}$ and $h_n, g_n \in \{0, 1\}$, for all $n \in \{0, \ldots, N_2-1\}$. This implies that $Y_n^{(2)}$ lies in the range
\begin{equation}
   -\left \Vert \alpha \right \Vert_1 \leq Y_n^{(2)} \leq \left \Vert \alpha \right \Vert_1. 
\end{equation}
Now, using the Hoeffding inequality (Lemma~\ref{lem:hoeffding}), we can say that for $\varepsilon_2 > 0$, we have
\begin{equation}
    \Pr\!\left( \left |\overline{Y}^{(2)} - \mathbb{E}\!\left[ \overline{Y}^{(2)}\right] \right|  \geq \varepsilon_2\right) \leq  \exp\!\left( \frac{-N_2\varepsilon_2^2}{2\left \Vert \alpha \right \Vert_1^2}\right).
\end{equation}
This implies that for 
\begin{equation}
    N_2 \geq \frac{2\left \Vert \alpha \right \Vert_1^2 \ln(\sfrac{2}{\delta_2})}{\varepsilon_2^2},
\end{equation}
we have
\begin{equation}
    \Pr \! \left (\left |\overline{Y}^{(2)} - \mathbb{E}\!\left[ \overline{Y}^{(2)}\right] \right| \leq \varepsilon_2 \right) \geq 
    1 - \delta_2 ,\label{eq:Z(2)_bound}
\end{equation}
thus concluding the proof.
\end{proof}

\subsubsection{Sample complexity of QBM-GSE}

Using the development above, in the proof of the following theorem, we analyze the sample complexity of the QBM-GSE algorithm (Algorithm~\ref{algo:qbm-gse}).

\begin{theorem}[Sample Complexity of QBM-GSE] Let $H$ be a Hamiltonian as defined in~\eqref{eq:def_h}, and let $\alpha \in \mathbb{R}_{\geq 0}^K$ be the coefficients vector of $H$. Let $\varepsilon>0$  and $J \in \mathbb{N}$. Then the sample complexity, $N$, of the QBM-GSE algorithm (Algorithm~\ref{algo:qbm-gse}) to reach an $\varepsilon$-stationary point of the optimization problem~\eqref{eq:qbm_learn_def} is given by
\begin{equation}
    N = 2J\left\lceil\frac{12  \ell \Delta}{\varepsilon^2}\right\rceil \left \lceil \frac{8J\left \Vert \alpha \right \Vert_1^2 \ln(\sfrac{16J\left\Vert \alpha \right \Vert_1^2}{\varepsilon^2})}{\varepsilon^2} \right\rceil,
\end{equation}
where the smoothness parameter $\ell$ is defined in~\eqref{eq:smoothness_para}, $
\Delta  \coloneqq \operatorname{Tr}[H\rho(\theta_0)] - \inf_{\theta \in \mathbb{R}^J} \operatorname{Tr}[H\rho(\theta)]$, and $\theta_0 \in \mathbb{R}^J$ is a randomly chosen initial point.
\end{theorem}
\begin{proof}
    Note that QBM-GSE is an SGD algorithm, where the stochastic gradients $\overline{g}(\theta)$, at any given point $\theta$, are estimated using QBGE (Algorithm~\ref{algo:quant_Boltzmann_grad_estimator}):
\begin{align}
    \overline{g}(\theta) = \left( \overline{g}_{1}(\theta), \ldots,  \overline{g}_{J}(\theta)\right)^{\mathsf{T}},
\end{align}
where $\overline{g}_j(\theta)$ is the stochastic partial derivative given as
\begin{align}
    \overline{g}_j(\theta)=\overline{Y}_j^{(1)}(\theta) + \overline{Y}_j^{(2)}(\theta).\label{eq:stoch_par_deri}
\end{align}
Here, QBGE evaluates $\overline{Y}_j^{(1)}(\theta)$ and $\overline{Y}_j^{(2)}(\theta)$ using Algorithms~\ref{algo:est_first_term} and~\ref{algo:est_second_term}, respectively.  
    
From Section~\ref{sec:sgd}, we know that in order to use SGD for optimization, the stochastic gradient should be unbiased. This is true for our case; i.e., for all $ \theta \in \mathbb{R}^J$, we have $\mathbb{E}[\overline{g}(\theta)] = \nabla_{\theta}\operatorname{Tr}[H\rho(\theta)]$, and we showed this previously in Section~\ref{sec:qbge}. 

Another requirement for SGD is that the variance of the stochastic gradient should be bounded from above. Specifically, the stochastic gradient should satisfy the condition given by~\eqref{eq:norm_grad_exp} for some constants $A$, $B$, and $C$. Therefore, we now proceed to obtain these constants for our case. Consider that
   \begin{align}
    & \mathbb{E}\!\left[\left \Vert \overline{g}(\theta) \right\Vert^2\right] \notag \\
    & = \mathbb{E}\!\left[\left \Vert \overline{g}(\theta) - \nabla_\theta \operatorname{Tr}[H\rho(\theta)] \right\Vert^2\right] + \left \Vert \nabla_\theta \operatorname{Tr}[H\rho(\theta)] \right\Vert^2\\
    & = \mathbb{E}\!\left[ \sum_{j=1}^{J} \left( \overline{g}_j(\theta) - \partial_j \operatorname{Tr}[H\rho(\theta)]\right)^2 \right] + \left \Vert \nabla_\theta \operatorname{Tr}[H\rho(\theta)] \right\Vert^2\\
    & = \sum_{j=1}^{J} \mathbb{E}\!\left[  \left( \overline{g}_j(\theta) - \partial_j \operatorname{Tr}[H\rho(\theta)]\right)^2 \right] + \left \Vert \nabla_\theta \operatorname{Tr}[H\rho(\theta)] \right\Vert^2\\
    & = \sum_{j=1}^{J} \mathbb{E}\!\left[\left ( \overline{Y}^{(1)}_{j}(\theta) + \overline{Y}^{(2)}_{j}(\theta) - \left [ - \frac{1}{2}\operatorname{Tr}\!\left[  \left\{  H,\Phi_{\theta}(G_{j})\right\}
\rho(\theta)\right]  + \left\langle H\right\rangle _{\rho(\theta)}\left\langle
G_{j}\right\rangle _{\rho(\theta)}\right] \right)^2\right] \notag\\
& \qquad + \left \Vert \nabla_\theta \operatorname{Tr}[H\rho(\theta)] \right\Vert^2\\
& = \sum_{j=1}^{J} \operatorname{Var}\!\left[ \overline{Y}^{(1)}_{j}(\theta) + \overline{Y}^{(2)}_{j}(\theta)\right] + \left \Vert \nabla_\theta \operatorname{Tr}[H\rho(\theta)] \right\Vert^2\\
& = \sum_{j=1}^{J} \operatorname{Var}\!\left[ \overline{Y}^{(1)}_{j}(\theta)\right] + \operatorname{Var}\!\left[ \overline{Y}^{(2)}_{j}(\theta)\right] + \left \Vert \nabla_\theta \operatorname{Tr}[H\rho(\theta)] \right\Vert^2\\
& \leq \sum_{j=1}^{J} (\varepsilon_1^2 + \delta_1 \left\Vert \alpha \right \Vert_1^2) + (\varepsilon_2^2 + \delta_2 \left\Vert \alpha \right \Vert_1^2) +  \left \Vert \nabla_\theta \operatorname{Tr}[H\rho(\theta)] \right\Vert^2\\
& \leq J \left (\varepsilon_1^2 + \varepsilon_2^2 + \left(\delta_1 + \delta_2 \right) \left\Vert \alpha \right \Vert_1^2 \right) +  \left \Vert \nabla_\theta \operatorname{Tr}[H\rho(\theta)] \right\Vert^2.
\end{align}
The fourth equality follows from~\eqref{eq:stoch_par_deri} and~\eqref{eq:par_deri_exp}. The sixth equality follows from the fact that the variance of the sum of two independent random variables, $X$ and $Y$, is equal to the sum of their individual variances, i.e., $\operatorname{Var}[X+Y] = \operatorname{Var}[X] + \operatorname{Var}[Y]$. The first inequality follows directly from the variance bounds of the sample means $\overline{Y}^{(1)}_{j}(\theta)$ and $\overline{Y}^{(2)}_{j}(\theta)$. Indeed, these are a consequence of the following reasoning. Letting $Y \in [-C,C]$ be a random variable such that $\Pr(|Y - \mathbb{E}[Y] | \leq \varepsilon' ) \geq 1-\delta'$, for $\varepsilon' >0 $ and $\delta' \in (0,1)$, and defining the set $\mathcal{S} \coloneqq \{y : |y-\mathbb{E}[Y]| \leq \varepsilon'\}$, we find that
\begin{align}
    \operatorname{Var}\!\left[ Y\right] &  = \sum_{y} p(y) |y-\mathbb{E}[Y]|^2 \\
    & = \sum_{y\in \mathcal{S} } p(y) |y-\mathbb{E}[Y]|^2 + \sum_{y\in \mathcal{S}^c} p(y) |y-\mathbb{E}[Y]|^2 \\
    & \leq \sum_{y\in \mathcal{S}} p(y) \varepsilon'^2 + \sum_{y\in \mathcal{S}^c} p(y) C^2 \\
    & \leq  \varepsilon'^2 + \delta' C^2 .
\end{align}
Applying this inequality and the Hoeffding bounds in~\eqref{eq:Z(1)_bound} and~\eqref{eq:Z(2)_bound}, we conclude the first inequality.
Now comparing the second inequality with the condition given by~\eqref{eq:norm_grad_exp}, we obtain the constants for our case: $A=0, B=1$, and $C = J\!\left (\varepsilon_1^2 + \varepsilon_2^2 + \left(\delta_1 + \delta_2 \right) \left\Vert \alpha \right \Vert_1^2 \right) $.

Recall from the convergence result of SGD (Lemma~\ref{lem:sgd_conv}) that the total number of iterations, $M$, is bounded from below as follows:
\begin{equation}
    M \geq \frac{12  \ell \Delta}{\varepsilon^2}\max\!\left \{ B, \frac{12  A \Delta}{\varepsilon^2}, \frac{2 C}{\varepsilon^2}\right\}.\label{eq:M_iter}
\end{equation}
Now, if we choose the algorithm parameters $\varepsilon_1$, $\varepsilon_2$, $\delta_1$,  and $\delta_2$ such that the following holds:
\begin{equation}
    2C = 2J\!\left (\varepsilon_1^2 + \varepsilon_2^2 + \left(\delta_1 + \delta_2 \right) \left\Vert \alpha \right \Vert_1^2 \right) \leq \varepsilon^2,\label{eq:algo-cond}
\end{equation} 
then the bound in~\eqref{eq:M_iter} can be written as follows:
\begin{equation}
    M \geq \frac{12  \ell \Delta}{\varepsilon^2}.
\end{equation}
This resolves the minimum number of iterations needed by SGD to reach an $\varepsilon$-stationary point, given that~\eqref{eq:algo-cond} holds. 

Similarly, we can evaluate the step size $\eta$ for SGD, which we do in the following way. Again recall from Lemma~\ref{lem:sgd_conv} that the step size $\eta$ is given as follows:
\begin{equation}\label{eq:eta_exp}
    \eta = \min\left \{ \frac{1}{\sqrt{\ell A M}},  \frac{1}{\ell B},  \frac{\varepsilon}{2 \ell C}\right\}. 
\end{equation}
Now, if we have that $0<\varepsilon<1$, then this condition along with the condition given by~\eqref{eq:algo-cond} implies the following:
\begin{align}
    2C \leq \varepsilon.
\end{align} 
Using this inequality in~\eqref{eq:eta_exp}, we finally obtain:
\begin{equation}
    \eta = \frac{1}{\ell}.
\end{equation}

That being said, the question is how to choose the algorithm parameters $\varepsilon_1$, $\varepsilon_2$, $\delta_1$,  and $\delta_2$ such that the condition given by~\eqref{eq:algo-cond} holds. One way to do that is to choose $\varepsilon_1 = \varepsilon_2 = \sfrac{\varepsilon}{2\sqrt{2J}}$ and $\delta_1 = \delta_2 = \sfrac{\varepsilon^2}{8J\left\Vert \alpha \right \Vert_1^2}$. This now resolves the sample complexities $N_1$ and $N_2$ of Algorithms~\ref{algo:est_first_term} and~\ref{algo:est_second_term}, respectively:
\begin{equation}
    N_1 = N_2  = \left \lceil \frac{8J\left \Vert \alpha \right \Vert_1^2 \ln(\sfrac{16J\left\Vert \alpha \right \Vert_1^2}{\varepsilon^2})}{\varepsilon^2} \right\rceil.
\end{equation}
From this, we get the total sample complexity of the QBM-GSE algorithm: 
\begin{equation}
    N = M \cdot J(N_1 + N_2) = 2J\left\lceil\frac{12  \ell \Delta}{\varepsilon^2}\right\rceil \left \lceil \frac{8J\left \Vert \alpha \right \Vert_1^2 \ln(\sfrac{16J\left\Vert \alpha \right \Vert_1^2}{\varepsilon^2})}{\varepsilon^2} \right\rceil.
\end{equation}
This concludes the proof.
\end{proof}

\section{Addressing the open problem regarding QBM learning}

\label{sec:resolving-open-problem}

We begin by recalling the problem at hand, originally put forward in \cite{Amin2018}. Consider a probability distribution $P_{\bs{\opn{v}}}^{\opn{data}}$ defined by classical training data over a random variable $\bs{\opn{v}}$. The goal is to learn this distribution using a parameterized model $P_{\bs{\opn{v}}}(\theta)$, where $\theta$ represents a vector of parameters. More concretely, we aim to minimize the average negative log-likelihood $\mathcal{L}(\theta)$ between the target distribution $P_{\bs{\opn{v}}}^{\opn{data}}$ and the model distribution $P_{\bs{\opn{v}}}(\theta)$:
\begin{equation}
    \mathcal{L}(\theta) \coloneqq - \sum_{\bs{\opn{v}}} P_{\bs{\opn{v}}}^{\opn{data}} \log P_{\bs{\opn{v}}}(\theta)\label{eq:qbm-l-obj}.
\end{equation}

One way to realize the model distribution $P_{\bs{\opn{v}}}(\theta)$ is via a QBM with some parameterized Hamiltonian $G(\theta)$ as defined in~\eqref{eq:para-qbm-hamil}. More formally, we define: 
\begin{equation}
    P_{\bs{\opn{v}}}(\theta) \coloneqq \opn{Tr}[\Lambda_{\bs{\opn{v}}} \rho(\theta)],
\end{equation}
where $(\Lambda_{\bs{\opn{v}}})_{\bs{\opn{v}}}$ is an efficiently implementable measurement and $\rho(\theta)$ is the parameterized thermal state corresponding to $G(\theta)$ as defined in~\eqref{eq:para-state}. Using the definition above, we can rewrite $\mathcal{L}(\theta)$ as follows:
\begin{align}
    \mathcal{L}(\theta) \coloneqq - \sum_{\bs{\opn{v}}} P_{\bs{\opn{v}}}^{\opn{data}} \log \opn{Tr}[\Lambda_{\bs{\opn{v}}} \rho(\theta)].
\end{align}

The partial derivative of $\mathcal{L}(\theta)$ with respect to the parameter $\theta_j$ is given by
\begin{align}
    \partial_j \mathcal{L}(\theta) & = - \sum_{\bs{\opn{v}}} P_{\bs{\opn{v}}}^{\opn{data}} \left (\frac{\opn{Tr}[\Lambda_{\bs{\opn{v}}} \partial_j \rho(\theta)]}{\opn{Tr}[\Lambda_{\bs{\opn{v}}} \rho(\theta)]}\right) \\
    & = - \sum_{\bs{\opn{v}}} P_{\bs{\opn{v}}}^{\opn{data}} \left (\frac{\opn{Tr}[\Lambda_{\bs{\opn{v}}} \left(-\frac{1}{2}\left\{  \Phi_{\theta}(G_{j}),\rho
(\theta)\right\}  +\rho(\theta)\left\langle G_{j}\right\rangle\right)]}{\opn{Tr}[\Lambda_{\bs{\opn{v}}} \rho(\theta)]}\right)\\
& = \sum_{\bs{\opn{v}}} P_{\bs{\opn{v}}}^{\opn{data}} \frac{\frac{1}{2}\opn{Tr}[\Lambda_{\bs{\opn{v}}}\left\{  \Phi_{\theta}(G_{j}),\rho
(\theta)\right\}]}{\opn{Tr}[\Lambda_{\bs{\opn{v}}} \rho(\theta)]} - \left\langle G_{j}\right\rangle\\
& = \sum_{\bs{\opn{v}}} P_{\bs{\opn{v}}}^{\opn{data}} \frac{\frac{1}{2}\opn{Tr}[\left\{ \Lambda_{\bs{\opn{v}}},  \Phi_{\theta}(G_{j})\right\}\rho
(\theta)]}{\opn{Tr}[\Lambda_{\bs{\opn{v}}} \rho(\theta)]} - \left\langle G_{j}\right\rangle,
\end{align}
where the second equality follows from~\eqref{eq:derivative-thermal-state}. Previous works suggested that efficient gradient estimation was infeasible due to the apparent computational difficulty in evaluating the numerator of the first term. However, we demonstrate that this limitation can be overcome---Algorithm~\ref{algo:est_first_term} provides an efficient method to estimate this term to arbitrary precision, thereby enabling efficient gradient computation under mild assumptions that we specify in what follows.

We now analyze how errors in estimating both the numerator and denominator propagate to affect the final precision of their ratio. Let $p$ and $q$ be unbiased estimates of $\frac{1}{2}\opn{Tr}[\left\{ \Lambda_{\bs{\opn{v}}},  \Phi_{\theta}(G_{j})\right\}\rho
(\theta)]$ and $\opn{Tr}[\Lambda_{\bs{\opn{v}}} \rho(\theta)]$, respectively, such that the following holds for some $\varepsilon_1, \varepsilon_2 \geq  0$:
    \begin{align}
        \left|p - \frac{1}{2}\opn{Tr}[\left\{ \Lambda_{\bs{\opn{v}}},  \Phi_{\theta}(G_{j})\right\}\rho
(\theta)]\right| & \leq \varepsilon_1, \\
        \left|q - \opn{Tr}[\Lambda_{\bs{\opn{v}}} \rho(\theta)]\right| & \leq \varepsilon_2\label{eq:bound-2}. 
    \end{align}
We also assume that there exists a scalar $r > 0$, such that $\operatorname{Tr}\left[\Lambda_{\bs{\opn{v}}} \rho(\theta)\right] \geq r$. This assumption is justified because the thermal state $\rho(\theta)$ is generically full-rank for typical parameter values. Finally, we assume that each $G_j$ is a Pauli string, and that $\varepsilon_2 < r$ for guaranteeing numerical stability of the ratio estimator. As will become clear, the resulting error bound depends on the gap $(r - \varepsilon_2)$. Also, note that the assumptions $\varepsilon_2 < r$ and  $\operatorname{Tr}\left[\Lambda_{\bs{\opn{v}}} \rho(\theta)\right] \geq r$ imply that $\operatorname{Tr}\left[\Lambda_{\bs{\opn{v}}} \rho(\theta)\right] - \varepsilon_2 > 0$. 

\begin{proposition}
    The following inequality holds:
    \begin{equation}
        \left|\frac{p}{q} - \frac{\frac{1}{2}\operatorname{Tr}\left[\Phi_{\theta}(G_j)\{\Lambda_{\bs{\opn{v}}}, \rho(\theta)\}\right]}{\operatorname{Tr}\left[\Lambda_{\bs{\opn{v}}} \rho(\theta)\right]}\right| \leq \frac{1}{\left(r - \varepsilon_2\right)}\frac{\varepsilon_2}{r} + \frac{\varepsilon_1}{r}.
    \end{equation}
\end{proposition}
\begin{proof}
    Consider that
    \begin{align}
        & \left|\frac{p}{q} - \frac{\frac{1}{2}\operatorname{Tr}\left[\Phi_{\theta}(G_j)\{\Lambda_{\bs{\opn{v}}}, \rho(\theta)\}\right]}{\operatorname{Tr}\left[\Lambda_{\bs{\opn{v}}} \rho(\theta)\right]}\right|\nonumber\\
        & = \left|\frac{p}{q} - \frac{p}{\operatorname{Tr}\left[\Lambda_{\bs{\opn{v}}} \rho(\theta)\right]} + \frac{p}{\operatorname{Tr}\left[\Lambda_{\bs{\opn{v}}} \rho(\theta)\right]} - \frac{\frac{1}{2}\operatorname{Tr}\left[\Phi_{\theta}(G_j)\{\Lambda_{\bs{\opn{v}}}, \rho(\theta)\}\right]}{\operatorname{Tr}\left[\Lambda_{\bs{\opn{v}}} \rho(\theta)\right]}\right|\\
        & \leq \left|\frac{p}{q} - \frac{p}{\operatorname{Tr}\left[\Lambda_{\bs{\opn{v}}} \rho(\theta)\right]} \right| + \left|\frac{p}{\operatorname{Tr}\left[\Lambda_{\bs{\opn{v}}} \rho(\theta)\right]} - \frac{\frac{1}{2}\operatorname{Tr}\left[\Phi_{\theta}(G_j)\{\Lambda_{\bs{\opn{v}}}, \rho(\theta)\}\right]}{\operatorname{Tr}\left[\Lambda_{\bs{\opn{v}}} \rho(\theta)\right]}\right|\\
        & \leq |p|\left|\frac{1}{q} - \frac{1}{\operatorname{Tr}\left[\Lambda_{\bs{\opn{v}}} \rho(\theta)\right]} \right| + \frac{1}{\left|\operatorname{Tr}\left[\Lambda_{\bs{\opn{v}}} \rho(\theta)\right]\right|} \left|p - \frac{1}{2}\operatorname{Tr}\left[\Phi_{\theta}(G_j)\{\Lambda_{\bs{\opn{v}}}, \rho(\theta)\}\right]\right|.
    \end{align}
Using the facts that  $|\!\operatorname{Tr}\left[\Lambda_{\bs{\opn{v}}} \rho(\theta)\right]\!| \geq r$,  and $p \in [-1, 1]$, which holds because
\begin{align}
\left| \frac{1}{2}\opn{Tr}[\left\{ \Lambda_{\bs{\opn{v}}},  \Phi_{\theta}(G_{j})\right\}\rho
(\theta)]\right| & \leq \left \| \frac{1}{2}\left\{ \Lambda_{\bs{\opn{v}}},  \Phi_{\theta}(G_{j}) \right \} \right \| \left \| \rho(\theta)\right \|_1 \\
& \leq \left \|  \Lambda_{\bs{\opn{v}}} \right\| \left\|   \Phi_{\theta}(G_{j})  \right \|  \\
& \leq  \left\|   G_{j}  \right \|  \\
& \leq 1,    
\end{align}
and the circuit estimating $p$ never returns a value outside of the interval $[-1,1]$,
 we have that
\begin{align}
     & \left|\frac{p}{q} - \frac{\frac{1}{2}\operatorname{Tr}\left[\Phi_{\theta}(G_j)\{\Lambda_{\bs{\opn{v}}}, \rho(\theta)\}\right]}{\operatorname{Tr}\left[\Lambda_{\bs{\opn{v}}} \rho(\theta)\right]}\right|\nonumber\\
     & \leq \left|\frac{1}{q} - \frac{1}{\operatorname{Tr}\left[\Lambda_{\bs{\opn{v}}} \rho(\theta)\right]} \right| + \frac{1}{r} \left|p - \frac{1}{2}\operatorname{Tr}\left[\Phi_{\theta}(G_j)\{\Lambda_{\bs{\opn{v}}}, \rho(\theta)\}\right]\right|\\
      & \leq \left|\frac{\operatorname{Tr}\left[\Lambda_{\bs{\opn{v}}} \rho(\theta)\right] - q}{q \operatorname{Tr}\left[\Lambda_{\bs{\opn{v}}} \rho(\theta)\right]}\right| + \frac{\varepsilon_1}{r} \\
      & \leq \frac{1}{|q|}\frac{\varepsilon_2}{r} + \frac{\varepsilon_1}{r}.
\end{align}

Now from~\eqref{eq:bound-2}, we have the following:
\begin{equation}
    \frac{1}{\operatorname{Tr}\left[\Lambda_{\bs{\opn{v}}} \rho(\theta)\right] + \varepsilon_2} \leq \frac{1}{q} \leq \frac{1}{\operatorname{Tr}\left[\Lambda_{\bs{\opn{v}}} \rho(\theta)\right] - \varepsilon_2} \leq \frac{1}{r - \varepsilon_2}.
\end{equation}
Using the above inequality, it follows that
\begin{equation}
\left|\frac{p}{q} - \frac{\frac{1}{2}\operatorname{Tr}\left[\Phi_{\theta}(G_j)\{\Lambda_{\bs{\opn{v}}}, \rho(\theta)\}\right]}{\operatorname{Tr}\left[\Lambda_{\bs{\opn{v}}} \rho(\theta)\right]}\right|
 \leq \frac{1}{\left(r - \varepsilon_2\right)}\frac{\varepsilon_2}{r} + \frac{\varepsilon_1}{r}.
\end{equation}
This concludes the proof.
\end{proof}

\end{document}